\documentclass[11pt]{article}

\usepackage{graphicx} % Required for inserting images
\usepackage{amsthm,amsmath,amssymb}
\usepackage[disable]{todonotes}
\usepackage{geometry}
\usepackage{hyperref}
\usepackage{svg}

\newtheorem{theorem}{Theorem}[section]

\newtheorem{lemma}[theorem]{Lemma}

\newcounter{tbox}
\newcommand{\sta}[1]{\vspace*{0.3cm}\refstepcounter{tbox}\noindent{ \parbox{\textwidth}{(\thetbox) \emph{#1}}}\vspace*{0.3cm}}

\newcommand{\bw}{\textsf{bw}}
\newcommand{\pw}{\textsf{pw}}
\newcommand{\pack}{\textsf{pack}}

\title{Forbidden Subgraphs of Graphs with Low Bandwidth}

\author{Maria Chudnovsky\thanks{Princeton University, Princeton, NJ, USA. Supported by NSF Grants DMS-2348219 and CCF-2505100, AFOSR grant FA9550-25-1-0275, and a Guggenheim Fellowship.} \and
Daniel Lokshtanov\thanks{Department of Computer Science, University of California Santa Barbara, Santa Barbara, CA, USA. Supported by NSF Grant CCF-2505099.} \and
Eran Nevo\thanks{Hebrew University of Jerusalem, Jerusalem, Israel; and Universidad de Valladolid, Valladolid, Spain. Partially supported by the Israel Science Foundation grant ISF-687/24, and thanks Harvard CMSA for their hospitality and support.}
}

\date{}

\begin{document}

\maketitle

\begin{abstract}
A {\em layout} of a graph $G$ is an injective function $f : V(G) \rightarrow \mathbb{Z}$, and the {\em bandwidth} of a layout $f$ is $\bw(G,f) = \max_{uv \in E(G)} |f(u) - f(v)|$. The bandwidth $\bw(G)$ of $G$ is the minimum bandwidth of a layout of $G$. Computing the bandwidth of a graph is a notoriously hard problem: 
assuming $\textsf{P} \neq \textsf{NP}$, there is no polynomial time algorithm, even on very restricted classes of trees [Monien, SIAM Journal on Algebraic Discrete Methods, 1986], and no constant factor approximation, even on trees [Dubey et al., JCSS 2011]. Assuming the Exponential Time Hypothesis, there is no algorithm with running time $f(k)n^{o(k)}$ to determine whether an input graph has bandwidth at most $k$, even on very restricted classes of trees [Dregi and Lokshtanov, ICALP 2014].

In this paper we show that {\sc Bandwidth} on general graphs is FPT-approximable. In particular we give an algorithm that takes as input a graph $G$ and an integer $k$, runs in time $2^{O(9^k)} \cdot n^{O(1)}$, and outputs a subtree $T$ of $G$ such that $\bw(T) \geq k$ or a layout of $G$ of bandwidth at most $(10^{85} \cdot k^{28})^{4^k}$.
This resolves in the affirmative an open problem of Chung and Seymour [Discrete Mathematics, 1989], who asked whether the bandwidth of every graph $G$ is upper bounded in terms of the maximum bandwidth of a subtree of $G$. 
Our theorem leads to a forbidden subgraph characterization for graphs of bounded bandwidth, and can be seen as an analog for bandwidth of the classic grid minor theorem for treewidth, the forbidden subtree theorem for pathwidth, and the forbidden subpath theorem for treedepth. 
\end{abstract}

\section{Introduction}

%\todo[inline]{go over whole paper and check for same symbol meaning different things. Check whether some variables should be renamed.}

\todo[inline]{Check whole paper for too big jumps/insufficiently justified statements.}

A {\em layout} of a graph $G$ on $n$ vertices is an injective function $f$ from the vertex set $V(G)$ of $G$ to the integers. The {\em bandwidth} of a layout is the maximum over all edges $uv$ of $G$ of $|f(u)-f(v)|$ and the bandwidth of the graph $G$ is the minimum bandwidth over all layouts of $G$. In the {\sc Bandwidth} problem the input is a graph $G$ and an integer $k$ and the task is to determine whether the bandwidth of $G$ is at most $k$. In the optimization version of the problem the task is to compute a layout of the input graph $G$ of minimum bandwidth. 

%Bandwidth had received much attention during the fifties in order to speed up sev-eral computations on sparse matrices. Ac-cording to Dewdney [1976], the introduc-tion of the bandwidth problem for graphs(BANDWIDTH)  was  first  stated  in  Harary[1967], however the problem was formally defined in Harper [1966]. \todo{reword}

{\sc Bandwidth} was formally defined as a graph problem in the late sixties~\cite{harper1966optimal, harary1967}, however an equivalent formulation in terms of sparse matrices received considerable attention already in the fifties~\cite{chinn1982bandwidth}. 
%has received considerable attention since the fifties~\cite{}, because of its close ties to computations on sparse matrices. 
In particular one can reinterpret an arbitrary square matrix $A$ as the adjacency matrix of a weighted, directed graph $D$.
%\todo{cite? dont see a good citation for this.}
Then the (underlying undirected graph of) $D$ has bandwidth at most $k$ if and only if there exists a permutation matrix~$P$ such that $PAP^T$ is a matrix whose non-zero entries all lie in a band of width $k$ centered around the main diagonal, giving the {\sc Bandwidth} problem its name. Many matrix operations can be sped up substantially when the input matrix has bounded bandwidth~\cite{George1981,saad2003iterative}. 
%More modern tools for sparse matrix computations outperform bandwidth based schemes, nevertheless commercial packages still offer functionality to reduce the bandwidth of sparse matrices as a pre-processing  step~\cite{DiazPS02}.
Lai et al.~\cite{lai1999survey} survey a large number of known heuristics for computing the bandwidth of a graph, as well as upper and lower bounds on the bandwidth of general graphs, and sharper bounds for special cases.

{\sc Bandwidth} is a notoriously hard problem. It is one of the classic \textsf{NP}-complete~\cite{GareyJ1979,Papadimitriou76} problems, and it remains \textsf{NP}-complete even on caterpillars of hair length at most $3$~\cite{monien1986bandwidth}.
Here a {\em caterpillar} is a tree $T$ where all vertices of degree $3$ or more lie on a path $B$, called the {\em backbone} of $T$, and the {\em hair length} of a caterpillar is the maximum distance from $B$ to a vertex in $T$.
%a very restricted subclass of trees~\cite{monien1986bandwidth}. 
%
Indeed, the problem remains \textsf{NP}-hard to approximate within any constant factor, even on caterpillars, and a polynomial time approximation algorithm with ratio $c\sqrt{\log n/\log\log n}$ for a sufficiently small $c$ would imply that every problem in \textsf{NP} can be solved in quasi-polynomial time~\cite{dubey2011hardness}.
On the positive side, Dunagan and Vempala~\cite{DunaganV01} gave a polynomial time $O(\log^3 n \log\log n)$-approximation algorithm for {\sc Bandwidth} in general graphs, improving over a poly-logarithmic approximation algorithm of Feige~\cite{feige1998approximating,Feige00}. For special classes of graphs, such as trees~\cite{Gupta00} and  caterpillars~\cite{feige2009approximating} better (but still super-constant) approximation algorithms are known. 
Polynomial time algorithms for exactly computing the bandwidth are known on 
bipartite permutation graphs~\cite{HeggernesKM09},
interval graphs~\cite{kleitman1990computing},
cographs~\cite{yan1997bandwidth}
and caterpillars with hair length $1$ and $2$~\cite{assmann1981bandwidth}.

From the perspective of parameterized algorithms, Saxe~\cite{saxe1980dynamic} gave an algorithm to determine whether an input graph $G$ admits a layout of bandwidth $k$ in time $2^{O(k)}n^{k+1}$. The {\sc Bandwidth} problem was one of the first fundamental parameterized problems shown hard for the W-hierarchy. In particular Bodlaender et al.~\cite{BodlaenderFH94} showed that the algorithm of Saxe~\cite{saxe1980dynamic} cannot be improved to one with running time $f(k)n^{O(1)}$ unless FPT=W[t] for every integer $t$. Dregi and Lokshtanov~\cite{DregiL14} showed that, assuming the Exponential Time Hypothesis, there is no $f(k)n^{o(k)}$ time algorithm for {\sc Bandwidth} even on trees $T$ with a path $P$ such that every component of $T-P$ is a caterpillar. 

Exponential time algorithms for computing the bandwidth exactly and approximately have also been considered. Feige and Kilian gave an algorithm for exactly computing the bandwidth with running time $10^nn^{O(1)}$ and polynomial space~\cite{Feige00}. This was subsequently improved by Cygan and Pilipczuk~\cite{CyganP12b} to $O(4.83^n)$, and again by Cygan and Pilipczuk~\cite{cygan2010exact} to $O(4.39^n)$, at the cost of using exponential space. In a separate work, Cygan and Pilipczuk~\cite{CyganP12} also designed a $O(9.363^n)$ time algorithm for {\sc Bandwidth} that uses polynomial space. F\"{u}rer et al.~\cite{furer2013exponential} gave a factor $2$ approximation algorithm with running time $O(1.98^n)$, while Cygan and Pilipczuk~\cite{cygan2010exact} gave faster exponential time approximation algorithms at the cost of worse (but still constant) approximation ratios.

Given the hardness of {\sc Bandwidth} from the perspective of both polynomial time approximation algorithms and parameterized algorithms, it is interesting to study the problem from the perspective of parameterized approximation~\cite{feldmann2020survey,marx2008parameterized}. 
Here the goal is to get an algorithm that takes as input $G, k$, runs in time $f(k)n^{O(1)}$ and either concludes that $\bw(G) \geq k$ or produces a layout of bandwidth at most $g(k)$ for some function $g$ (where both $f$ and $g$ should ideally grow as slowly as possible with $k$). 
Prior to our work the only results in this direction were parameterized approximation algorithms for {\sc Bandwidth} on trees~\cite{DregiL14} and on asteroidal-triple-free graphs~\cite{GolovachHKLMS11}.
Indeed, {\sc Bandwidth} was one of the very first fundamental parameterized problems to be shown not to admit $f(k)n^{O(1)}$ time algorithms under plausible complexity assumptions~\cite{downey1999parameterized}. 
Recently we have seen breakthrough hardness of approximation results that lower bound the approximation ratio of parameterized approximation algorithms running in time $f(k)n^{O(1)}$ for problems such as {\sc Clique}~\cite{lin2021constant,karthik2022almost}, {\sc Dominating Set}~\cite{chen2019constant,karthik2018parameterized} and {\sc Constraint Satisfaction}~\cite{guruswami2024parameterized}. For each of these problems the message is essentially that no non-trivial approximation in time  $f(k)n^{O(1)}$ is possible. 
Since {\sc Bandwidth} is a ``harder'' problem from the perspective of parameterized complexity (it is hard for W[t] for every $t$, while the above problems are only hard for W[1] or W[2], see e.g.~\cite{downey1999parameterized}), 
the strong hardness approximation results for {\sc Clique}, {\sc Dominating Set} and {\sc Constraint Satisfaction} might lead one to believe  
%In light of these results one might believe 
that {\sc Bandwidth} also does not admit a meaningful parameterized approximation algorithm. 
Our main result runs counter to this intuition and demonstrates that the {\sc Bandwidth} problem does admit an $(10^{85} \cdot k^{28})^{4^k}$-approximation algorithm with running time $2^{O(9^k)} \cdot n^{O(1)}$.

\begin{theorem}\label{thm:main}
There exists an algorithm that takes as input a graph $G$ and an integer $k$, runs in time $2^{O(9^k)} \cdot n^{O(1)}$, and either concludes that $\bw(G) \geq k$ or produces a layout with bandwidth at most $(10^{85} \cdot k^{28})^{4^k}$.
When the algorithm concludes that $\bw(G) \geq k$ it outputs a subtree $T$ of $G$ with $\bw(T) \geq k$.
\end{theorem}

The second part of the statement of Theorem~\ref{thm:main} is a structural result which is interesting in its own right.
In particular, 
Chung and Seymour~\cite{chung1989graphs} posed as an open problem whether the bandwidth of a graph $G$ is upper bounded as a function of the maximum bandwidth of a subtree of $G$.

\smallskip
\noindent
{\bf Problem 1.}~\cite{chung1989graphs}{\em
~Does there exist a function $f : \mathbb{N} \rightarrow \mathbb{N}$ such that for every graph $G$ with $\bw(G) > f(k)$, $G$ contains a subtree of bandwidth at least $k$?}
\smallskip

%: namely that every graph $G$ of bandwidth at least $(10^{85} \cdot k^{28})^{4^k}$ must contain a subtree $T$ of bandwidth at least $k$.
%
%
%This answers in the affirmative an open problem of Chung and Seymour~\cite{chung1989graphs}, who posed as an open problem whether the bandwidth of a graph $G$ is upper bounded as a function of the maximum bandwidth of a subtree of $G$. \todo{write explicitly "problem 1: ..." and maybe put before theorem statement}

Theorem~\ref{thm:main} resolves Problem 1 with $f(k) = (10^{85} \cdot k^{28})^{4^k}$.
The study of forbidden structure theorems for bandwidth dates back to at least the early 1980s. 
A simple lower bound for the bandwidth of a graph $G$ is its {\em local density}. The local density $\hat{\Delta}(G)$ of a graph $G$ is the maximum of $\frac{|V(G')|-1}{\textsf{diam}(G')}$ over all subgraphs $G'$ of $G$, where the diameter $\textsf{diam}(G)$ is the maximum over all pairs of vertices in $G$ of their shortest path distance.
%
%A related notion is the {\em radial local density}, $\hat{\Delta}_r(G)$, which is the maximum over all subgraphs $G'$ of $G$ of $\frac{V(G')-1}{2\textsf{rad}(G')}$, where $\textsf{rad}(G)$ is the minimum over all vertices $u$ of the maximum over all vertices $v$ of the shortest path distance from $u$ to $v$.
%
It is not too hard to see that for every graph $G$, it holds that 
%$\bw(G) \geq \hat{\Delta}(G) \geq \hat{\Delta}_r(G) \geq \hat{\Delta}(G)/2$
$\bw(G) \geq \hat{\Delta}(G)$. On the other hand, Gupta et al.~\cite{gupta2003bounded} showed that $\bw(G) \leq O( \log^2(n) \cdot \hat{\Delta}(G)^{1.5})$.

Chv{\'a}talov{\'a}~\cite{chvatalova1982bandwidth} showed that there exist trees with constant local density and arbitrarily large bandwidth, that is, that the $\log^2(n)$ factor in the upper bound of Gupta et al.~\cite{gupta2003bounded} cannot be replaced by any function of $\hat{\Delta}(G)$.
The trees in Chv{\'a}talov{\'a}'s construction are subdivisions of complete binary trees. Here a {\em subdivision} of a graph $G$ is a graph $H$ obtained from $G$ by repeatedly replacing edges by degree $2$ vertices adjacent to the endpoints of the removed edge. 

The {\em pathwidth} of a layout $\lambda$ of $G$ is the maximum over all integers $i \leq n$ of the number of vertices $v \in V(G)$ such that $\lambda(v) \leq i$ and $v$ has at least one neighbor $u$ such that $\lambda(u) > i$. The pathwidth of a graph $G$, denoted by $\pw(G)$, is the minimum pathwidth of a layout of $G$. \todo{comment about non standard definition}
It is known that $\pw(G) \leq \bw(G)$ for every graph $G$, and that every graph of sufficiently large pathwidth contains a subdivision of a complete binary tree as a subgraph (see e.g.~\cite{robertson1983graph, bienstock1991quickly, cattell1996simple}). 
In light of this, Chung and Seymour~\cite{chung1989graphs} asked whether there exist trees with constant local density, constant pathwidth, but arbitrarily large bandwidth.
Chung and Seymour answered this question in the affirmative, giving a construction of such trees, which they call Cantor combs.

Subsequently Dregi and Lokshtanov~\cite{DregiL14, dregi2017beyond} showed that, for trees, these substructures are the only ones that force bandwidth to be large. In particular, they define the following generalization of Cantor Combs. 
A {\em skewed $b$-Cantor comb} of depth $1$ is the path on two vertices where one of the two vertices is its backbone. 
%Its backbone consists 
%
For $k > 1$, a caterpillar $T$ with backbone $B$ is a skewed $b$-Cantor comb of depth $k$ if there exists a subpath $Q$ of the backbone $B$ of $T$ such that $T-Q$ has precisely $3$ connected components, say $C_L$, $C_R$ and $P$, such that $C_L \cap B$ and $C_R \cap B$ are non-empty, $C_L$ with backbone $C_L \cap B$ is a skewed $b$-Cantor comb of depth $k-1$,  $C_R$ with backbone $C_R \cap B$ is a skewed $b$-Cantor comb of depth $k-1$, and  $P$ is disjoint from $B$, and $P$ is a path of length at least $2bd$, where $d$ is the maximum distance from the neighbor of $P$ on $Q$ to a vertex in $B$ (see Figure~\ref{fig:cantorCombs}). 
%\todo{figure? E: $v\in C_R \cup C_L$?? DL: {\em Changed, does this help?} A figure will clarify this too.}

\begin{figure}[h]
    \centering
    \includegraphics[width=0.9\textwidth]{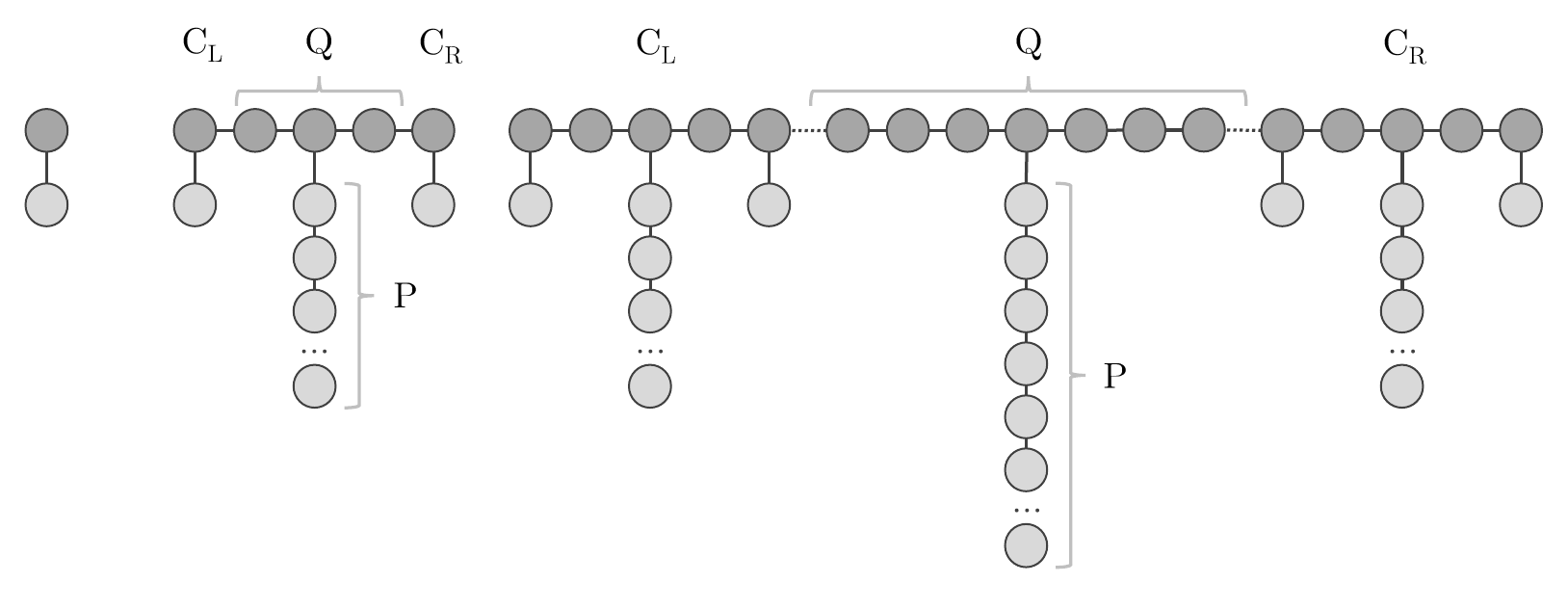}
    \caption{\em Example of skewed 3-Cantor combs of depth $1$, $2$, and $3$ respectively. The darkened path of each graph is its backbone.}
    \label{fig:cantorCombs}
\end{figure}

Dregi and Lokshtanov~\cite{DregiL14, dregi2017beyond} prove that for every $k$, skewed $k$-Cantor combs of depth $k$ have bandwidth at least $k$. As an extension of the parameterized approximation algorithm of Dregi and Lokshtanov~\cite{DregiL14} for the bandwidth of trees, Dregi shows the following. 

\begin{theorem}[\cite{dregi2017beyond}]\label{thm:treeBandwidth}
For every tree $T$, $\pw(T) \geq k+1$, or $\hat{\Delta}(T) \geq k+1$, or $T$ contains a skewed $(k+1)$-Cantor comb of depth $k+1$ as a subgraph, or $\bw(T) \leq (5k)^{6k}$.
\end{theorem}

%they show that every tree $T$ with local density $k_1$, pathwidth $k_2$ and no (generalized) Cantor comb of bandwidth at least $k_3$ as a subgraph has bandwidth at most $h(k_1, k_2, k_3)$ for a function $h$ polynomial in $k_1$ and $k_3$, and exponential in $k_2$.
%

Theorem~\ref{thm:main} together with Theorem~\ref{thm:treeBandwidth} immediately imply a forbidden subgraph characterization of bandwidth of general graphs (see Section~\ref{sec:proofsOfMain} for a proof). The proof yields a slightly tighter triple-exponential threshold; we round up to the compact form below.

\begin{theorem}\label{thm:mainInTermsOfObstructions}
For every graph $G$ if $\bw(G) > 4^{4^{(5k)^{7k}}}$ then
$G$ contains a subtree $T$ such that $\pw(T) \geq k+1$, or $\hat{\Delta}(T) \geq k+1$, or $T$ is a skewed $(k+1)$-Cantor comb of depth $k+1$.
\end{theorem}

%\begin{proof}
%Let $G$ be a graph such that $G$ does not contain a subtree $T$ such that $\pw(T) \geq k+1$, or $\hat{\Delta}(T) \geq k+1$, or $T$ is a skewed $(k+1)$-Cantor comb of depth $k+1$.
%
%By Theorem~\ref{thm:treeBandwidth} the bandwidth of every subtree of $G$ is at most $(5k)^{6k}$. Then, by Theorem~\ref{thm:main}, the bandwidth of $G$ is at most 
%$$(10^{84} \cdot 4^{11 {(5k)^{6k}}} \cdot {(5k)^{24k}})^{4^{(5k)^{6k}}} \leq (10^{84} \cdot 4^{{(8k)^{6k}}} \cdot {(5k)^{24k}})^{4^{(5k)^{6k}}} \leq 4^{4^{{5k}^{7k}}}$$
%\end{proof}

%result of Dregi and Lokshtanov gives a characterization of the substructures of a graph $G$ that force large bandwith. In particular every graph $G$ with no subtree of local density at least $k_1$, no subtree of pathwidth at least $k_2$ and no Cantor comb of bandwidth at least $k_3$ as a subgraph has bandwidth upper bounded as a function of $k_1$, $k_2$ and $k_3$. 

Forbidden subgraph theorems such as Theorem~\ref{thm:mainInTermsOfObstructions} for graph width parameters are central in algorithmic and structural graph theory. 
Theorem~\ref{thm:mainInTermsOfObstructions} is to bandwidth what the celebrated Grid Minor Theorem~\cite{robertson1986graph,chekuri2016polynomial,chuzhoy2015excluded,chuzhoy2021towards} is to treewidth (There exists a $c$ such that every graph of treewidth at least $ck^{9}(\log k)^c$ contains a grid of treewidth at least $k$ as a minor, or equivalently~\cite{kawarabayashi2018new}, a subdivision of a wall of treewidth $k$ as a subgraph),
%subgraph a subdivision of a wall of treewidth at least $k$, or equivalently~\cite{kawarabayashi2018new}, every graph of treewidth at least $\Omega(k^9)$ contains as a minor a grid of treewidth at least $k$), 
%
and the analogous results for 
%the forbidden forest~\cite{} theorem for 
pathwidth~\cite{robertson1983graph, bienstock1991quickly, cattell1996simple} (every graph of pathwidth more than $\frac{5 \cdot 3^{k-1} - 5}{2}$ contains a tree of pathwidth at least $k$ as a subgraph), 
and tree-depth~\cite{nesetril2014sparsity} 
%
%and the forbidden sub-path theorem~\cite{} for tree-depth 
(every graph of tree-depth at least $2^k$ contains as a subgraph a path of tree-depth at least $k$).

%\todo[inline]{E: wall? meaning a grid? {\em DL, no wall as a subgraph, grid as a minor. I wanted to stick to subgraphs. Will add a citation to what a wall is. Is it ok now? E:Yes.}}

\subsection{Proof Outline.}
In this overview we will focus solely on the structural aspect of the statement of Theorem~\ref{thm:main}, namely that every graph $G$ of huge bandwidth contains a subtree $T$ of large bandwidth. Every step of the proof (except one) is algorithmic, and so the parameterized approximation algorithm essentially follows from the proofs of the structural statements. 

Rather than working with layouts, it is more convenient to work with a more general object that we call {\em embeddings}. An embedding of a graph $G$ into a graph $H$ is simply a function $f : V(G) \rightarrow V(H)$. We define the {\em stretch} of the embedding $f$ to be the maximum over all edges $uv \in E(G)$ of the shortest path distance between $f(u)$ and $f(v)$ in $H$.
The {\em congestion} of an embedding $f$ is the maximum taken over all vertices $v$ in $H$ of the number of vertices $u$ in $G$ such that $f(u) = v$. We remark that Bienstock~\cite{bienstock1990embedding} also studied a combinatorial object called embeddings, this object is closely related to, but not quite the same as, embeddings as considered in this paper. The embeddings of Bienstock~\cite{bienstock1990embedding} need to be injective, and their congestion is measured differently than in this work.

It is not too difficult to see that the bandwidth of a graph $G$ is precisely the minimum stretch of an embedding of $G$ into a path with congestion $1$. While not directly relevant to our work, embeddings also generalize the tree-partition-width~\cite{ding1995some,wood2009tree} of a graph $G$. The tree-partition-width of $G$ is precisely the minimum congestion of a stretch at most $1$ embedding of $G$ into a tree. 

We have seen that embeddings into a path generalize graph layouts. In fact they are essentially equivalent: an embedding of $G$ into a path $P$ with stretch $s$ and congestion $c$ can be ``stretched out'', by arbitrarily breaking ties, into an embedding of $G$ into a path $Q$, which is $c$ times longer than $P$, such that the resulting embedding of $G$ into $Q$ has congestion $1$ and stretch at most $c(s+1) - 1$. Since we are only concerned with graphs whose bandwidth is a constant, it suffices to consider embeddings whose stretch and congestion is also constant, without worrying too much about what precisely that constant is. 
Throughout this overview we will refer to quantities that are upper bounded by a function of $k$ with terms such as ``small'' or ``low''. 

Embeddings exhibit a compositional property: if $h_1$ is an embedding of $G$ into $H_1$ and $h_2$ is an embedding of $H_1$ into $H_2$, then the function $h : V(G) \rightarrow V(H_2)$ defined as $h(u) = h_2(h_1(u))$ is an embedding of $G$ into $H_2$. Furthermore the stretch of $h$ is at most the stretch of $h_1$ times the stretch of $h_2$, and the congestion of $h$ is at most the congestion of $h_1$ times the congestion of $h_2$. 

Composition of embeddings yields a simple proof of the following fact: if $S$ is a small set of vertices of bounded degree, and $G - S$ has small bandwidth, then $G$ also has small bandwidth. Consider a low stretch and congestion embedding of $G - S$ into a path $P$. Then make a low stretch and congestion embedding of $P$ into another path $Q$ where we ``fold over" $P$ such that every vertex of $P$ that contains (the image of) a vertex in $N(S)$ gets mapped to the second vertex of $Q$. 
Since $S$ is a small set of small degree $N(S)$ is also small, and we only need to ``fold over" the path $P$ $N(S)$ many times. 
Composing the two embeddings yields a low stretch and congestion embedding of $G-S$ where all vertices of $N(S)$ get mapped to the second vertex of $Q$. We can now extend this to a low stretch and congestion embedding of $G$ into $Q$ by mapping all vertices of $S$ to the first vertex of $Q$.

Composition also leads to the following natural line of attack toward proving Theorem~\ref{thm:main}: suppose we can embed $G$ into a ``simpler'' graph $H$ (for some appropriate notion of simpler that ensures that we will only simplify few times). If $H$ has low bandwidth then $G$ does too. But what if $H$ has high bandwidth? We can avoid this problem if $H$ is a subgraph of $G$. In this case we know that the bandwidth of $H$ is at most the bandwidth of $G$. 
Our proof gives the ``easiest'' instantiation of this idea, in the sense that we will find a subtree $T$ of $G$ and an embedding of $G$ into $T$ with low stretch and congestion. If $T$ has high bandwidth we win, since we just found a subtree with high bandwidth. If $T$ has low bandwidth we also win, because then $G$ has low bandwidth by the composition of embeddings. 
    
Not every graph $G$ can be embedded into a subtree with constant stretch and congestion, for example a grid cannot (we do not prove this, since we do not need it). Thus, in our search for a subtree $T$ and an embedding of $G$ into $T$ we will also accept an alternative outcome of directly finding a subtree $T$ of $G$ of large bandwidth. For example, if $G$ has large local density (that is, $\hat{\Delta}(G) \geq 2k$) then there is a breadth first search tree of $G$ that has local density at least $k$, and therefore bandwidth at least $k$. If $G$ has large pathwidth (that is, $\pw(G) > \frac{5 \cdot 3^{k-1} - 5}{2}$) then $G$ has a subtree of pathwidth at least $k$, and hence of large bandwidth. Thus, throughout the proof we will assume that $G$ has small local density and small pathwidth. 

Our proof proceeds by induction on the maximum pathwidth $\tau$ of a subtree of $G$. It is known that the pathwidth of $G$ is functionally equivalent (see Theorem~\ref{thm:pwObstructions}) to the maximum pathwidth of a subtree of $G$, so one would expect that this is just an unnecessarily complicated way of doing induction on the pathwidth of $G$. Indeed, the induction can be carried out on either parameter, but each induction step in our proof incurs a cubic blowup in the bandwidth bound, so the number of induction steps appears in a doubly-exponential position in the final bound. Inducting on the maximum subtree pathwidth, which is at most $k$, keeps this exponent at $k$, whereas inducting on $\pw(G)$ would inflate it to a much larger function of $k$ (see Theorem~\ref{thm:pwObstructions} for the precise relationship).

Working with subtree pathwidth also lets us exploit a nice characterization, due to Scheffler~\cite{scheffler1990linear}, of the pathwidth of trees: a tree $T$ has pathwidth at least $k$ if and only if $T$ contains a vertex $v$ such that $T-v$ has at least three connected components of pathwidth at least $k-1$.

%\todo[inline]{DL: After thinking about this I think we should actually be able to do the proof using pathwidth of $G$, and do the packing argument on subgraphs of maximum pathwidth instead of subtrees? I believe this should make the whole bound single exponential instead of double exponential, but would make it slightly harder to find the packing algorithmically ... but now i think we should be able to do this, and get explicit dependence on $k$ in the running time and only single exponential dependence... Resolved after several rewrites: doing the argument on max pathwidth subtrees makes the argument cleaner but the bound quite a bit worse, even with the more clever bound that pw <= poly(local density, pathwidth of subtrees)}

Let $\tau$ be the maximum pathwidth of a subtree of $G$. Consider a maximum packing of subtrees of pathwidth $\tau$, that is a maximum cardinality set $\{T_1, T_2, \ldots, T_\rho\}$ of subtrees of $G$, such that $V(T_i) \cap V(T_j) = \emptyset$ for every pair of distinct integers $i$, $j \leq \rho$, and each tree $T_i$ has pathwidth precisely $\tau$. To find $\tau$ and a packing, we combine the algorithm of Hicks~\cite{Hicks04} for minor containment, the bound on the size of minor-minimal trees of a given pathwidth~\cite{cattell1996simple}, and an algorithm for computing path decompositions~\cite{furer2016faster,BodlaenderJT23}.

%\todo[inline]{using local density not pathwidth}
If the size $\rho$ of the packing is small we can use that $G$ has low local density to prove an Erd\"{o}s-P\'{o}sa-like (see~\cite{Erdos_Posa_1965,raymond2017recent}) result: in this case there exists a small vertex set $S$ such that $G - S$ has no subtree of pathwidth $\tau$. Then, by induction on $\tau$, $G-S$ either has a subtree of large bandwidth or $G-S$ has low bandwidth. We already saw that if $G-S$ has low bandwidth for a small set of low degree vertices (and all vertices in $G$ must have low degree, because local density is small), then $G$ has low bandwidth as well. Thus, in the case when the number $\rho$ of trees in the packing is small we are done. 

We now consider the case when the number $\rho$ of trees in the packing is large. In this case we first do a simple pre-processing step: whenever some vertex $v$ has a neighbor in a tree $T_i$, but $v$ is not in any other tree $T_j$ for $j \neq i$ we add $v$ to $T_i$. This ensures that the union of the vertex sets of the trees $T_1, \ldots, T_\ell$ covers all the vertices of the graph.

Let $H$ be the graph obtained from $G$ by contracting each tree in the packing to a single vertex. 
If $H$ has a vertex of degree at least $3$, then  that there are $4$ trees in the packing, say $T_1$, $T_2$, $T_3$ and $T_4$ such that $T_1$ has a neighbor in each of $T_2$, $T_3$ and $T_4$. Make a subtree $T$ of $G$ by taking the four trees $T_1$, $T_2$, $T_3$ and $T_4$ and adding, for each $T_i \in \{T_2, T_3, T_4\}$, an edge from a vertex in $T_1$ to a neighbor of that vertex in $T_i$.
Let $v$ be the unique vertex in $T_1$ which lies on the $T_2$-$T_3$ path, the $T_2$-$T_4$ path and the $T_3$-$T_4$ path in $T$, and apply Scheffler's characterization~\cite{scheffler1990linear} of pathwidth of trees to $v$.
Since each of $T_2$, $T_3$ and $T_4$ has pathwidth at least $\tau$ it follows that $T$ has pathwidth at least $\tau+1$ contradicting that $\tau$ was the maximum pathwidth of a subtree of $G$.
We conclude that the maximum degree of $H$ is $2$, in other words $H$ is either a path or a cycle.
The cycle case (essentially) reduces to the path case by removing only a few vertices, 
%and removing a small set of small degree vertices can't substantially decrease bandwidth,
so we might as well consider only the path case. 

%\todo[inline]{DL: we should probably update how we handle the cycle case by removing edges instead of vertices, because then we can directly reduce to the path case, making the definition of weak pre-layouts / pre-layouts a little bit cleaner. Thought through and resolved; not worth the pain.}

%\smallskip
Suppose now that $H$ is a path. We re-order the trees so that each tree $T_i$ only has neighbors in $T_{i-1}$ and $T_{i+1}$. 
The same Erd\"{o}s-P\'{o}sa style argument that we used for the case when $\rho$ is small can be used to show that, for every $i$ and small number $t$, the graph $G[V(T_i) \cup V(T_{i+1}) \ldots \cup V(T_{i+t-1})]$ induced by the vertex sets of $t$ consecutive subtrees in the packing, contains a small vertex set that hits every subtree of pathwidth at least $\tau$.
Additionally, since the local density of $G$ is small, there is an abundance of small separators everywhere in the graph. 
We use these two properties to prove the following decomposition theorem for $G$.
There exists a sequence $V_1, V_2, \ldots, V_\ell$ of vertex sets of $G$ such that:
\begin{enumerate}\setlength{\itemsep}{-.5pt}
    \item  Only consecutive pieces can have non-empty intersection: if $V_i$ and $V_j$ have non-empty intersection then $|i-j| \leq 1$.
        
    \item Each piece $V_i$ induces a connected subgraph of $G$. In fact a slightly stronger property holds: $G[V_i] - (V_{i-1} \cup V_{i+1})$ is connected. 
        
    \item Every vertex of $G$ appears in at least one and in at most two parts $V_i$. 
    
    \item Every edge of $G$ has both endpoints in at least one part $V_i$.
    
    \item Every boundary is small: for every part $V_i$, $|V_i \cap V_{i+1}|$ is small. This implies that for every $i$, only few vertices in $V_i$ can have neighbors outside of $V_i$. 
    
    \item  Each piece $G[V_i]$ is structurally simpler than $G$: in each piece $G[V_i]$ there is a small size set $Z_i$ such that $G[V_i]-Z_i$ does not have any subtree of pathwidth $\tau$. In particular we can apply induction on each $G[V_i]-Z_i$ and obtain for each piece $G[V_i]$ either a subtree of large bandwidth or a layout of small bandwidth. 
\end{enumerate}
For a reader familiar with path decompositions the above just states that we can find a path decomposition of $G$ such that adhesions are small, and every bag $V_i$ induces a {\em connected} subgraph $G[V_i]$ which is strictly simpler than $G$. 
  
For each part $G[V_i]$ of the decomposition we obtain either a subtree of large bandwidth, or a layout of $G[V_i]-Z_i$ of small bandwidth. Since each $Z_i$ is small, in the latter case we obtain a small bandwidth layout of the part $G[V_i]$. If either one of the parts has a subtree of large bandwidth, then this is a subtree of $G$ of large bandwidth and we are done. Thus we only need to handle the case when all the pieces $G[V_i]$ have small bandwidth. In this case we will produce a spanning tree $T$ of $G$ and a low stretch and congestion embedding of $G$ into $T$. 

When picking the spanning tree $T$ we crucially exploit the connectivity property of the decomposition. In particular we make the spanning tree $T$ in such a way that for every part $V_i$ of the decomposition, the smallest subtree $T_i$ of $T$ that contains $V_i$ is contained in $V_{i-1} \cup V_i \cup V_{i+1}$ (these trees $T_1, \ldots, T_\ell$ are {\em not} the same as the trees we used to make the decomposition $V_1, \ldots, V_\ell$).
Then we use the fact that $G[V_i]$ has bounded bandwidth to find a low stretch and low congestion embedding $f_i$ of $G[V_i]$ into $T_i$. 
We also ensure that the embeddings of different parts are consistent: for every $i$ and every vertex $v$ in $V_{i} \cap V_{i+1}$, $f_i(v) = f_{i+1}(v)$.
%the embeddings of $G[V_i]$ into $T_i$ and $G[V_{i+1}]$ into $T_{i+1}$ map $v$ to the same vertex. 
Thus we obtain an embedding $f$ of $G$ into $T$. The stretch of the embedding $f$ is just the max of the stretch of the individual embeddings $f_i$ of $G[V_i]$ into $T_i$. The congestion of $f$ is larger than the congestion of the individual embeddings $f_i$, but not by much. 
Since each embedding $f_i$ maps $G_i$ to $T_i$, whose vertex set is a subset of $V_{i-1} \cup V_i \cup V_{i+1}$, each vertex of $G$ is in the range of at most three embeddings $f_i$. Hence the congestion of $f$ is at most three times larger than the maximum congestion of an $f_i$.

In very broad strokes we are done - we either found a subtree of $G$ of large bandwidth or we found an embedding of $G$ with low stretch and low congestion into a spanning tree of $G$. However, several questions remain. Most prominently, why does each part $G[V_i]$ having low bandwidth imply that the part can be embedded with low stretch and congestion into $T_i$? And why does ensuring that the embeddings of consecutive parts are consistent not force stretch or congestion to go up by too much? We briefly sketch how we address each of these issues. 

Let us start with the easiest one of the two questions above: why does each part $G[V_i]$ having low bandwidth imply that the part can be embedded with low stretch and congestion into the tree $T_i$?
Here we exploit the compositionality of embeddings. Since $G[V_i]$ embeds nicely into a path $P$ on $|V_i|$ vertices, it is sufficient to find an embedding of $P$ into $T_i$ with low stretch and congestion.
It turns out that a path $P$ can be embedded into {\em any} connected graph $G'$ on at least $|V(P)|/2$ vertices with stretch and congestion at most $2$, by following a depth first search (DFS) traversal of $G'$ and writing down vertices of $G'$ when they are first visited by the DFS and when the recursive call corresponding to the vertex terminates. 
%every time the DFS finishes the recursive call corresponding to the vertex. 
This gives a sequence $v_1, v_2, \ldots,$ of vertices of $G'$ in which every vertex of $G'$ appears twice and consecutive vertices in the sequence are either adjacent or have a common parent in the DFS tree. We obtain the desired embedding of $P$ into $G'$ by mapping the $i$'th vertex of $P$ to the $i$'th vertex in the sequence. 

\smallskip
We now know that each $G_i$ can be embedded into $T_i$ with low stretch and congestion. However we still do not know how to coordinate the embeddings in such a way that we get a single well-defined embedding of $G$ into $T$.
Let $S_i = V_i \cap (V_{i-1} \cup V_{i+1})$. By the properties of our decomposition, $S_i$ is a small set. In order to coordinate the embedding $f_i$ of $V_i$ with the other embeddings, it is sufficient to consider the restriction of $f_i$ to $S_i$, that is $f_i(v)$ for every $v \in S_i$. 
We therefore investigate, given two connected graphs $G$ and $R$ such that $G$ has small bandwidth, a small vertex set $S \subseteq V(G)$ and a function $r : S \rightarrow V(R)$, what are the necessary and sufficient conditions for there to exist a low stretch and congestion embedding of $G$ into $R$ which agrees with $r$.

A few necessary conditions are quite immediate. First, $|V(R)|$ cannot
%\todo{look for can't} 
be much smaller than $|V(G)|$, otherwise congestion will be large. 
Second, every pair of vertices in $S$ that are close together in $G$ cannot be mapped far apart in $R$, otherwise stretch will be large. 
Surprisingly, these two necessary conditions are sufficient! In particular, for graphs $G$, $R$, a vertex set $S \subseteq V(G)$ and a function $r : S \rightarrow V(R)$, the $S$-{\em stretch} of $r$ is defined as $\max_{u,v \in S} \frac{d_R(r(u), r(v))}{d_G(u,v)}$.
We prove the following theorem, which is of independent interest. 

\begin{theorem}\label{lem:lowStretchToEmbedding}
There exists a polynomial time algorithm that takes as input
a connected graph $G$, positive integers $k$, $s$, $r$
such that $r \geq 2$,
a layout $f$ of $G$ of bandwidth at most $k$,
a non-empty vertex set $S$ in $G$,
a connected graph $H$ such that $|V(H)| \geq |V(G)|/r$
and a function $h : S \rightarrow V(H)$ with $S$-stretch at most $s$,
and outputs an embedding $g$ of $G$ into $H$ with
stretch at most $480(|S|+1)^3(k+1)^2s$
and congestion at most $1300(|S|+1)^4(k+1)r$
such that for every $s \in S$, $g(s) = h(s)$.
\end{theorem}

%The proof of this fact is not very long, but still quite non-trivial. 
A key step towards the proof of Theorem~\ref{lem:lowStretchToEmbedding} is to show that, for every graph $G$ of low bandwidth and small set $S$ of vertices in $G$, there exists an embedding of $G$ into a path $P$ with low stretch, low congestion, and also low $S$-{\em contraction}. Here the $S$-contraction of an embedding $f$ of a graph $G$ into a graph $R$ is defined as $\max_{u,v \in S} \frac{d_G(u,v)}{d_R(f(u), f(v))}$.
%is the maximum ratio between the distance in $G$ between two vertices in $S$ and the distance between their images in $P$. 
In other words $S$-contraction is large when two vertices in $S$ that are far apart in $G$ get mapped closed together in $P$. 
Embeddings of metrics into other metrics with low stretch and contraction are called low {\em distortion} embeddings and have been extensively studied (see e.g.~\cite{indyk2001algorithmic,Lin02} and references within). A crucial ingredient in our proof is a classic result of Matou{\v{s}}ek~\cite{matouvsek1990bi} that every metric space with $n$ points can be embedded into $\mathbb{R}$ with contraction $1$ and stretch at most $12n$.

Armed with Theorem~\ref{lem:lowStretchToEmbedding}, all that remains to do in order to complete the proof of Theorem~\ref{thm:main} is to select for each part $V_i$ of the decomposition a mapping $h_i$ of $S_i$ to $T_i$ with low $S_i$-stretch, such that for every $i$, $h_i$ and $h_{i+1}$ agree on $S_{i} \cap S_{i+1}$. 
This final step also requires some care. Indeed it requires one additional technical property of the main decomposition, and crucially relies on $G$ having low local density, as well as the parts of the decomposition being connected. 

\paragraph{Comparison with the Conference Version.}
This paper is an updated version of the conference paper~\cite{ChudnovskyLN26}. 
The two main differences between this paper and~\cite{ChudnovskyLN26} are both in the proof of the main decomposition lemma, namely Lemma~\ref{lem:weakPreLayoutDensity}. 
The original version~\cite{ChudnovskyLN26} invoked Courcelle's Theorem, leading to a non-elementary dependence on $k$ in the running time of the algorithm. 
We instead rely on a greedy improvement scheme, which leads to a double exponential running time dependence on $k$.
Additionally, in Section~\ref{sec:weakPreLayouts} we use local density, as opposed to pathwidth, as the engine behind the 
Erd\"{o}s-P\'{o}sa style arguments which are used for the proof of Lemma~\ref{lem:weakPreLayoutDensity}. This allows us to obtain a slightly better approximation ratio than that of~\cite{ChudnovskyLN26}.

\paragraph{Overview of the Paper.} In Section~\ref{sec:prelims} we set up notation, define the notions used throughout the paper, and collect some useful results from the literature.
In Section~\ref{sec:embeddings} we develop a toolbox for working with embeddings of small stretch and small congestion. 
In Section~\ref{sec:prescribed} we prove our first main result, namely Theorem~\ref{lem:lowStretchToEmbedding}.
In Section~\ref{sec:treepackings} we study the structure of graphs that do not contain a subtree of pathwidth at least $\tau+1$, but do contain subtrees of pathwidth $\tau$. In particular we show that if we pack a maximum number of vertex disjoint subtrees of pathwidth $\tau$ in $G$ then these subtrees must be arranged in the graph as either a path or a cycle. 
In Section~\ref{sec:weakPreLayouts} we show how to use the arrangement of pathwidth-$\tau$ subtrees in $G$ (where $\tau \leq k$) to obtain a decomposition of $G$ by small separators into pieces that are structurally simpler than the original graph. Specifically the pieces of the decomposition are simple enough that we may call the algorithm recursively on each of the pieces and assume by induction that, on each of the pieces the algorithm either returns a subtree of bandwidth at least $k$ or a layout of low bandwidth. 
If either of the recursive calls return a subtree bandwidth $k$ then we are done. Thus, in Section~\ref{sec:intoSubtree} we consider the case when all of the recursive calls on the pieces of the decomposition returned low bandwidth embeddings of those pieces. In this outcome we find a special subtree $T$ of $G$ and then use the results of Section~\ref{sec:prescribed} to find a low stretch and low congestion embedding of $G$ into $T$. Hence, if $T$ has bandwidth at least $k$ we win because we found a subtree of large bandwidth. On the other hand, if $T$ has low bandwidth then composing the embedding of $G$ into $T$ with the low bandwidth layout of $T$ yields a low bandwidth layout of $G$.
In Section~\ref{sec:proofsOfMain} we put all of the results in previous sections together and complete the proofs of Theorems~\ref{thm:main} and~\ref{thm:mainInTermsOfObstructions}.
Finally, in Section~\ref{sec:conclusion} we conclude with a few remarks about potential ways to improve our results.

\section{Notation, Definitions, and Preliminary Results}\label{sec:prelims}
%\todo[inline]{minors, minor model, Menger's Theorem, vertex neighbors, vertex set neighbors}
We will mostly follow the notation and terminology for graphs as defined in the textbook of Diestel~\cite{diestelBook}. All graphs are simple and undirected unless specified otherwise. The vertex set and edge set of a graph $G$ are denoted by $V(G)$ and $E(G)$ respectively. When clear from context we will denote by $n$ the number of vertices in the graph $G$. 
For a vertex $v$ its open neighborhood is defined as $N(v) = \{u : uv \in E(G)\}$ and its closed neighborhood is $N[v] = N(v) \cup \{v\}$. For a set $S$ of vertices $N[S] = \bigcup_{v \in S} N[v]$ and $N(S) = N[S] \setminus S$.
A {\em subgraph} of $G$ is a graph $H$ with $V(H) \subseteq V(G)$ and $E(H) \subseteq E(G)$; equivalently, $H$ can be obtained from $G$ by vertex and edge deletions. We write $H \subseteq G$ to denote that $H$ is a subgraph of $G$. 
The subgraph of $G$ {\em induced by} the vertex set $S$ is denoted by $G[S]$ and defined as $G[S] = (S, \{uv \in E(G) ~:~ \{u,v\} \subseteq S\})$. 
For a vertex set $S \subseteq V(G)$ we write $G - S$ for the subgraph $G[V(G) \setminus S]$; for a single vertex $v$, $G - v = G - \{v\}$; and for an edge $e$, $G - e$ denotes the subgraph obtained by deleting $e$ (but no vertices). A {\em subtree} of $G$ is a subgraph of $G$ that is a tree (not necessarily induced).
A {\em minor model} of a graph $H$ in a graph $G$ is a family $\{X_v ~:~ v \in V(H)\}$  of connected vertex sets in $G$ such that for every pair $u$, $v$ of vertices in $H$, $X_u$ and $X_v$ are disjoint, and for every edge $uv \in E(H)$ there exists an edge between $X_u$ and $X_v$ in $G$. A graph $H$ is a {\em minor} of $G$ if there exists a minor model of $H$ in $G$.

For a path $P$ the {\em length} of the path is the number of edges on $P$. For a graph $G$ and two vertices $u$ and $v$ the {\em distance} in $G$ from $u$ to $v$ is denoted as $d_G(u, v)$ and defined as the length of a shortest (minimum length) path from $u$ to $v$. If no path between $u$ and $v$ exists the distance from $u$ to $v$ is infinite. 
The {\em diameter} of a graph $G$ is $\textsf{diam}(G) = \max_{u,v \in V(G)} d_G(u,v)$ and the {\em radius} of $G$ is $\textsf{rad}(G) = \min_{u \in V(G)} \max_{v \in V(G)} d_G(u,v)$.
$P_{\mathbb{Z}}$ is the infinite path with vertex set $\mathbb{Z}$ and edge set $\{(i, i+1) ~:~ i \in \mathbb{Z}\}$. It is well known~\cite{diestelBook} that the diameter of a graph is at least its radius, and at most twice the radius. 
$P_n$ is the path with vertex set $\{1, \ldots, n\}$ and edge set $\{(i, i+1) ~:~ 1 \leq i < n\}$.

Let $D$ and $R$ be sets and $f : D \rightarrow R$ be a function. We denote by $2^D$ the power set $\{D' ~:~ D' \subseteq D\}$ of $D$. We define the function $f^{-1} : R \rightarrow 2^D$ as $f^{-1}(r) = \{d \in D ~:~ f(d) = r\}$. For sets $S \subseteq D$ we will write $f(S)$ for $\bigcup_{x \in S} \{f(x)\}$ and for sets $S \subseteq R$ we will write $f^{-1}(S)$ for $\bigcup_{x \in S} \{f^{-1}(x)\}$.

A vertex set $S$ separates two vertices $a$ and $b$ if every path from $a$ to $b$ in $G$ passes through $S$. The vertex set $S$ separates vertex sets $A$ from $B$ if it separates every $a \in A$ from every $b \in B$.
We will use the well known Menger's Theorem (see e.g.~\cite{diestelBook}): for every graph $G$ and vertex sets $A$ and $B$ the minimum size of a vertex set $S$ that separates $A$ from $B$ is equal to the maximum number of pairwise vertex-disjoint paths from $A$ to $B$ (sharing no vertex; a single vertex in $A \cap B$ counts as a path).
The {\em internal} vertices of a path $P$ are all its vertices except for its endpoints. Two paths $P$ and $Q$ are internally vertex disjoint if no vertex is an internal vertex of both $P$ and $Q$. A well known variant of Menger's Theorem is that for every graph and vertex sets $A$ and $B$ the minimum size of a vertex set  $S$ that separates $A$ from $B$ and is disjoint from $A \cup B$ is equal to the maximum number of internally vertex disjoint paths from $A$ to $B$.

A caterpillar is a tree $T$ containing a path $B$ such that all vertices of degree $3$ or more lie on $B$. We then say that $B$ is a backbone of $T$ and every connected component of $T - B$ is a stray or a hair. 
We will use the following well-known bound on the number of non-isomorphic trees.
\begin{theorem}[\cite{otter1948number}]\label{thm:treeCount}
The number of non-isomorphic trees on at most $n$ vertices is $2^{O(n)}$.
\end{theorem}

\paragraph{Embeddings and Layouts.}
An {\em embedding} of $G$ to $H$ is a mapping $f : V(G) \rightarrow V(H)$.
The {\em stretch} of $f$ is $\max_{uv \in E(G)} d_H(f(u), f(v))$.
The {\em congestion} of $f$ is $\max_{v \in V(H)} |f^{-1}(v)|$.
For a vertex set $S \subseteq V(G)$ the $S$-{\em contraction} of a function $f : S \rightarrow V(H)$ (or contraction of $S$ by $f$) is $\max_{u,v \in S} \frac{d_G(u,v)}{d_H(f(u), f(v))}$. Here the maximum is taken over all {\em distinct} pairs $u$, $v$ of vertices in $S$.
If $|S| \leq 1$ the $S$-contraction of $f$ is defined to be $1$, and if $f(u) = f(v)$ for two distinct vertices $u$ and $v$ in $S$ the $S$-contraction of $f$ is defined to be infinite. 
For a vertex set $S \subseteq V(G)$ the $S$-{\em stretch} of a function $f : S \rightarrow V(H)$ (or stretch of $S$ by $f$) is $\max_{u,v \in S} \frac{d_H(f(u),f(v))}{d_G(u, v)}$. Here the maximum is taken over all {\em distinct} pairs $u$, $v$ of vertices in $S$.
%\todo[inline]{Eran and Maria: $d_G(f(u), f(v))$ does not make sense. What did you mean? $\max_{u,v \in S} \frac{d_H(f(u),f(v))}{d_G(u, v)}$?}

%
The $S$-contraction and $S$-stretch of an embedding $f$ of $G$ into $H$ is the $S$-contraction ($S$-stretch) of the restriction of $f$ to the domain $S$. Note that the stretch of an embedding $f$ of $G$ to $H$ is equal to its $V(G)$-stretch. 
A {\em layout} of a graph $G$ is an embedding of $G$ into $P_{n}$ with congestion $1$. The {\em bandwidth} of a layout is its stretch, and the {\em bandwidth} of a graph $G$ is the minimum bandwidth of a layout of $G$. We will denote by $\textsf{bw}(G, f)$  the bandwidth of the layout $f$, and by $\textsf{bw}(G)$ the bandwidth of $G$.

\paragraph{(Radial) Local Density.}
The {\em ball of radius $r$ around $v$} is in $G$ denoted as $B_r(v)$ and is the set of all vertices at distance at most $r$ from $v$. If the ball is taken in a graph $G'$ other than $G$ it is denoted by $B_r^{G'}(v)$. The {\em radial local density} of $G$ is denoted by $\hat{\Delta}_r(G)$ and defined as $\max_{v,q} \frac{|B_q(v)|-1}{2q}$. It is easy to see that, equivalently, $\hat{\Delta}_r(G) = \max_{G' \subseteq G,\, \textsf{rad}(G') \geq 1} \frac{|V(G')|-1}{2 \cdot \textsf{rad}(G')}$. 
Recall that the {\em local density} of a graph $G$ is denoted as $\hat{\Delta}(G)$ and defined as $\hat{\Delta}(G) = \max_{G' \subseteq G} \frac{|V(G')|-1}{\textsf{diam}(G')}$.
%
%\todo{E: what's $\hat{\Delta}(G)$? Something looks wrong here. I agree with $\bw(G) \geq \hat{\Delta}_r(G)$. {\em local density $\hat{\Delta}(G)$ is defined in the intro, page 3}}
We have that for every graph $G$, $\hat{\Delta}(G) \geq \hat{\Delta}_r(G) \geq \hat{\Delta}(G)/2$ because for every graph its diameter is at least its radius and at most twice the radius. Thus, for every graph it holds that $\bw(G) \geq \hat{\Delta}(G) \geq \hat{\Delta}_r(G)$ (see~\cite{chinn1982bandwidth} for the simple argument that $\bw(G) \geq \hat{\Delta}(G)$).

\begin{lemma}\label{lem:densityMonotone}
For every graph $G$ and every subgraph $H$ of $G$, $\hat{\Delta}(H) \leq \hat{\Delta}(G)$ and $\hat{\Delta}_r(H) \leq \hat{\Delta}_r(G)$.
\end{lemma}

\begin{proof}
For both quantities, every subgraph of $H$ is also a subgraph of $G$, so the maximum ranging over subgraphs in the definition for $H$ is over a subset of the candidates appearing in the corresponding definition for $G$.
\end{proof}

A nice feature of radial local density is that every graph $G$ contains a subtree that completely preserves its radial local density.

\begin{lemma}\label{lem:preserveRadialLocalDensity}
For every connected graph $G$ there is a subtree $T$ of $G$ such that $\hat{\Delta}_r(G) = \hat{\Delta}_r(T)$.
\end{lemma}

\begin{proof}
Let $v$ be a vertex of $G$ and $q$ be an integer such that $\hat{\Delta}_r(G) = \frac{|B_q(v)|-1}{2q}$. Let $T$ be a BFS-tree of $G$ rooted at $v$. Since $T$ is a subgraph of $G$ we have that $\hat{\Delta}_r(G) \geq \hat{\Delta}_r(T)$.
Further, we have that $B_q^T(v) = B_q(v)$ and hence 
$\hat{\Delta}_r(T) \geq \frac{|B_q^T(v)|-1}{2q} = \frac{|B_q(v)|-1}{2q} \geq \hat{\Delta}_r(G)$.
\end{proof}

%\todo[inline]{The discussion about radial local density, and that there is always a tree with max radial local density should probably be in the intro. }

%{\em radial local density} of a graph $G$ is the maximum over all vertices v 

\paragraph{Metrics.}
A {\em metric} is a set $S$ together with a function $d : S \times S \rightarrow \mathbb{R}$ such that for all $u$, $v$, $w$ in $S$ we have $d(u, v) \geq 0$ with equality if and only if $u = v$, $d(u, v) = d(v, u)$, and $d(u, w) \leq d(u, v) + d(v, w)$. A metric is {\em integral} if $d(u, v) \in \mathbb{Z}$ for every $u, v \in S$. Notice that for every connected unweighted graph $G$, the distance function $d_G$ is an integral metric. We will frequently make use of the following result.

\begin{theorem}[\cite{matouvsek1990bi}]\label{thm:lineEmbedding}
There exists an algorithm that takes as input an integral metric $(S, d)$, runs in time polynomial in $|S|$, and outputs a function $f : S \rightarrow \mathbb{N}$ such that for every $u, v \in S$ we have 
$d(u,v) \leq |f(u) - f(v)| \leq 12|S|d(u,v)$.
\end{theorem}

We remark that the statement of Theorem~\ref{thm:lineEmbedding} in~\cite{matouvsek1990bi} is for metrics, rather than integer metrics, and for functions $f : S \rightarrow \mathbb{R}$, as opposed to  $f : S \rightarrow \mathbb{N}$. However, it is well known (see e.g.~\cite{DBLP:journals/toct/FellowsFLLRS13}) that, without loss of generality, one can assume that for every $u$, $v$ in $S$ such that $f(u) < f(v)$ and $f^{-1}(\{t~:~ f(u) < t < f(v)\}) = \emptyset$ we have $f(v)-f(u)=d(u,v)$. Hence for integral metrics the range of $f$ can be assumed to be $\mathbb{N}$.

\paragraph{Pathwidth.}
The {\em boundary} of a vertex set $S \subseteq G$ is the set $\partial(S) = \{v \in S ~:~ N(v) - S \neq \emptyset\}$ where $N(v)$ is the set of neighbors of $v$ in $G$. The {\em pathwidth} of a layout $f$ of $G$ is $\max_{i \leq n} |\partial(f^{-1}(\{1, \ldots, i\})))|$. The pathwidth of a graph $G$ is the minimum pathwidth of a layout of $G$. We will denote by $\textsf{pw}(G, f)$  the pathwidth of the layout $f$, and by $\textsf{pw}(G)$ the pathwidth of $G$.

For every layout $f$ of bandwidth $k$ and $i \leq n$ we have that 
$\partial(f^{-1}(\{1, \ldots, i\})) \subseteq f^{-1}(\{i-k+1, \ldots, i\})$, and therefore $|\partial(f^{-1}(\{1, \ldots, i\}))| \leq k$. Thus, for every layout $f$ of $G$ we have $\textsf{pw}(G, f) \leq \textsf{bw}(G, f)$, and hence for every graph $G$ we have $\textsf{pw}(G) \leq \textsf{bw}(G)$. 
%
%\todo{E: .. mention what?}
It is worth mentioning that the definition of pathwidth used here is known as the {\em vertex serparation number}~\cite{ellis1987graph}, while pathwidth is most commonly defined via path decompositions~\cite{Bodlaender98}. It is known that vertex separation number and pathwidth of a graph $G$ are the same~\cite{Bodlaender98,Kinnersley92,kinnersley1989obstruction}. We will need the following basic results about pathwidth:

\begin{lemma}[\cite{Bodlaender98}]\label{lem:pwsubgraph}
If $H$ is a minor of $G$ then $\textsf{pw}(H) \leq \textsf{pw}(G)$.
\end{lemma}

The next theorem characterizes the pathwidth of trees. We will use Theorem~\ref{thm:treePathwidth} to compute the pathwidth of trees that we are concerned with, as well as to lower bound the pathwidth of certain trees in the proof of our main decomposition theorem. 

\begin{theorem}[\cite{scheffler1990linear}]\label{thm:treePathwidth}
For every tree $T$ and integer $k$, $\textsf{pw}(T) \geq k+1$ if and only if there exists a node $v \in V(T)$ and three distinct components $C_1$, $C_2$, $C_3$ of $T-v$ such that $\textsf{pw}(C_i) \geq k$ for every $i \in \{1,2,3\}$.
\end{theorem}

%The complete binary tree $B_0$ of height $0$ is a single node, this node is considered the root of $B_0$. The complete binary tree of height $i \geq 1$ is obtained by taking two disjoint copies of $B_{i-1}$ adding a new node $v$, and adding edges from $v$ to the roots of the two copies of $B_{i-1}$. The vertex $v$ is then set to be root of $B_{i}$. It is well known (and a simple induction shows) that $|V(B_{i})| = 2^{i+1} - 1$. Scheffler~\cite{scheffler1990linear} determined the pathwidth of complete binary trees.

%\begin{theorem}[\cite{scheffler1990linear}]\label{thm:pwBinaryTrees}
%$\pw(B_i) = \lceil i/2 \rceil$.
%\end{theorem}

The next theorem shows that for every fixed tree $T$, every graph that does not contain $T$ as a minor has pathwidth upper bounded by a function of $T$. This result was first proved by Robertson and Seymour~\cite{robertson1983graph}, we will rely on a subsequent result with better dependence on $T$ (see also~\cite{cattell1996simple}).

\begin{theorem}[\cite{bienstock1991quickly}]\label{thm:pwExcludingTree}
For every tree $T$, if $G$ does not contain $T$ as a minor then $\textsf{pw}(G) \leq |V(T)|-2$.
\end{theorem}

A tight bound on the size of tree obstructions for pathwidth was established by Cattell, Dinneen, and Fellows.

\begin{theorem}[\cite{CattellDF96}]\label{thm:pwTreeObstructionSize}
For every integer $\tau \geq 1$, every tree of pathwidth at least $\tau$ contains as a minor a tree of pathwidth at least $\tau$ on at most $\frac{5 \cdot 3^{\tau-1} - 1}{2}$ vertices.
Furthermore, for every positive integer $k$ there exists a tree $T_k$ of pathwidth at least $k$.
\end{theorem}

The next lemma extracts an actual subtree from a tree minor model.

\begin{lemma}\label{lem:minorToSubtree}
Let $T$ be a tree and $G$ a graph. If $T$ is a minor of $G$, then $G$ contains a subtree $T^*$ such that $T$ is a minor of $T^*$.
\end{lemma}

\begin{proof}
Fix a minor model $\{X_v\}_{v \in V(T)}$ of $T$ in $G$. For every $v \in V(T)$, fix a spanning tree $S_v$ of $G[X_v]$, and for every edge $uv \in E(T)$, fix one edge $e_{uv} \in E(G)$ with one endpoint in $X_u$ and one in $X_v$. Let $T^*$ be the subgraph of $G$ with edge set $\bigcup_{v \in V(T)} E(S_v) \cup \{e_{uv} : uv \in E(T)\}$ and vertex set $\bigcup_{v \in V(T)} X_v$. Then $T^*$ is a tree, and $\{X_v\}_{v \in V(T)}$ is a minor model of $T$ in $T^*$.
\end{proof}

Combining Theorems~\ref{thm:treePathwidth} and~\ref{thm:pwExcludingTree} with Lemma~\ref{lem:minorToSubtree} we obtain the following result.

\begin{theorem}\label{thm:pwObstructions}
For every graph $G$ and integer $k \geq 1$, $G$ either contains a subtree of pathwidth at least $k$, or 
$\textsf{pw}(G) \leq  \frac{5 \cdot 3^{k-1} - 5}{2}$.
\end{theorem}

\begin{proof}
For every integer $k \geq 1$, let $T_k$ be a tree of pathwidth at least $k$ on at most  $\frac{5 \cdot 3^{k-1} - 1}{2}$ vertices, as guaranteed by Theorem~\ref{thm:pwTreeObstructionSize}.
If $G$ does not contain $T_k$ as a minor, then by Theorem~\ref{thm:pwExcludingTree} we have $\textsf{pw}(G) \leq |V(T_k)| - 2 = \frac{5 \cdot 3^{k-1} - 5}{2}$.

Otherwise, by Lemma~\ref{lem:minorToSubtree}, $G$ contains a subtree $T'$ such that $T_k$ is a minor of $T'$. From Lemma~\ref{lem:pwsubgraph} we have that $\textsf{pw}(T') \geq \textsf{pw}(T_k) \geq k$.
\end{proof}

% OLD PROOF.
%\begin{proof}
%Suppose first that $G$ contains $B_{2k}$ as a minor, in which case $G$ contains (as a subgraph) a tree of pathwidth at least $\textsf{pw}(B_{2k})$. By Theorem~\ref{thm:pwBinaryTrees} we have $\textsf{pw}(B_{2k}) \geq k$ and so $G$ contains as a subtree $T$ of pathwidth at least $k$. 
%If $G$ does not contain $B_{2k}$ as a minor then, by Theorem~\ref{thm:pwExcludingTree} we have $\textsf{pw}(G) \leq 4^{k+1}$.
%\end{proof}

We will also need an algorithm to compute path decompositions.

\begin{theorem}[\cite{furer2016faster,BodlaenderJT23,DBLP:journals/jal/BodlaenderK96}]\label{thm:computePw}
There exists an algorithm that takes as input a graph $G$ and integer $k$, runs in time $2^{O(k^2)}n$ and either concludes that $\textsf{pw}(G) > k$ or produces a layout of $G$ with pathwidth at most $k$.
\end{theorem}

The next theorem, due to Hicks, is stated in terms of {\em branch decompositions}, a notion we do not define in this paper. The only fact about branch decompositions that we will use is that there is a polynomial-time algorithm that, given a graph $G$ and a path decomposition of $G$ of width $k$, outputs a branch decomposition of $G$ of width at most $k+1$~\cite{RobertsonS91}.

\begin{theorem}[\cite{Hicks04}]\label{thm:hicksMinor}
There is an algorithm that takes as input a graph $G$, a branch decomposition of $G$ of width $w$, and a graph $H$ on $h$ vertices, runs in time $O(3^{w^2} \cdot (h+w-1)! \cdot |E(G)|)$ and either outputs a minor model of $H$ in $G$ or determines that $H$ is not a minor of $G$.
\end{theorem}

\section{Manipulating Embeddings}\label{sec:embeddings}
In this section we develop some basic tools for dealing with embeddings with low congestion and stretch. 
We show that the existence of good embeddings is essentially transitive, relate low bandwidth layouts with low stretch and congestion embeddings into a path, and use this to show that adding a small number of edges to a low bandwidth graph cannot increase the bandwidth too much.  

\begin{lemma}\label{lem:composeEmbeddings}
Let $F$, $G$, $H$ be graphs, $S \subseteq V(F)$ be a vertex set. 
Let $f$ be an embedding of $F$ to $G$ with congestion $c_f$, stretch $s_f$ and $S$-contraction $z_f$.
Let $g$ be an embedding of $G$ to $H$ with congestion $c_g$, stretch $s_g$ and $f(S)$-contraction $z_g$.
Let $h : V(F) \rightarrow V(H)$ be the embedding defined as $h(v) = g(f(v))$. 
Then congestion of $h$ is at most $c_fc_g$, stretch of $h$ is at most $s_fs_g$ and $S$-contraction of $h$ is at most $z_fz_g$.
\end{lemma}

\begin{proof}
For stretch, let $uv$ be an edge in $F$ and $P$ be a 
shortest
path from $f(u)$ to $f(v)$ in $G$. We have that
$$d_H(h(u), h(v)) \leq \sum_{xy \in E(P)}d_H(g(x), g(y)) \leq \sum_{xy \in E(P)} s_g \leq |E(P)|s_g \leq s_fs_g\mbox{.}$$
For congestion, let $u$ be a vertex of $H$. We have that 
$$|h^{-1}(u)| = \sum_{v \in g^{-1}(u)} |f^{-1}(v)| \le \sum_{v \in g^{-1}(u)} c_f \le c_fc_g\mbox{.}$$
For $S$-contraction, let $x$ and $y$ be distinct elements of $S$. Then 
$$d_H(h(x), h(y)) \geq \frac{d_G(f(x), f(y))}{z_g} \geq \frac{d_F(x, y)}{z_f z_g}\mbox{.}$$
\end{proof}

The following easy lemma shows how to fold a path such that the endpoints of the path are placed in prescribed positions. 

%\todo[inline]{Maybe we should remove Lemma~\ref{lem:pathWithFixedEndpoints} and invoke Lemma~\ref{lem:pathWithTwoEndpointsGeneral} instead (at cost of factor $3$ instead of $2$). No since this messes with numbers :)}

\begin{lemma}\label{lem:pathWithFixedEndpoints}
Let $P = v_1, \ldots, v_t$ be a path with $t\ge 2$, and $a$, $b$ and $k \geq 1$ be integers such that $|a-b| \leq (t-1)k$. Then there exists an embedding $f$ from $P$ to $P_{\mathbb{Z}}$ with congestion at most $2$ and stretch at most $k$ such that $f(v_1) = a$, $f(v_t) = b$ and $f(v_i) \geq \min(a,b)$ for every $v_i$ in $P$.
\end{lemma}

\begin{proof}
Without loss of generality $a \leq b$.
Consider first the case when $|a - b| \geq t - 1$.
Set $\delta = \frac{b-a}{t-1}$ and define $f(v_i) = \lfloor a + \delta(i-1)\rfloor$. Then $f(v_1) = a$, $f(v_t) = b$, and $\delta \geq 1$ so the congestion of $f$ is at most $1$. Further, the stretch of $f$ is at most $\lceil \delta \rceil$. However $\delta \leq k$ since $|a-b| \leq (t-1)k$. Since $k$ is an integer we get $\lceil \delta \rceil \leq k$ as well. 

Consider now the case when $a=b$. In this case let $f(v_i)=a+i-1$ for $i \leq \lceil t/2 \rceil$ and $f(v_i) = a + (t - i)$  for $i > \lceil t/2 \rceil$. Then $f$ has congestion at most $2$, stretch at most $1$ and satisfies $f(v_1)=f(v_t)=a=b$.

Finally consider the case when $1 \leq b-a < t-1$.
%Without loss of generality $a < b$.
Set $f(v_i) = a + i - 1$ for $i \leq b - a$ and use the construction for the $b=a$ case for the remaining sub-path $P' = v_{b-a+1}, \ldots, v_t$. 
Since $f(v_{b-a}) = b - 1$ we have that $f$ has stretch at most $1$.
Since $f(v_i) < b$ for all $i \leq b - a$ and $f(v_i) \geq b$ for all $i > b - a$, $f$ has congestion at most $2$. For all three cases the condition $f(v_i) \geq \min(a,b)$ for every $v_i$ in $P$ holds by construction.
\end{proof}

We will also need a generalization of Lemma~\ref{lem:pathWithFixedEndpoints} where the path $P$ is embedded into a general graph $G$ instead of into $P_{\mathbb Z}$. We first show that we can embed a path $P$ into any connected graph $G$ with sufficiently many vertices, as long as we want the first and last vertex of $P$ to be mapped to the same vertex of $G$.

\begin{lemma}\label{lem:pathWithOneEndpointsGeneral}
There exists a polynomial time algorithm that takes as input
a path $P = v_1, \ldots, v_t$,
an integer $k \geq 2$,
a connected graph $G$ on at least $t/k$ vertices,
and a vertex $r$ of $G$,
and outputs an embedding $f$ of $P$ into $G$ with stretch at most $2$, congestion at most $k$, such that $f(v_1) = f(v_t) = r$.
\end{lemma}

\begin{proof}
We may assume without loss of generality that $G$ is a tree rooted at $r$: if not the algorithm computes a spanning tree $G'$ of $G$ rooted at $r$ and works with $G'$ rather than $G$. An embedding $f$ of $P$ into $G'$ with stretch at most $2$, congestion at most $k$, such that $f(v_1) = f(v_t) = r$ is also such an embedding into $G$. From now on we will consider the case that $G$ is a tree rooted at $r$.  Further we may assume that $|V(G)| = \lceil t/k \rceil$, as otherwise we may pick a leaf $v$ of $G$, set $G' = G-v$ and work with $G'$ instead. Again an embedding $f$ of $P$ into $G'$ with stretch at most $2$, congestion at most $k$, such that $f(v_1) = f(v_t) = r$ is also such an embedding into $G$. Thus we will assume that $G$ is a tree rooted at $r$ and that $|V(G)| = \lceil t/k \rceil$.

The algorithm computes a depth-first-search traversal of $G$ rooted at $r$, recording the first and last time every vertex is visited by the traversal. More formally we define a function $DFS(T, r)$ that takes as input a tree $T$ and vertex $r \in T$, and outputs a sequence of vertices in $T$. If $V(T) = \{r\}$ then $DFS(T, r) = r,r$. Otherwise let $T_1, \ldots, T_p$ be the connected components of $T-r$ and let $r_i$ be the unique neighbor of $r$ in $T_i$. We define $DFS(T,r)$ as follows:
$$DFS(T,r) = r \circ DFS(T_1, r_1) \circ \ldots \circ DFS(T_p, r_p) \circ r\mbox{.}$$
Here $\circ$ denotes concatenation of sequences (so $a,b,c \circ b,a,d = a,b,c,b,a,d$). 

It follows from the definition of $DFS(G, r)$ that every vertex of $G$ appears precisely two times in the sequence. Furthermore for every pair of vertices $u$, $v$ appearing consecutively in the sequence $DFS(G, r)$, either $u=v$, $u$ is the parent of $v$, $v$ is the parent of $u$, or $u$ and $v$ have the same parent in $G$ (see e.g.~\cite{kleinberg2006algorithm}). Thus $d_G(u, v) \leq 2$. 

From $DFS(G, r)$ we make a sequence $S$ on precisely $|V(G)| \cdot k$ vertices by replacing the {\em first} occurrence of every vertex $v$ in $DFS(G, r)$ by $k-1$ occurrences of $v$. Since every vertex occurs precisely two times in  $DFS(G, r)$ it occurs precisely $k$ times in $S$.
Since $|V(G)| = \lceil t/k \rceil$ we have that $0 \leq |V(G)| \cdot k - t < k$. Let $\ell$ be a leaf of $G$. From the definition of $DFS(G, r)$ it follows that all $k$ occurrences of $\ell$ appear consecutively in $S$. Create the sequence $S'$ from $S$ by removing precisely $|V(G)| \cdot k - t$ occurrences of $\ell$. Note that at least one occurrence of $\ell$ remains in the sequence. Hence, for every pair $u, v$ of {\em distinct} vertices appearing consecutively in $S'$ they also appear consecutively in $S$. Therefore we have that $|S'| = t$, the first and last vertex in the sequence $S'$ is $r$, and for every pair $u, v$ of consecutive vertices in $S'$ we have that $d_G(u,v) \leq 2$.

The function $f$ maps $v_i$ to the $i$'th vertex in the sequence $S'$. Since the length of $S'$ is precisely $t$, $f$ is well defined and $f(v_1) = f(v_t) = r$. Further, since every vertex of $G$ appears at most $k$ times in $S'$ the congestion of $f$ is at most $k$. For every edge $v_{i}v_{i+1}$ in $P$ the vertices $f(v_i)$ and $f(v_{i+1})$ appear consecutively in $S'$ and hence $d_G(f(v_i), f(v_{i+1})) \leq 2$. Thus the stretch of $f$ is at most $2$.  

Each step of the algorithm runs in polynomial time. 
\end{proof}

Next we extend Lemma~\ref{lem:pathWithOneEndpointsGeneral} to the case where we would like to map the two endpoints of $P$ into different prescribed vertices.
%Note that Lemma~\ref{lem:pathWithTwoEndpointsGeneral} generalizes Lemma~\ref{lem:pathWithFixedEndpoints}, except that the congestion is slightly worse, and 

\begin{lemma}\label{lem:pathWithTwoEndpointsGeneral}
There exists a polynomial time algorithm that takes as input
a path $P = v_1, \ldots, v_t$,
two positive integers $s$ and $k$ such that $k \geq 2$,
a connected graph $G$ on at least $t/k$ vertices,
and two vertices $a$ and $b$ of $G$ such that $d_G(a, b) \leq s \cdot (t-1)$,
and outputs an embedding $f$ of $P$ into $G$ with stretch at most $\max(2,s)$ and congestion at most $k+1$, such that $f(v_1) = a$ and $f(v_t) = b$.
\end{lemma}

\begin{proof}
We first prove the lemma for the case where $d_G(a, b) \geq t-1$. If $t = 1$ then $a = b$, thus setting $f(v_1) = a$ satisfies the lemma trivially. We therefore assume $t \geq 2$. Let $Q = q_1, q_2, \ldots, q_\ell$ be a shortest path from $a$ to $b$ in $G$ (thus $q_1 = a$ and $q_\ell = b$). For every $i \leq t$ we define $f(v_i) = q_j$ where $j = 1 + \lceil (i-1)(\ell-1)/(t-1) \rceil$. Note that $f(v_1) = q_1 = a$ and $f(v_t) = q_\ell = b$. Since $d_G(a, b) = \ell - 1 \leq s(t-1)$ it follows that $(\ell-1)/(t-1) \leq s$ and hence $d_G(f(v_i), f(v_{i+1})) \leq s$. Thus the stretch of $f$ is at most $s$. Since  $\ell - 1 =  d_G(a, b) \geq t-1$ we have that $(\ell-1)/(t-1) \geq 1$ and therefore $f(v_i) \neq f(v_{i'})$ for $i \neq i'$. Thus the congestion of $f$ is at most $1$.

We now consider the case that $d_G(a, b) < t-1$.
%$a=b$ then the statement of the lemma follows immediately from Lemma~\ref{lem:pathWithOneEndpointsGeneral}. Suppose now that $1 \leq d_G(a, b) < t-1$.
Set $\ell = d_G(a, b) + 1$ and define $P_1 = v_1, v_2, \ldots, v_\ell$ and $P_2 = v_\ell, v_{\ell+1}, \ldots, v_t$. Let $Q = q_1, \ldots, q_\ell$ be a shortest path from $a$ to $b$ in $G$. Define the embedding $f_1$ from $P_1$ to $Q$ as $f_1(v_i) = q_i$ and observe that $f_1(v_1)=a$, $f_1(v_\ell)=b$, the stretch of $f_1$ is at most $1$ and the congestion of $f_1$ is at most $1$. Obtain an embedding $f_2$ of $P_2$ into $G$ with $f_2(v_\ell) = f_2(v_t) = b$ by applying Lemma~\ref{lem:pathWithOneEndpointsGeneral} to $P_2$ and $G$ with $r=b$. The stretch of $f_2$ is at most $2$ and the congestion of $f_2$ is at most $k$. Define an embedding $f$ of $P$ into $G$ by setting $f(v_i) = f_1(v_i)$ for $i \leq \ell$ and $f(v_i) = f_2(v_i)$ for $i > \ell$. The congestion of $f$ is at most $k+1$ (the sum of congestions of $f_1$ and $f_2$). Further, for every edge $v_iv_{i+1}$ of $P$, it is either an edge of $P_j$ for $j \in \{1,2\}$. Then $f(v_i) = f_j(v_i)$ and $f(v_{i+1}) = f_j(v_{i+1})$ and since the stretch of $f_j$ is at most $2$, the stretch of $f$ is at most $2$ as well. All constructions take polynomial time.
\end{proof}

The next lemma shows how to turn any embedding of $G$ into $P_\mathbb{Z}$ into a layout.
A vertex set $S$ {\em comes first} in a layout $f$ (or an embedding $f$ of a graph into $P_\mathbb{Z}$) if, for all $x \in S$ and $y \notin S$ we have $f(x) \leq f(y)$.

\begin{lemma}\label{lem:embeddingToLayout}
There exists a polynomial time algorithm that takes as input 
a graph $G$, 
an embedding $f$ of $G$ into $P_\mathbb{Z}$ with stretch $s_f$ and congestion $c_f$, 
and a vertex set $S$ of $G$ that comes first in $f$,
and outputs a layout $g$ of $G$ with bandwidth at most $(s_f+1)c_f - 1$ where $S$ comes first.
\end{lemma}

\begin{proof}
We assume without loss of generality that $V(G) = \{v_1, \ldots, v_n\}$.
We define a total ordering $\leq$ on the vertices of $G$ as follows: for $i \neq j$ we have that 
$v_i \leq v_j$ if $f(v_i) < f(v_j)$, 
or if $f(v_i) = f(v_j)$ and $v_i \in S$ and $v_j \notin S$,
or if $f(v_i) = f(v_j)$ and $|\{v_i,v_j\} \cap S| \in \{0,2\}$ and $i < j$. In other words, we sort the vertices by their $f$-value, first breaking ties by whether they are in $S$ or not (selecting vertices in $S$ first), and then by their name. 

Since $\leq$ is a total order it defines a layout $g$. To see that $S$ comes first in $g$ consider two elements $v_i$ and $v_j$ such that $v_i \in S$ and $v_j \notin S$. We prove that $g(v_i) < g(v_j)$. 
Since $S$ comes first in $f$ we have that $f(v_i) \leq f(v_j)$. If $f(v_i) < f(v_j)$ then $g(v_i) < g(v_j)$. If $f(v_i) = f(v_j)$ then again $g(v_i) < g(v_j)$ by the first tie breaking rule. Thus $S$ comes first in $g$.

Finally we bound the stretch of $g$. 
Let $v_iv_j$ be an edge of $G$. Let $p = \min(f(v_i), f(v_j))$ and $q = \max(f(v_i), f(v_j))$. We have that $q-p \leq s_f$. Furthermore, 
%\todo[inline]{E: isn't the first inequality below actually an equality? DL: no, suppose G is a clique $v1,v2,v3,v4$ and f(v1)=f(v2)=1, f(v3)=f(v4)=2. Look at $v2v3$}
$$g(v_i) - g(v_j) \leq |f^{-1}(\{p, p+1, \ldots, q\})| - 1 \leq (q-p+1)c_f -1 \leq (s_f+1)c_f - 1\mbox{.}$$
Comparing two vertices $v_i$ and $v_j$ in the total ordering $\leq$ takes polynomial time, hence the sort takes polynomial time as well. 
\end{proof}

Most of the time when we apply Lemma~\ref{lem:embeddingToLayout} we will have $s_f \geq 1$ and use the upper bound $2s_fc_f$ for $(s_f+1)c_f - 1$. The stricter bound of $(s_f+1)c_f - 1$ is just there to make clear that when congestion of $f$ is $1$ then Lemma~\ref{lem:embeddingToLayout} does not increase stretch. 
The next lemma shows that at the cost of a small increase in the bandwidth we may put any small set $S$ first. 
\begin{lemma}\label{lem:makeSFirst}
There exists a polynomial time algorithm that takes as input a graph $G$, bandwidth $k$ layout $f$ of $G$, and a vertex set $S \subseteq V(G)$ and outputs a bandwidth $2(k+1)(|S|+1)$ layout of $G$ in which $S$ comes first. 
\end{lemma}

\begin{proof}
We treat the layout $f$ as a stretch $k$ and congestion $1$ embedding of $G$ into a path $P$ on $n$ vertices. 
Let $S'$ be the union of $S$ and the first and last vertex of $G$ (according to $f$). Thus $|S'| \leq |S| + 2$.
Divide $P$ into $\ell = |S'|-1 \leq |S|+1$ sub-paths $P_1, \ldots, P_\ell$ where, for every $1 \leq i \leq \ell$, $P_i$ is the sub-path of $P$ that starts at the $i$'th vertex of $S'$ on $P$ and ends at the $i+1$'st vertex of $S'$ on $P$.
For each path $P_i$ obtain using Lemma~\ref{lem:pathWithFixedEndpoints} (with $a=b=k=1$) an embedding $g_i$ of $P_i$ into $P_\mathbb{Z}$ such that the first and last vertex of $P_i$ are mapped to $1$ by $g_i$. 
Each $g_i$ has stretch at most $1$ and congestion at most $2$. 
Define $g : V(P) \rightarrow \mathbb{Z}$ as follows. For every $v \in V(P)$, if $v \in V(P_i)$ then $g(v) = g_i(v)$. Note that $g$ is well defined: every vertex $v \in V(P)$ is in $V(P_i)$ for some $i$. Further, if $v \in V(P_i) \cap V(P_j)$ for $i \neq j$ then $v$ is an endpoint of $P_i$ and $P_j$ and therefore mapped to $1$ both by $g_i$ and by $g_j$. 
Note that $g$ maps every vertex of $f(S)$ to $1$ and every vertex of $P$ to at least $1$, so $f(S)$ comes first in $g$.
Further $g$ has stretch $1$ and congestion $2\ell = 2(|S|+1)$.

Let $h'$ be the embedding of $G$ into $P_\mathbb{Z}$ defined as $h'(v) = g(f(v))$ for every $v \in V(G)$. By Lemma~\ref{lem:composeEmbeddings} the stretch of $h'$ is at most $k$ and congestion of $h'$ is at most $2(|S|+1)$. Additionally $S$ comes first in $h'$ because $f(S)$ comes first in $g$.
Finally, obtain a layout $h$ of $G$ by applying Lemma~\ref{lem:embeddingToLayout} to $h'$ and $S$. By Lemma~\ref{lem:embeddingToLayout} $S$ comes first in $h$, and the bandwidth of $h$ is at most $2(k+1)(|S|+1)$. Each of the steps takes polynomial time.
\end{proof}

We will often need to lift the layout of $G-S$ for a small set $S$ to a layout of $G$.

\begin{lemma}\label{lem:deleteSet}
There exists a polynomial time algorithm that takes as input a graph $G$ of maximum degree $d$, a vertex set $S$ in $G$, a bandwidth $k$ layout $f$ of $G-S$, and outputs a layout of $G$ with bandwidth at most $12k|S|d$.
\end{lemma}

\begin{proof}
The algorithm sets $X = N(S)$ and applies Lemma~\ref{lem:makeSFirst} to $G-S$ and $X$ to make an embedding $g$ of $G-S$ with bandwidth $2(k+1)(|X|+1)$ in which $X$ comes first. Since $G$ has maximum degree $d$ we have that $|X| \leq |S|d$, and so the bandwidth of $g$ is at most $12k|S|d$. Extend $g$ to a layout $g'$ of $G$ by mapping the vertices of $S$ immediately before the vertices of $G-S$. Since all edges from $S$ have the other endpoint in $S \cup X$ we have that for every edge $uv$ with one endpoint in $S$ it holds that $g'(u) - g'(v) \leq |X|+|S| \leq |S|(d+1) \leq 12k|S|d$. Since the algorithm of Lemma~\ref{lem:makeSFirst} takes polynomial time, this algorithm runs in polynomial time as well.
\end{proof}

\section{Low Bandwidth Embeddings with Prescribed Vertices}\label{sec:prescribed}
Let $G$ be a graph of bandwidth at most $k$, and $H$ be a connected graph. 
Treating the layout of $G$ as an embedding of $G$ into a path, Lemma~\ref{lem:pathWithOneEndpointsGeneral} tells us that 
%if $H$ is not much smaller than $G$ then 
there exists an embedding $g$ of $G$ into $H$ where the stretch and congestion of the embedding are upper bounded as a function of $k$ and $|V(G)|/|V(H)|$.
Suppose now that we are in a more general setting. In addition to $G$ and $H$ we are given a ``small'' vertex set $S \subseteq V(G)$ and a function $h : S \rightarrow V(H)$. We now want the embedding $g$ to ``respect'' $h$ in the sense that $g(v) = h(v)$ for every $v$ in $S$.
How much can this additional requirement affect the stretch and congestion of $g$?
An immediate answer is that the stretch of $g$ can increase a lot. If $u$ and $v$ are vertices in $S$ that are close together in $G$, but $h(u)$ and $h(v)$ are very far away from each other in $H$, then the stretch of $g$ must be large. 
However, it turns out that this is essentially the only bad thing that can happen. 
In particular, the main result of this section is that 
%if $S$ and the $S$-stretch of $h$ are both small, then 
there exists an embedding $g$ of $G$ into $H$ with stretch and congestion that are upper bounded by a function of $k$, $|V(G)|/|V(H)|$, $|S|$, and the $S$-stretch of $h$. 

The main technical ingredient of the proof is that for every graph $G$ of bandwidth at most $k$ and every small vertex set $S \subseteq V(G)$, there exists a tree $T$ with few leaves and an embedding of $G$ into $T$ with low stretch, congestion, and $S$-contraction.
%
%We first show that this goal can be attained if we embed into a tree with at most $|S|$ leaves instead of a path. 

%\todo[inline]{define/discuss distortion in intro. I'm adding a sentence right here for now, but makes sense to add in future. }

%In this section we prove our first main theorem, namely that whenever a graph $G$ has bandwidth at most $k$, then for every vertex subset $S$ of $G$ there exists a layout of $G$ with small bandwidth that simultaneously is a low distortion (low $S$-stretch and $S$-contraction) embedding of $S$ into the line. 
%

%
%The main technical ingredient is that this goal can be attained if we embed into a tree with at most $|S|$ leaves instead of a path. 

\begin{lemma}\label{lem:treeEmbedding}
There exists a polynomial time algorithm that takes as input a connected graph $G$, a vertex set $S$, and a bandwidth $k$ layout $f$ of $G$ in which $S$ comes first, and outputs a tree $T$ with at most $2|S|$ leaves and an embedding $g$ of $G$ into $T$ with stretch at most $k$, congestion $1$, and $S$-contraction at most $\max(2|S|-1, 2k)$. 

\end{lemma}

\begin{proof}
We re-name the vertices of $G$ as $v_1, \ldots, v_n$ such that $f(v_i) = i$ for every $i$.
Then define $V_i = \{v_1, \ldots, v_i\}$ and $G_i = G[V_i]$. 
For each integer $i$ we define $C_i$ to be the connected component of $G_i$ containing $v_i$. Note that since $G$ is connected we have $C_n = G$. Next we show that the family of $C_i$'s forms a laminar family. In fact we show a slightly stronger statement.

\sta{\label{clm:CiLaminar}
For every $i$, $j$, if $i \leq j$ and $C_i \cap C_j \neq \emptyset$ then $C_i \subseteq C_j$
}
Let $u$ be an arbitrary vertex in $C_j$ and $v$ be an arbitrary vertex in $C_i$. Let $w$ be a vertex in $C_i \cap C_j$. There is a path from $u$ to $w$ in $C_j$ (and hence in $G_j$) and a path from $w$ to $v$ in $C_i$, and hence in $G_j$. Thus $u$ and $v$ are in the same component of $G_j$, namely $C_j$. This proves~(\ref{clm:CiLaminar}). 

An immediate consequence of (\ref{clm:CiLaminar}) is the following:

\sta{\label{clm:edgeContainment}
For every $i$, $j$, if $i < j$ and $v_iv_j \in E(G)$ then $C_i \subseteq C_j$
}
This follows directly from (\ref{clm:CiLaminar}) because $v_i \in C_i \cap C_j$.

A vertex $v_i$ is an {\em anchor} if $C_i \cap S \neq \emptyset$. For an anchor vertex $v_i$ its {\em leaf set} $\ell(i)$ is defined as $C_i \cap S$. For every vertex $v_i$ we define the anchor $a(i)$ of $v_i$ to be the lowest $j \geq i$ such that $v_j$ is an anchor and $C_i \subseteq C_j$. Observe that since $G$ is connected $v_n$ is an anchor and hence $a(i)$ is well defined for every $i$. Further note that for every anchor $v_i$, $a(i) = i$. For every non-anchor vertex $v_i$ we define its {\em leaf set} $\ell(i)$ to be $\ell(a(i))$. With this definition we have that $\ell(i) = \ell(a(i))$ for every vertex $v_i$, whether it is an anchor or not. Next we show that the leaf sets also satisfy a property similar to~(\ref{clm:CiLaminar}).

\sta{\label{clm:ellIsLaminar}
For every pair $i$,$j$ of distinct integers in $\{1, \ldots, n\}$, if $\ell(i) \cap \ell(j) \neq \emptyset$ and $a(i) \leq a(j)$ then $\ell(i) \subseteq \ell(j)$.}
Since 
$\ell(i) \cap \ell(j) \neq \emptyset$ and $\ell(i) = \ell(a(i))$ and $\ell(j) = \ell(a(j))$ we have $\ell(a(i)) \cap \ell(a(j)) \neq \emptyset$. So $C_{a(i)} \cap C_{a(j)} \neq \emptyset$ and hence by~(\ref{clm:CiLaminar}) $C_{a(i)} \subseteq C_{a(j)}$. But then $\ell(i) \subseteq \ell(j)$, proving~(\ref{clm:ellIsLaminar}). 

This has the following consequence. 

\sta{\label{clm:ellSizeBound} $|\{\ell(j) ~:~ j \leq n\}| \leq 2|S|-1$}
It follows directly from (\ref{clm:ellIsLaminar}) that $\{\ell(j) ~:~ j \leq n\}$ is a laminar family over the universe $S$. Every laminar family over a universe of size $|S|$ contains at most $2|S|-1$ non-empty sets (see e.g.~\cite{schrijver2003combinatorial} together with the well-known fact that a rooted tree with $|S|$ leaves has at most $2|S|-1$ vertices~\cite{diestelBook}), proving (\ref{clm:ellSizeBound}). 
%\todo{bad citation since it is an exercise there} 
%\todo[inline]{according to chat gpt, these are better citations but i dont have access to check it :( -- Combinatorial Optimization: Polyhedra and Efficiency, Vol. A, Springer, 2003, A. Schrijve -- OR Bollobás, B. Combinatorics: Set Systems, Hypergraphs, Families of Vectors and Combinatorial Probability. }

For every $j < n$ we define the {\em parent} $p(j)$ of $v_j$ as follows. If there exists an $i > j$ such that $\ell(i) = \ell(j)$ then $p(j)$ is set to be the smallest such $i$. Otherwise $p(j)$ is defined to be the smallest $i > j$ such that $C_j \subseteq C_i$. Since $G$ is connected $C_i \subseteq C_n$ for every $i < n$ and thus every vertex (except the root $v_n$) has a parent.

We define $T$ to be the graph with vertex set $V(G)$ and edge set $\{v_jv_{p(j)} : 1 \leq j < n\}$. Since every vertex $v_j$ with $j < n$ has an edge to a vertex $v_{p(j)}$ with $p(j) > j$, every vertex is in the same connected component of $T$ as $v_n$, and hence $T$ is connected. Since $T$ has precisely $n-1$ edges we conclude that $T$ is a tree. 

Next we show that $T$ has at most $2|S|$ leaves. Define $Z = \{v_i ~:~ \nexists j < i \mbox{ s.t. } \ell(j) = \ell(i)\}$. We have that $|Z|$ is upper bounded by the number of distinct sets in  $\{\ell(j) ~:~ j \leq n\}$ and hence by (\ref{clm:ellSizeBound}) we have $|Z| \leq 2|S|-1$. Thus it suffices to show that for every vertex $v_j$ of $T$, if $j < n$ and $v_j \notin Z$ then $v_j$ is not a leaf. Let $v_j$ be such a vertex. Then $v_jv_{p(j)}$ is an edge incident with $v_j$. Further, since $v_j \notin Z$ there exists an $i < j$ such that $\ell(i) = \ell(j)$. Pick the largest such $i$. By the definition of parents, we have that $j$ is the parent of $i$, and hence $v_iv_j$ is also an edge of $T$. But then $v_j$ is not a leaf. Hence $T$ has at most $2|S|$ leaves. 

We define the embedding $g : V(G) \rightarrow V(T)$ as the identity embedding $g(v_i) = v_i$ for all $1 \leq i \leq n$. It remains to upper bound the stretch, congestion and $S$-contraction of $g$. From the definition of $g$ it immediately follows that the congestion $c_g$ of $g$ is $1$.

Towards upper bounding the stretch we define {\em ancestors} of a vertex. We define $p^1(i) = p(i)$. For every $t \geq 2$ such that $p^{t-1}(i) \neq v_n$ we define $p^t(i) = p(p^{t-1}(i))$. We say that $j$ is an {\em ancestor} of $i$ if there exists an integer $t$ such that $j = p^t(i)$. Note that if $j = p^t(i)$ then $d_T(v_i, v_j) = t$.
%Observe that $j$ is an ancestor of $i$ if and only if there exists a path in $T$ from $i$ to $j$ such ...\todo{closer to root}

\sta{\label{clm:canReach}
For every pair $i$,$j$ of distinct integers in $\{1, \ldots, n\}$, if $C_i \subseteq C_j$ then $j$ is an ancestor of $i$. }

It is sufficient to prove (\ref{clm:canReach}) for pairs $i, j$ where $j$ is the smallest integer greater than $i$ such that $C_i \subseteq C_j$. Then (\ref{clm:canReach}) follows by induction on $j-i$. Suppose now that $j$ is the smallest integer greater than $i$ such that $C_i \subseteq C_j$. If $\ell(i) = \ell(j)$ then $j$ is an ancestor of $i$ because the sequence $p^1(i)$, $p^2(i)$, $\ldots$ iterates through all integers $i'$ greater than $i$ such that $\ell(i') = \ell(i)$.

We now consider the case that $\ell(i) \neq \ell(j)$.
We claim that in this case $i$ is an anchor. Suppose not. Then $a(i) > i$ and $C_i \subseteq C_j \subseteq C_{a(i)}$.
We have that $a(i) \leq a(j)$ since $C_i \subseteq C_{a(j)}$. 
Further, $a(j) \leq a(i)$ since $C_j \subseteq C_{a(i)}$.
Thus $a(i)=a(j)$.
But then $\ell(i) = \ell(a(i)) = \ell(a(j)) = \ell(j)$, contradicting $\ell(i) \neq \ell(j)$. We conclude that $i$ is an anchor. 

Since $i$ is an anchor, $j$ is also an anchor. We prove that $j$ is the parent of $i$. In particular, we show that there does not exist an $i' > i$ such that $\ell(i') = \ell(i)$. Suppose for contradiction that such an $i'$ exists. Then $i < i' \leq a(i')$. Then $C_i \cap C_{a(i')}$ is non-empty (since $\ell(i) = \ell(a(i'))$), so by (\ref{clm:CiLaminar}) $C_i \subseteq C_{a(i')}$. Since $j$ is smallest such that $C_i \subseteq C_j$ we have $C_i \subseteq C_j \subseteq C_{a(i')}$. 
Thus $$\ell(i) = C_i \cap S \subseteq C_j \cap S = \ell(j) \subseteq C_{a(i')} \cap S = \ell(i') = \ell(i)\mbox{,}$$
but this contradicts that $\ell(j) \neq \ell(i)$. Thus $j$ is the parent of $i$, proving (\ref{clm:canReach}).

Armed with (\ref{clm:canReach}) we can upper bound the stretch of $g$. Let $v_iv_j$ be an edge of $G$. Without loss of generality $i < j$. By~(\ref{clm:edgeContainment}) we have $C_i \subseteq C_j$, so by (\ref{clm:canReach}), $j = p^t(i)$ for some $t \geq 1$. But $t \leq j-i$, and since $f$ has bandwidth $k$ we have $j-i \leq k$. Thus $d_T(v_i, v_j) = t \leq k$.

Before turning to the $S$-contraction of $g$ we need to develop additional understanding of the last anchor vertex with a given leaf set, as well as upper bound $p(i)-i$ for every $i$. We start with the latter. 

\sta{\label{clm:parentJump}
For every $1 \leq i < n$ it holds that $p(i) \leq i + k$. 
}

Since $G$ is connected and $v_{p(i)} \notin C_i$, there is an edge in $G$ with one endpoint in $C_i$ and the other outside. Select $v_jv_{j'} \in E(G)$  such that $v_j \in C_i$, $v_{j'} \notin C_i$. Then $j \leq i$ and $j' > i$. Furthermore, since $f$ is a bandwidth $k$ layout of $G$ we have that $j' - j \leq k$, and therefore $j' - i \leq j' - j \leq k$. Since $v_j \in C_i \cap C_{j'}$ we get from (\ref{clm:CiLaminar}) that $C_{i} \subseteq C_{j'}$. Thus, by (\ref{clm:canReach}), $j'$ is an ancestor of $i$, and therefore $i < p(i) \leq j' \leq i+k$, proving (\ref{clm:parentJump}). 

We now study the last vertex with a given leaf set. A vertex $v_i$ is a {\em top anchor} if for every $j > i$, $\ell(i) \neq \ell(j)$. 
We first show that every top anchor is an anchor.
%we could have dropped the requirement that $v_i$ is an anchor from the definition of top anchor without changing the set of top anchors. 

\sta{\label{clm:lastIsAnchor}
For every $i$ if for every $j > i$ we have $\ell(i) \neq \ell(j)$, then $v_i$ is an anchor.}
If $v_i$ is not an anchor then $a(i) > i$ and $\ell(i) = \ell(a(i))$, contradicting the assumption that for every $j > i$ we have $\ell(i) \neq \ell(j)$. This proves~(\ref{clm:lastIsAnchor}). 

Next we show that for a top anchor $v_i$ there are many equivalent ways to characterize membership in $C_i$.

\sta{\label{clm:equivalenceInCi}
If $v_i$ is a top anchor, then for every $j \leq n$ the following are equivalent: 
\begin{enumerate}\setlength\itemsep{-5pt}
\item[(i)] $v_j \in C_i$
\item[(ii)] $v_{a(j)} \in C_{i}$
\item[(iii)] $\ell(j) \subseteq \ell(i)$
\item[(iv)] $j=i$ or $v_{p(j)} \in C_i$
\end{enumerate}
}

We first prove $v_j \in C_i \rightarrow v_{a(j)} \in C_i$. If $v_j \in C_i$ and $v_j$ is an anchor then $a(j)=j$ so $a(j) \in C_i$. Otherwise $a(j)$ is the smallest anchor such that $C_j \subseteq C_{a(j)}$. Since $i$ is an anchor such that  $C_j \subseteq C_i$ it follows that $a(j) \leq i$. Since $C_{a(j)} \cap C_i$ is non-empty, (\ref{clm:CiLaminar}) yields that $C_{a(j)} \subseteq C_i$, so $a(j) \in C_i$.

Next we prove $v_{a(j)} \in C_{i} \rightarrow \ell(j) \subseteq \ell(i)$. If $v_{a(j)} \in C_{i}$ then $C_{a(j)} \subseteq C_i$, so $\ell(j) = C_{a(j)} \cap S \subseteq C_i \cap S = \ell(i)$.

Then we prove $\ell(j) \subseteq \ell(i) \rightarrow v_j \in C_i$. Suppose that $\ell(j) \subseteq \ell(i)$. Suppose now for contradiction that $a(j) > i$. 
Both $\ell(a(j))$ and $\ell(i)$ contain $\ell(j)$,
and $a(i) = i < a(j) = a(a(j))$ and so by (\ref{clm:ellIsLaminar}) we have $\ell(i) \subseteq \ell(a(j))$. However this implies that $\ell(j) \subseteq \ell(i) \subseteq \ell(a(j)) = \ell(j)$, so $\ell(i) = \ell(a(j))$ contradicting that $i$ is a {\em top} anchor. Thus $a(j) \leq i$.
Then $C_{a(j)}$ and $C_i$ both contain $\ell(j)$ and so by (\ref{clm:CiLaminar}) we have $C_{a(j)} \subseteq C_i$ and $v_j \in C_{a(j)}$, proving $v_j \in C_i$. This establishes that {\em (i)}, {\em (ii)} and {\em (iii)} are equivalent.

We now show that $v_j \in C_i \rightarrow j=i \mbox{ or } v_{p(j)} \in C_i$. Suppose $v_j \in C_i$ and $j \neq i$. Then $\ell(j) \subseteq \ell(i)$ (since {\em (i)} implies {\em (iii)}). If $\ell(p(j)) = \ell(j)$ then $v_{p(j)} \in C_i$ (since {\em (iii)} implies {\em (i)}). Otherwise $p(j)$ is the smallest integer greater than $j$ such that $C_j \subseteq C_{p(j)}$. Since $C_j \subseteq C_i$ we have that $p(j) \leq i$ and both $C_{p(j)}$ and $C_i$ contain $v_j$, and so by (\ref{clm:CiLaminar}) we have $C_{p(j)} \subseteq C_i$, and hence $v_{p(j)} \in C_i$.

Finally we prove that $j=i \mbox{ or } v_{p(j)} \in C_i \rightarrow v_j \in C_i$. If $j = i$ then $v_j \in C_i$, suppose now that $v_{p(j)} \in C_i$. From {\em (i) $\rightarrow$ (iii)} we obtain that $\ell(p(j)) \subseteq \ell(i)$.
If $\ell(j) = \ell(p(j))$ then  {\em (iii) $\rightarrow$ (i)} yields $v_j \in C_i$. 
Otherwise $p(j)$ is the smallest integer greater than $j$ such that $C_j \subseteq C_{p(j)}$. 
But then $C_j \subseteq C_{p(j)} \subseteq C_i$ and so $v_{j} \in C_i$ as claimed. This completes the proof of (\ref{clm:equivalenceInCi}).

We now upper bound the $S$-contraction of $g$. To that end, let $v_p$ and $v_q$ be vertices in $S$ with $p < q$. Since $p < q$ we have that $\ell(p) \cap \{v_p, v_q\} = \{v_p\}$. Let $i$ be the largest integer such that $\ell(i) \cap \{v_p, v_q\} = \{v_p\}$. 
Then $v_i$ is a top anchor, and by (\ref{clm:lastIsAnchor}) we have that $v_i$ is an anchor. Since $\ell(n) = S$ we have that $i < n$. 
We first lower bound $d_T(v_p, v_q)$, then we upper bound $d_G(v_p, v_q)$ in terms of $p(i)$.
%

%By (\ref{clm:lastIsAnchor}) we have that $i$ is a top anchor. 
By (\ref{clm:equivalenceInCi}, {\em (i)} $\rightarrow$ {\em (iv)}) we have that for every $v_j \in C_i - \{v_i\}$ it holds that $v_{p(j)} \in C_i$. By (\ref{clm:equivalenceInCi}, {\em $\neg$(i)} $\rightarrow$ {\em $\neg$(iv)}) we have that for every $v_j \notin C_i$, $v_{p(j)} \notin C_i$. Thus, $v_iv_{p(i)}$ is the only edge in $T$ that has one endpoint in $C_i$ and the other outside. Since $v_i$ is an anchor we have that $C_i \cap \{v_p, v_q\} = \{v_p\}$. Thus, every path from $v_p$ to $v_q$ in $T$ must use the edge $v_iv_{p(i)}$, and in particular must go through $v_{p(i)}$.
Therefore we have that
$$d_T(v_p, v_q) \geq d_T(v_p, v_{p(i)}) \geq \left\lceil \frac{p(i)-p}{k} \right\rceil \geq \left\lceil \frac{p(i) - |S|}{k}\right\rceil$$
Here $d_T(v_p, v_{p(i)}) \geq \lceil \frac{p(i)-p}{k} \rceil$ follows from~(\ref{clm:parentJump}).

Since $\ell(j) \neq \ell(i)$ for all $j > i$ we have that $C_i \subseteq C_{p(i)}$. Since $v_i$ is an anchor, so is $v_{p(i)}$. Thus $\ell(p(i)) = C_{p(i)} \cap S$. Thus $\ell(i) \subseteq \ell(p(i))$, but $\ell(p(i)) \cap \{v_p, v_q\} \neq \{v_p\}$ since $i$ is the {\em largest} such that $\ell(i) \cap \{v_p, v_q\} = \{v_p\}$. We conclude that $C_{p(i)} \cap \{v_p, v_q\} = \ell(p(i)) \cap \{v_p, v_q\} = \{v_p, v_q\}$.
Therefore $d_G(v_p, v_q) \leq |C_{p(i)}| - 1 \leq p(i) - 1$.

We can now upper bound the $S$-contraction of $g$. If $p(i) \leq 2|S|$ we have that 
$\frac{d_G(v_p, v_q)}{d_T(v_p, v_q)} \leq 2|S|-1$. 
If $p(i) \geq 2|S|$ we have that 
$$\frac{d_G(v_p, v_q)}{d_T(v_p, v_q)} \leq  \frac{p(i)}{(p(i) - |S|)/k} \leq  \frac{p(i)}{p(i)/2k} \leq 2k\mbox{.}$$
Thus the $S$-contraction of $g$ is at most $\max(2|S|-1, 2k)$, as claimed. 
\end{proof}

Next we show how to embed trees with few leaves into the line, while ensuring low distortion for some small prescribed set $S$. 

\begin{lemma}\label{lem:embedFewLeafTree}
There is a polynomial time algorithm that takes as input
a tree $T$ with $\ell$ leaves
and a vertex set $S$ in $T$, 
and outputs an embedding of $T$ into $P_{\mathbb{Z}}$ with congestion at most $4\ell + 2|S|$, stretch at most $24\ell + 12|S|$, and $S$-contraction at most $1$. 
\end{lemma}

\begin{proof}
Let $S'$ be the set consisting of all the leaves of $T$, all the vertices in $S$ and all the vertices of degree at least $3$ in $T$. Since every tree on $\ell$ leaves has at most $\ell-1$ vertices of degree at least $3$ (see, e.g.,~\cite{diestelBook}), $|S'| \leq 2\ell + |S|$. Let ${\cal P}$ be the set of paths in $T$ with both endpoints in $S'$ and no internal vertices in $S'$. Since every vertex $u \notin S'$ has degree exactly $2$, for every vertex $u \notin S'$ there exists precisely one path $P \in {\cal P}$ such that $u \in V(P)$. Further, standard counting arguments on trees show (see e.g.~\cite{diestelBook}) that $|{\cal P}| \leq |S'| - 1 \leq 2\ell + |S|$.

Let $f_{S'} : S' \rightarrow \mathbb{Z}$ be the function obtained by applying Theorem~\ref{thm:lineEmbedding} to the metric $(S', d_T)$. We set $k = 12|S'|$. For each path $P \in {\cal P}$ we proceed as follows. Let $s_1$ and $s_2$ be the endpoints of $P$. Note that $\{s_1, s_2\} \subseteq S'$. We set $a = f_{S'}(s_1)$ and $b = f_{S'}(s_2)$. Note that (by Theorem~\ref{thm:lineEmbedding}) $|a - b| \leq 12|S'| \cdot d_T(s_1, s_2) \leq |E(P)|k$. We apply Lemma~\ref{lem:pathWithFixedEndpoints} to obtain an embedding  $f_P$ of $P$ into $P_{\mathbb{Z}}$ with congestion at most $2$ and stretch at most $k$, such that $f_P(s_1) = a = f_{S'}(s_1)$ and $f_P(s_2) = b = f_{S'}(s_2)$.

Finally we define a function $f : V(T) \rightarrow \mathbb{Z}$: for each vertex $s \in S'$ we set $f(s) = f_{S'}(s)$. For every vertex $u \notin S'$ we select the unique $P \in {\cal P}$ such that $u \in V(P)$ and define $f(u) = f_P(u)$.
The $S$-contraction of $f$ is at most $1$ because for every $u, v \in S$ we have that $|f(u) - f(v)| = |f_{S'}(u) -f_{S'}(v)| \geq d_T(u, v)$.
For every vertex $u \in V(T)$ (including the vertices in $S$) there exists at least one path $P \in {\cal P}$ such that $f(u) = f_P(u)$. For each $P \in {\cal P}$ the congestion is at most $2$. Furthermore $|{\cal P}| \leq 2\ell + |S|$. Thus the congestion of $f$ is at most $2|{\cal P}| \leq 4\ell + 2|S|$. 
For the stretch of $f$ we have that for every edge $uv \in E(T)$ there exists a path $P \in {\cal P}$ such that $\{u,v\} \subseteq V(P)$. Since $f$ agrees with $f_P$ on all vertices of $V(P)$ it follows that $|f(u)-f(v)| = |f_P(u) - f_P(v)| \leq k$. Thus the stretch of $f$ is at most $k = 12|S'| \leq 24\ell + 12|S|$, concluding the proof.     
\end{proof}

We can now put the previous results together and show that every graph with low bandwidth can be embedded into the line with low stretch and congestion while also ensuring that the embedding is a low $S$-stretch and $S$-contraction embedding for any given small set $S$. 
%
%\todo[inline]{Note that $G$ needs to be connected below}

\begin{lemma}\label{lem:simultaneousEmbedding}
There is an algorithm that takes as input a connected graph $G$, a layout $f$ of $G$ of bandwidth at most $k$, and a vertex set $S$, and outputs an embedding $g$ of $G$ into $P_{\mathbb{Z}}$ with congestion at most $10|S|$, stretch at most $120(|S|+1)^2(k+1)$, and $S$-contraction at most $4(k+1)(|S|+1)$.
\end{lemma}

\begin{proof}
Given $G$, $S$, and $f$, the algorithm first uses Lemma~\ref{lem:makeSFirst} to compute a layout
$f_1$ with bandwidth at most $2(k+1)(|S|+1)$ in which $S$ comes first. 
It then applies Lemma~\ref{lem:treeEmbedding} to $G$, $f_1$ and $S$ and obtains a tree $T$ with at most $2|S|$ leaves and an embedding $f_2$ of $G$ into $T$ with
stretch at most $2(k+1)(|S|+1)$, 
congestion $1$, 
and $S$-contraction at most $4(k+1)(|S|+1)$. 
%\todo{needs to be updated with $f_1$'s numbers?}
%
Then it applies Lemma~\ref{lem:embedFewLeafTree} to $T$ and $f_2(S)$ to get an embedding $h$ of $T$ into $P_\mathbb{Z}$ with 
congestion at most $10|S|$, 
stretch at most $60|S|$, and 
$f_2(S)$-contraction at most $1$. 
Finally it defines $g : V(G) \rightarrow \mathbb{Z}$ as $g(u) = h(f_2(u))$ for every $u$ in $V(G)$. By Lemma~\ref{lem:composeEmbeddings} the 
stretch of $g$ is at most $120(|S|+1)^2(k+1)$,
the congestion is at most $10|S|$,
and the $S$-contraction is at most $4(k+1)(|S|+1)$.
\end{proof}

We are now ready to prove the main result of this section, namely Theorem~\ref{lem:lowStretchToEmbedding} \todo{restate}
%A very useful consequence of Theorem~\ref{thm:simultaneousEmbedding} is that when $G$ has low bandwidth we have an abundance of flexibility of how to embed $G$ (with low stretch and congestion) into any graph $H$ whose number of vertices is not much less than $G$. This is encapsulated in the next lemma. 

%\todo[inline]{Theorem~\ref{lem:lowStretchToEmbedding} used to be a lemma, check on invocations.}

%\begin{theorem}\label{lem:lowStretchToEmbedding}
%There exists a polynomial time algorithm that takes as input
%a connected graph $G$, positive integers $k$, $s$, $r$
%such that $r \geq 2$,
%a layout $f$ of $G$ of bandwidth at most $k$,
%a vertex set $S$ in $G$,
%a connected graph $H$ such that $|V(H)| \geq |V(G)|/r$
%and a function $h : S \rightarrow V(H)$ with $S$-stretch at most $s$,
%and outputs an embedding $g$ of $G$ into $H$ with
%stretch at most $480(|S|+1)^3(k+1)^2s$
%and congestion at most $1300(|S|+1)^4(k+1)r$
%such that for every $s \in S$, $g(s) = h(s)$.
%\end{theorem}

\begin{proof}[Proof of Theorem~\ref{lem:lowStretchToEmbedding}]
Given $G$, $k$, $r$, $s$, $f$, $S$, $H$ and $h$ the algorithm proceeds as follows. First it applies the algorithm of Lemma~\ref{lem:simultaneousEmbedding} and computes an embedding $g_1$ of $G$ into $P_{\mathbb{Z}}$ 
with stretch at most $120(|S|+1)^2(k+1)$,
congestion at most $10|S|$, and
$S$-contraction at most $4(k+1)(|S|+1)$.
It then treats $g_1$ as an embedding of $G$ into a path $P$ on $\ell$ vertices, where $\ell \leq n \cdot 120(|S|+1)^2(k+1)$.
Since $g_1$ has finite $S$-contraction, $g_1$ restricted to $S$ is injective\todo{note where S contraction is defined that this means injective, replace finite with well defined}. Rename the vertices in $S$ as $s_1, s_2, \ldots, s_{|S|}$ such that for every $i < |S|$ it holds that $g_1(s_i) < g_1(s_{i+1})$.

Let $P_0$ be the subpath of $P$ that starts with the first vertex of $P$ and ends with $g_1(s_1)$. For every $1 \leq i < |S|$ let $P_i$ be the subpath of $P$ that starts with $g_1(s_i)$ and ends with $g_1(s_{i+1})$. Finally, let $P_{|S|}$ be the subpath of $P$ that starts with $g_1(s_{|S|})$ and ends with the last vertex of $P$.

Since the $S$-contraction of $g_1$ is at most $4(k+1)(|S|+1)$, we have that for every $1 \leq i < |S|$, $|E(P_i)| \geq \frac{d_G(s_i, s_{i+1})}{4(k+1)(|S|+1)}$. Since the $S$-stretch of $h$ is at most $s$ we have that 
$$d_H(h(s_i), h(s_{i+1})) \leq d_G(s_i, s_{i+1}) \cdot s \leq |E(P_i)| \cdot 4(k+1)(|S|+1) \cdot s \mbox{.}$$

For every $1 \leq i < |S|$ we apply Lemma~\ref{lem:pathWithTwoEndpointsGeneral} to obtain an embedding $g_2^i$ of $P_i$ into $H$ such that $g_2^i(g_1(s_i)) = h(s_i)$, $g_2^i(g_1(s_{i+1})) = h(s_{i+1})$, and $g_2^i$ has 
stretch at most $4(k+1)(|S|+1) \cdot s$
and congestion at most $121(|S|+1)^2(k+1) \cdot r$.
We apply Lemma~\ref{lem:pathWithOneEndpointsGeneral} to obtain an embedding $g_2^0$ of $P_0$ into $H$ such that $g_2^0(g_1(s_1)) = h(s_1)$, the stretch of $g_2^0$ is at most $2$ and the congestion is at most $120(|S|+1)^2(k+1) \cdot r$.
Similarly we apply Lemma~\ref{lem:pathWithOneEndpointsGeneral} to obtain an embedding $g_2^{|S|}$ of $P_{|S|}$ into $H$ such that $g_2^{|S|}(g_1(s_{|S|})) = h(s_{|S|})$, the stretch of $g_2^{|S|}$ is at most $2$ and the congestion is at most $120(|S|+1)^2(k+1) \cdot r$.
For the invocations of Lemmas~\ref{lem:pathWithOneEndpointsGeneral} and~\ref{lem:pathWithTwoEndpointsGeneral} above we use the observation that $|V(P)| = \ell \leq n \cdot 120(|S|+1)^2(k+1) \leq |V(H)| \cdot 120(|S|+1)^2(k+1) \cdot r$ (using $|V(H)| \geq n/r$) to guarantee that $H$ has sufficiently many vertices. 

Every vertex of $P$ is on some $P_i$, for every edge $e$ of $P$ there is some $P_i$ such that both endpoints of $e$ are in $P_i$, and for every vertex $v$ of $P$ that appears both on $P_i$ and $P_j$ for $j \neq i$ it holds that $v \in g_1(S)$ and $g_2^i(v) = g_2^j(v)$. 
We define an embedding $g_2$ of $P$ into $H$ as follows: for every $v \in V(P)$ let $i$ be the smallest such that $v \in V(P_i)$. We set $g_2(v) = g_2^i(v)$.
For every edge $uv$ of $P$, there is some $P_j$ containing both endpoints, and since $g_2^i$ and $g_2^j$ agree on shared vertices we have $g_2(u) = g_2^j(u)$ and $g_2(v) = g_2^j(v)$. Hence $d_H(g_2(u), g_2(v))$ is at most the stretch of $g_2^j$, which is at most $4(k+1)(|S|+1) \cdot s$ (since $|S|, k, s \geq 1$ the bound $\max(2,\, 4(k+1)(|S|+1) s)$ collapses). So the stretch of $g_2$ is at most $4(k+1)(|S|+1) \cdot s$. The
congestion of $g_2$ is at most $(|S|+1) \cdot 121(|S|+1)^2(k+1) \cdot r = 121(|S|+1)^3(k+1) \cdot r$, since each vertex of $V(P)$ contributes to exactly one $g_2^i$ (the one with smallest $i$ such that $v \in V(P_i)$), so the congestion at any $z \in V(H)$ is the sum of contributions from the $|S|+1$ sub-embeddings, each bounded by the per-segment congestion above.
Furthermore, for every vertex $s_i \in S$ it holds that $g_2(g_1(s_i)) = h(s_i)$.

Finally we define an embedding $g$ of $G$ into $H$ by setting $g(v) = g_2(g_1(v))$ for every $v \in V(G)$. By definition $g(s_i) = h(s_i)$ for every $s_i \in S$. By Lemma~\ref{lem:composeEmbeddings} the
stretch of $g$ is at most $480(|S|+1)^3(k+1)^2s$
and the congestion is at most $1300(|S|+1)^4(k+1)r$. 
This completes the proof. 
\end{proof}

%\todo[inline]{Looks like the numbers in Lemma~\ref{lem:lowStretchToEmbedding} are now ok, but section 6 and onwards needs updating}

\section{Tree Packings}\label{sec:treepackings}
\todo[inline]{this section should be read carefully}
For a non-negative integer $\tau$, a {\em pathwidth-$\tau$ tree packing} in $G$ is a family ${\cal X}$ of pairwise disjoint connected vertex sets in $G$ such that for every $X \in {\cal X}$, $G[X]$ contains a subtree of pathwidth at least $\tau$. The {\em size} of a pathwidth-$\tau$ tree packing ${\cal X}$ is $|{\cal X}|$. We collect a few observations about pathwidth-$\tau$ tree packings. For a graph $G$, we define $\pack_\tau(G)$ to be the maximum size of a pathwidth-$\tau$ tree packing in $G$. 

\begin{lemma}\label{lem:treePackingMaxUnion}
There is an algorithm that takes as input a connected graph $G$, a non-negative integer $\tau$, and a non-empty pathwidth-$\tau$ tree packing ${\cal X}$ of $G$, runs in polynomial time, and outputs a pathwidth-$\tau$ tree packing ${\cal X}'$ of $G$ with $|{\cal X}'| = |{\cal X}|$ and $\bigcup_{X' \in {\cal X}'} X' = V(G)$.
\end{lemma}

\begin{proof}
The algorithm initializes ${\cal X}' := {\cal X}$ and repeats the following step until $\bigcup_{X' \in {\cal X}'} X' = V(G)$. Since each set in ${\cal X}'$ is non-empty, $G$ is connected, and $\bigcup_{X' \in {\cal X}'} X'$ is a non-empty proper subset of $V(G)$, there exists a vertex $v \in V(G) \setminus \bigcup_{X' \in {\cal X}'} X'$ that has a neighbor in $\bigcup_{X' \in {\cal X}'} X'$. Pick such a $v$ together with a set $X \in {\cal X}'$ such that $v \in N(X)$, and replace ${\cal X}'$ by $({\cal X}' \setminus \{X\}) \cup \{X \cup \{v\}\}$.

Each iteration preserves $|{\cal X}'|$. The new set $X \cup \{v\}$ is connected in $G$ (since $X$ is connected and $v$ has a neighbor in $X$) and $G[X \cup \{v\}]$ contains a subtree of pathwidth at least $\tau$ (the same one contained in $G[X]$). It is disjoint from every other $Y \in {\cal X}' \setminus \{X\}$, since $X$ is and $v \notin \bigcup_{X' \in {\cal X}'} X' \supseteq Y$. Hence ${\cal X}'$ remains a pathwidth-$\tau$ tree packing of $G$.

The quantity $|\bigcup_{X' \in {\cal X}'} X'|$ strictly increases by $1$ at each iteration, so the loop terminates after at most $n$ iterations, at which point $\bigcup_{X' \in {\cal X}'} X' = V(G)$. Each iteration runs in polynomial time.
\end{proof}

\begin{lemma}\label{lem:decidePathwidthSubtree}
There exists an algorithm that, given a graph $G$ and integer $\tau \geq 1$, runs in time $2^{O(9^\tau)} n^{O(1)}$ and determines whether $G$ contains a subtree of pathwidth at least $\tau$.
\end{lemma}

\begin{proof}
Set $N_\tau = \frac{5 \cdot 3^{\tau-1} - 1}{2}$.
The algorithm starts with a preprocessing step that depends only on $\tau$ and not on $G$. Enumerate all non-isomorphic trees on at most $N_\tau$ vertices, and use Theorem~\ref{thm:treePathwidth} to filter to those of pathwidth at least $\tau$. Let ${\cal T}_\tau$ denote the resulting list. By Theorem~\ref{thm:treeCount} we have $|{\cal T}_\tau| \leq 2^{O(N_\tau)} = 2^{O(3^\tau)}$, and by Theorem~\ref{thm:pwTreeObstructionSize} the list ${\cal T}_\tau$ is non-empty.

The main step proceeds as follows. Apply the algorithm of Theorem~\ref{thm:computePw} to $G$ with parameter $k = N_\tau - 2 = \frac{5 \cdot 3^{\tau-1} - 5}{2}$. If the algorithm reports $\textsf{pw}(G) > \frac{5 \cdot 3^{\tau-1} - 5}{2}$, then, by Theorem~\ref{thm:pwObstructions}, $G$ contains a subtree of pathwidth at least $\tau$, so the algorithm outputs \textsf{yes}.

Suppose now that the algorithm of Theorem~\ref{thm:computePw} produces a path decomposition of $G$ of width at most $N_\tau - 2$. Convert this path decomposition to a branch decomposition of $G$ of width at most $w = N_\tau - 1$ in polynomial time, as discussed immediately before Theorem~\ref{thm:hicksMinor}. 
By Lemma~\ref{lem:minorToSubtree} and Lemma~\ref{lem:pwsubgraph}, $G$ contains a subtree of pathwidth at least $\tau$ if and only if $G$ contains a tree of pathwidth at least $\tau$ as a minor; by Theorem~\ref{thm:pwTreeObstructionSize} this is in turn equivalent to $G$ containing some $T' \in {\cal T}_\tau$ as a minor. 
For each $T' \in {\cal T}_\tau$, run the algorithm of Theorem~\ref{thm:hicksMinor} to decide whether $T'$ is a minor of $G$. The algorithm outputs \textsf{yes} if at least one of these calls returns \textsf{yes}, and \textsf{no} otherwise.

We bound the running time of the algorithm. The preprocessing step runs in time $2^{O(3^\tau)}$, independent of $n$. The call to Theorem~\ref{thm:computePw} runs in time $2^{O(N_\tau^2)} n = 2^{O(9^\tau)} n$. The conversion to a branch decomposition is polynomial in $n$. Each call to the algorithm of Theorem~\ref{thm:hicksMinor} runs in time $O(3^{w^2} (h + w - 1)! \cdot |E(G)|)$ with $w \leq N_\tau - 1 = O(3^\tau)$ and $h \leq N_\tau = O(3^\tau)$, which is $2^{O(9^\tau)} n$ (the $3^{w^2} = 2^{O(9^\tau)}$ factor dominates the factorial $(h + w - 1)! = 2^{O(\tau \cdot 3^\tau)}$). The total number of calls to the algorithm of Theorem~\ref{thm:hicksMinor} is $|{\cal T}_\tau| = 2^{O(3^\tau)}$, so the total time spent on such calls is $2^{O(9^\tau)} n$, proving the claimed bound for the running time.
\end{proof}

\begin{lemma}\label{lem:detectPathwidthSubtree}
There exists an algorithm that given a graph $G$ and integer $\tau \geq 1$, runs in time $2^{O(9^\tau)} n^{O(1)}$, and either outputs a subtree of $G$ of pathwidth at least $\tau$ or determines that no such subtree exists.
\end{lemma}

\begin{proof}
The algorithm first invokes the algorithm of Lemma~\ref{lem:decidePathwidthSubtree} on $G$. If the decision algorithm reports that $G$ contains no subtree of pathwidth at least $\tau$, the algorithm reports the same and halts.

Otherwise, the algorithm initializes $H := G$ and repeats the following step. Using the algorithm of Lemma~\ref{lem:decidePathwidthSubtree}, check whether there exists an edge $e \in E(H)$ or a vertex $v \in V(H)$ such that $H - e$ or $H - v$ contains a subtree of pathwidth at least $\tau$. If such an $e$ or $v$ exists, set $H := H - e$ or $H := H - v$ accordingly and repeat; otherwise output $H$.
Each iteration strictly decreases $|V(H)| + |E(H)|$ by at least one, so the loop terminates after at most $|V(G)| + |E(G)|$ iterations. The output graph $H$ contains as a subgraph a tree $T$ of pathwidth at least $\tau$, since this invariant is maintained by the algorithm. $T$ cannot be a proper subgraph of $H$, since any edge of $E(H) \setminus E(T)$ or vertex of $V(H) \setminus V(T)$ would pass the test. Hence, when the algorithm terminates, $H=T$ so $H$ is a tree.  
Each iteration makes at most $|V(G)| + |E(G)|$ invocations of the algorithm of Lemma~\ref{lem:decidePathwidthSubtree}, for an overall running time of $2^{O(9^\tau)} n^{O(1)}$.
\end{proof}

\begin{lemma}\label{lem:ballSeparator}
For every graph $G$, vertex $v$ and positive integer $q$, there exists an integer $q'$ such that $q \leq q' \leq 2q-1$ and $|N(B_{q'}(v))| \leq 4 \cdot \hat{\Delta}_r(G)$.
\end{lemma}

\begin{proof}
Suppose, for contradiction, that $|N(B_{q'}(v))| > 4 \hat{\Delta}_r(G)$ for every integer $q' \in \{q, q+1, \ldots, 2q-1\}$. Since $N(B_{q'}(v)) = B_{q'+1}(v) \setminus B_{q'}(v)$, summing these $q$ strict inequalities gives
\[
|B_{2q}(v)| - |B_q(v)| = \sum_{q' = q}^{2q-1} |N(B_{q'}(v))| > 4q \cdot \hat{\Delta}_r(G).
\]
Hence $|B_{2q}(v)| > 4q \cdot \hat{\Delta}_r(G) + 1$ since $v \in B_q(v)$. However, the definition of radial local density gives $\hat{\Delta}_r(G) \geq \frac{|B_{2q}(v)| - 1}{2 \cdot 2q}$ which yields $|B_{2q}(v)| \leq 4q \cdot \hat{\Delta}_r(G) + 1$, giving the desired contradiction. 
\end{proof}

\begin{lemma}\label{lem:densitySplitIntoBalls}
There is a polynomial time algorithm that takes as input a graph $G$ and a positive integer $q$ and outputs a vertex set $S$ of size at most $8 \cdot n \hat{\Delta}_r(G)/q$ such that every connected component of $G-S$ has radius less than $q$.
\end{lemma}

\begin{proof}
We proceed by induction on $|V(G)|$. If $q = 1$, the algorithm outputs $S := \emptyset$ if $G$ has no edges, and $S := V(G)$ otherwise. In the first case $\hat{\Delta}_r(G) = 0$ and every component of $G - S = G$ is a singleton of radius $0 < 1$. In the second case $\hat{\Delta}_r(G) \geq 1/2$ (from any vertex of degree $\geq 1$), so $|S| = n \leq 8 n \hat{\Delta}_r(G) = 8 n \hat{\Delta}_r(G)/q$, and $G - S$ has no vertices so the radius condition holds vacuously. Henceforth assume $q \geq 2$.

If $G$ has a connected component $C$ of radius less than $q$, the algorithm calls itself on $G - C$. By the inductive hypothesis the recursive call outputs a set $S$ with $|S| \leq 8 |V(G - C)| \hat{\Delta}_r(G - C)/q \leq 8 n \hat{\Delta}_r(G)/q$.
%, where the last inequality uses $|V(G - C)| \leq n$ and the fact that vertex deletion can only shrink balls, hence $\hat{\Delta}_r(G - C) \leq \hat{\Delta}_r(G)$. 
Every component of $(G - C) - S$ has radius less than $q$ by induction, and $C$ itself is a component of $G - S$ with radius less than $q$ by assumption.

Suppose now that every connected component of $G$ has radius at least $q$. Pick an arbitrary vertex $v$. By Lemma~\ref{lem:ballSeparator} applied to $v$ and $\lfloor q/2 \rfloor$ there exists an integer $q'$ such that $\lfloor q/2 \rfloor \leq q' \leq q-1$ and $|N(B_{q'}(v))| \leq 4 \cdot \hat{\Delta}_r(G)$.
Since the connected component of $G$ containing $v$ has radius at least $q$, $|B_r(v)| \geq r+1$ for every $r \leq q$, and in particular $|N[B_{q'}(v)]| = |B_{q'+1}(v)| \geq q'+2 \geq \lfloor q/2 \rfloor + 2 > q/2$.

The algorithm calls itself recursively on $G - N[B_{q'}(v)]$ and obtains a set $S'$ which, by the inductive hypothesis, satisfies
\[
|S'| \leq 8 \cdot (n - |N[B_{q'}(v)]|) \cdot \hat{\Delta}_r(G)/q \leq 8 \cdot (n - q/2) \cdot \hat{\Delta}_r(G)/q = 8 n \hat{\Delta}_r(G)/q - 4 \hat{\Delta}_r(G).
\]
The algorithm outputs $S = S' \cup N(B_{q'}(v))$, and we have
\[
|S| \leq |S'| + 4 \hat{\Delta}_r(G) \leq 8 n \hat{\Delta}_r(G)/q.
\]

Finally, every connected component $C$ of $G - S$ has radius less than $q$: either $C$ is a component of $\bigl(G - N[B_{q'}(v)]\bigr) - S'$, in which case the bound on the radius holds by the inductive hypothesis applied to the recursive call, or $C = B_{q'}(v)$, which has radius $q' \leq q-1 < q$.
\end{proof}

\begin{lemma}\label{lem:computeTreePackingDensity}
There exists an algorithm that given as input a graph $G$ and integers $k \geq 1$ and $\tau \geq 1$ such that $\hat{\Delta}_r(G) \leq k$,
runs in time $2^{O(9^\tau)} n^{O(1)}$,
and outputs a pathwidth-$\tau$ tree packing ${\cal X}$ of $G$, and a set $S \subseteq V(G)$ such that $|S| \leq 36 \cdot k^2 \cdot |{\cal X}|$ and $G - S$ has no subtree of pathwidth at least $\tau$.
\end{lemma}

\begin{proof}
The algorithm sets $S_0 = \emptyset$, ${\cal X} = \emptyset$, $i = 1$, and proceeds as follows.
If $G - S_{i-1}$ does not contain a subtree of pathwidth at least $\tau$ the algorithm outputs $S = S_{i-1}$, ${\cal X}$, and halts.
Otherwise the algorithm sets $r_i$ to be the smallest positive integer such that there exists a vertex $v \in V(G - S_{i-1})$ for which the induced subgraph $G[B_{r_i}^{G - S_{i-1}}(v)]$ contains a subtree of pathwidth at least $\tau$. The algorithm finds $r_i$, together with such a vertex $v_i$ and witness subtree $T_i$, by calling the algorithm of Lemma~\ref{lem:detectPathwidthSubtree} on $G[B_r^{G - S_{i-1}}(v)]$ for every pair $(v, r)$ with $v \in V(G - S_{i-1})$ and $1 \leq r \leq n$.
The algorithm then applies Lemma~\ref{lem:ballSeparator} to $G - S_{i-1}$, $v_i$ and $r_i$ to obtain an integer $r_i'$ with $r_i \leq r_i' \leq 2r_i - 1$ and $|N(B_{r_i'}^{G - S_{i-1}}(v_i))| \leq 4 \hat{\Delta}_r(G - S_{i-1}) \leq 4k$, where the last inequality uses that vertex deletion can only shrink balls, hence $\hat{\Delta}_r(G - S_{i-1}) \leq \hat{\Delta}_r(G) \leq k$.
The algorithm next applies Lemma~\ref{lem:densitySplitIntoBalls} to the induced subgraph $G[B_{r_i'}^{G - S_{i-1}}(v_i)]$ with parameter $r_i$ to obtain a set $S^{\mathrm{ball}}_i \subseteq B_{r_i'}^{G - S_{i-1}}(v_i)$ such that every connected component of $G[B_{r_i'}^{G - S_{i-1}}(v_i)] - S^{\mathrm{ball}}_i$ has radius less than $r_i$, and
\[
|S^{\mathrm{ball}}_i| \leq \frac{8 \cdot |B_{r_i'}^{G - S_{i-1}}(v_i)| \cdot \hat{\Delta}_r(G - S_{i-1})}{r_i} \leq \frac{8 \cdot (2 r_i' k + 1) \cdot k}{r_i} \leq 32 k^2,
\]
where the second inequality uses $|B_r(v)| \leq 2 r \hat{\Delta}_r + 1$ and the third uses $r_i' \leq 2 r_i - 1$ together with $k \geq 1$.
The algorithm sets $S_i = S_{i-1} \cup N(B_{r_i'}^{G - S_{i-1}}(v_i)) \cup S^{\mathrm{ball}}_i$, ${\cal X} = {\cal X} \cup \{T_i\}$, increments $i$, and returns to the beginning of the loop.

By construction we have that at the beginning of iteration $i$, $|S_{i-1}| \leq 36 k^2 (i-1)$ and $|{\cal X}| = i-1$, hence when the algorithm terminates $|S| \leq 36 k^2 \cdot |{\cal X}|$.
Furthermore, when the algorithm terminates $G - S$ has no subtree of pathwidth at least $\tau$. Hence, to prove the statement of the lemma it suffices to show that ${\cal X}$ is a pathwidth-$\tau$ tree packing. Since $T_1, T_2, \ldots, T_i$ are subtrees of $G$ of pathwidth at least $\tau$, it suffices to show that they are pairwise vertex disjoint.

Let $i < i'$ be non-negative integers such that the algorithm did not terminate before the iteration where the tree $T_{i'}$ and set $S_{i'}$ are defined.
By construction $V(T_i) \subseteq B_{r_i'}^{G - S_{i-1}}(v_i)$, and we have $N(B_{r_i'}^{G - S_{i-1}}(v_i)) \cup S^{\mathrm{ball}}_i \subseteq S_i \subseteq S_{i'-1}$. Since $T_{i'}$ is disjoint from $S_{i'-1}$, $T_{i'}$ is in particular disjoint from $N(B_{r_i'}^{G - S_{i-1}}(v_i))$. Since $T_{i'}$ is connected, it follows that $T_{i'}$ is either entirely contained in $B_{r_i'}^{G - S_{i-1}}(v_i)$ or entirely disjoint from $B_{r_i'}^{G - S_{i-1}}(v_i)$. In the latter case $T_{i'}$ is disjoint from $V(T_i) \subseteq B_{r_i'}^{G - S_{i-1}}(v_i)$ as desired.

It remains to rule out the former case. Suppose for contradiction that $T_{i'} \subseteq B_{r_i'}^{G - S_{i-1}}(v_i)$. Since $T_{i'}$ is also disjoint from $S^{\mathrm{ball}}_i \subseteq S_{i'-1}$, $T_{i'}$ is contained in some connected component $C$ of $G[B_{r_i'}^{G - S_{i-1}}(v_i)] - S^{\mathrm{ball}}_i$. By Lemma~\ref{lem:densitySplitIntoBalls} the component $C$ has radius less than $r_i$, so $C \subseteq B_{r_i - 1}^{C}(u)$ for some $u \in C$. Since $G[B_{r_i'}^{G - S_{i-1}}(v_i)] - S^{\mathrm{ball}}_i$ is a subgraph of $G - S_{i-1}$, distances in the former are at least those in the latter, and therefore $C \subseteq B_{r_i - 1}^{G - S_{i-1}}(u)$. But then $G[B_{r_i - 1}^{G - S_{i-1}}(u)]$ contains $T_{i'}$ as a subtree, which has pathwidth at least $\tau$, contradicting the choice of $r_i$ if $r_i \geq 2$, or, if $r_i = 1$, contradicting that $T_{i'}$ has at least two vertices (since its pathwidth is at least $\tau \geq 1$). We conclude that $T_i$ and $T_{i'}$ are vertex disjoint, and hence that ${\cal X}$ is a pathwidth-$\tau$ tree packing such that $|S| \leq 36 k^2 \cdot |{\cal X}|$.

To upper bound the running time, observe that the algorithm terminates after at most $n$ iterations of the outer loop, since the trees $T_1, T_2, \ldots$ are pairwise vertex disjoint and each contains at least one vertex of $G$. In each iteration the algorithm makes at most $n^2$ calls to the algorithm of Lemma~\ref{lem:detectPathwidthSubtree} (one per pair $(v, r)$ with $v \in V(G - S_{i-1})$ and $1 \leq r \leq n$). Each such call takes at most $2^{O(9^\tau)} n^{O(1)}$ time, and the per-iteration applications of Lemmas~\ref{lem:ballSeparator} and~\ref{lem:densitySplitIntoBalls} take polynomial time. Combined, this gives the claimed upper bound on the running time.
\end{proof}

Next we explore the structure of a pathwidth-$\tau$ tree packing ${\cal X}$ such that $\bigcup_{X \in {\cal X}} X = V(G)$ when $G$ has no subtree of pathwidth at least $\tau+1$.
Let $G$ be a graph, $\tau$ be a non-negative integer, and  ${\cal X}$ be a pathwidth-$\tau$ tree packing such that $\bigcup_{X \in {\cal X}} X = V(G)$.
The {\em tree-packing graph} of ${\cal X}$ is the graph $H$ with a vertex $v_X$ for every $X \in {\cal X}$ and an edge $v_Xv_{X'}$ for every pair $X$, $X' \in {\cal X}$ such that there is an edge in $G$ with one endpoint in $X$ and the other in $X'$.

\begin{lemma}\label{lem:treePackingStructure}
Let $G$ be a graph, $\tau$ be a non-negative integer such that $G$ has no subtree of pathwidth at least $\tau+1$, and  ${\cal X}$ be a pathwidth-$\tau$ tree packing such that
$\bigcup_{X \in {\cal X}} X = V(G)$.
The tree-packing graph $H$ of ${\cal X}$ is a path or a cycle. 
\end{lemma}

\begin{proof}
Since $G$ is connected and ${\cal X}$ forms a partition of $V(G)$, $H$ is connected as well. 
Thus it suffices to show that $H$ has maximum degree $2$. Aiming towards a contradiction, let $v_X$ be a vertex in $H$ with three distinct neighbors $v_A$, $v_B$, $v_C$ for sets $X,A,B,C \in {\cal X}$. 
Then $G[A]$, $G[B]$ and $G[C]$ contain spanning trees $T_A$, $T_B$, $T_C$ respectively, each of pathwidth at least $\tau$. 
Let $T_X$ be a spanning tree of $G[X]$.
Since $v_X$ is adjacent to $v_A$, $v_B$, and $v_C$ in $H$ we have that each of $X \cap N_G(A)$, $X \cap N_G(B)$ and $X \cap N_G(C)$ are all non-empty. 
Make a spanning tree $T$ of $G[X \cup A \cup B \cup C]$ from starting with $T_X$, $T_A$, $T_B$ and $T_C$ and adding an edge of $G$ between $X$ and $A$, $X$ and $B$, and $X$ and $C$ respectively. Such edges exist because $v_X$ is adjacent to $v_A$, $v_B$, and $v_C$ in $H$.
Let $P_{AB}$ be the unique shortest path in $T$ with one endpoint in $A$ and the other in $B$.
Let $P_{AC}$ be the unique shortest path in $T$ with one endpoint in $A$ and the other in $C$.
Let $v$ be the vertex furthest from $A$ on $P_{AB}$ such that $v \in V(P_{AC})$.
Then $v \in X$ (since both $P_{AB}$ and $P_{AC}$ contain unique neighbor of $A$ in $T$). 
Furthermore, let $C_A$, $C_B$ and $C_C$ be the connected components of $T - v$ containing $A$, $B$ and $C$ respectively. 
Then $C_A$, $C_B$ and $C_C$ are all distinct (since the unique path in $T$ between each pair from $\{A,B,C\}$ contains $v$).
But then $T$ is a subtree of $G$ that contains a vertex $v$ such that $T-v$ has tree distinct components $C_A$, $C_B$ and $C_C$, each of pathwidth at least $\tau$.
By Theorem~\ref{thm:treePathwidth} it follows that $\pw(T) \geq \tau + 1$, contradicting that $G$ does not have a subtree with pathwidth at least $\tau+1$. 
\end{proof}

We are now ready to prove the main result of this section.

\begin{lemma}\label{lem:treePackingPathCycleDensity}
There is an algorithm that, given positive integers $d \geq 2$, $\tau \geq 1$, and a connected graph $G$ on $n$ vertices of radial local density at most $d$ such that $G$ has no subtree of pathwidth at least $\tau+1$, runs in time $2^{O(9^\tau)} \cdot n^{O(1)}$ and either concludes that $G$ has no subtree of pathwidth at least $\tau$ or outputs
a pathwidth-$\tau$ tree packing ${\cal X} = \{X_1, \dots, X_t\}$ of $G$ with $\bigcup_{i=1}^{t} X_i = V(G)$, the tree-packing graph $H$ of ${\cal X}$ (which is a path or a cycle on vertex set $\{1, \dots, t\}$), and, for every subpath $Q$ of $H$, a vertex set $S_Q \subseteq V(G)$ such that $|S_Q| \leq 36 d^2 \cdot |V(Q)|$, and $G\!\left[\bigcup_{i \in V(Q)} X_i\right] - S_Q$ has no subtree of pathwidth at least $\tau$.
\end{lemma}

\begin{proof}
The algorithm invokes Lemma~\ref{lem:computeTreePackingDensity} on $G$ with parameters $(G, d, \tau)$.
%, which is valid since $\hat{\Delta}_r(G) \leq d$ and $d, \tau \geq 1$.
%
Let $({\cal X}_0, S_0)$ be the output, where ${\cal X}_0$ is a pathwidth-$\tau$ tree packing of $G$ and $S_0 \subseteq V(G)$ satisfies $|S_0| \leq 36 d^2 \cdot |{\cal X}_0|$ and is such that $G - S_0$ has no subtree of pathwidth at least $\tau$.
If ${\cal X}_0 = \emptyset$ then $S_0 = \emptyset$ and $G$ has no subtree of pathwidth at least $\tau$; in this case the algorithm outputs this conclusion and halts. We therefore assume below that ${\cal X}_0 \neq \emptyset$.

The algorithm now applies Lemma~\ref{lem:treePackingMaxUnion} to ${\cal X}_0$ to obtain a pathwidth-$\tau$ tree packing ${\cal X}$ of $G$ with $|{\cal X}| = |{\cal X}_0|$ and $\bigcup_{X \in {\cal X}} X = V(G)$, and computes the tree-packing graph $H$ of ${\cal X}$. By Lemma~\ref{lem:treePackingStructure}, $H$ is a path or a cycle; we identify $V(H)$ with $\{1, \dots, t\}$ where $t := |{\cal X}|$, and accordingly write ${\cal X} = \{X_1, \dots, X_t\}$.

We now describe the main loop. For each subpath $Q$ of $H$ the algorithm defines $G_Q = G\!\left[\bigcup_{i \in V(Q)} X_i\right]$. Since $G_Q$ is an induced subgraph of $G$, monotonicity of $\hat{\Delta}_r$ under vertex deletion yields $\hat{\Delta}_r(G_Q) \leq \hat{\Delta}_r(G) \leq d$, and $G_Q$ has no subtree of pathwidth at least $\tau+1$. The algorithm invokes Lemma~\ref{lem:computeTreePackingDensity} on $(G_Q, d, \tau)$, and obtains a pathwidth-$\tau$ tree packing ${\cal X}_Q$ and vertex set $S_Q$, such that $|S_Q| \leq 36 d^2 \cdot |{\cal X}_Q|$, and $G_Q - S_Q$ has no subtree of pathwidth at least $\tau$. We distinguish two cases.

Suppose first that $|{\cal X}_Q| \leq |V(Q)|$ for every subpath $Q$ of $H$. Then, for every $Q$ we have that
$$|S_Q| \;\leq\; 36 d^2 \cdot |{\cal X}_Q| \;\leq\; 36 d^2 \cdot |V(Q)|.$$
Thus the algorithm outputs ${\cal X}$, $H$, and the family $\{S_Q : Q \text{ a subpath of } H\}$ and halts.

Suppose now that there exists a subpath $Q$ of $H$ such that $|{\cal X}_Q| > |V(Q)|$. The algorithm defines
\[
\widetilde{\cal X} \;:=\; \bigl({\cal X} \setminus \{X_i : i \in V(Q)\}\bigr) \cup {\cal X}_Q,
\]
which is a pathwidth-$\tau$ tree packing of $G$ with $|\widetilde{\cal X}| > |{\cal X}|$. The algorithm now re-applies Lemma~\ref{lem:treePackingMaxUnion} to $\widetilde{\cal X}$ to obtain a packing ${\cal X}'$ such that $|{\cal X}'| = |\widetilde{\cal X}|$ and $\bigcup_{X' \in {\cal X}'} X' = V(G)$, and computes the tree-packing graph $H'$ of ${\cal X}'$. We have that $|{\cal X}'| > |{\cal X}|$ and that, by Lemma~\ref{lem:treePackingStructure}, $H'$ is a path or a cycle. The algorithm updates ${\cal X}$ to ${\cal X}'$, $H$ to $H'$ and re-starts the main loop.

The main loop terminates after at most $n$ iterations because each iteration increases $|{\cal X}|$ by at least one, and $|{\cal X}| \leq |V(G)| = n$. Each iteration of the main loop computes ${\cal X}_Q$ and $S_Q$ for at most $|V(H)|^2 \leq n^2$ subpaths $Q$ of $H$ via a call to Lemma~\ref{lem:computeTreePackingDensity}, costing $2^{O(9^\tau)} \cdot n^{O(1)}$. The applications of Lemma~\ref{lem:treePackingMaxUnion} and computations of the tree-packing graphs take polynomial time. Therefore the total running time is $2^{O(9^\tau)} \cdot n^{O(1)}$, as claimed.
\end{proof}

\section{Weak Pre-Layouts}\label{sec:weakPreLayouts}
We now describe how we use Lemma~\ref{lem:treePackingPathCycleDensity} to get a decomposition of the input graph $G$. To describe the structure obtained by the decompositions we  define {\em weak pre-layouts}. An enhanced version of weak pre-layouts, namely {\em pre-layouts}, will be defined in the next section. 
Let $s$, $\tau$ and $\rho$ be positive integers. 
A {\em weak pre-layout} of a graph $G$ 
with {\em separator size} $s$,
and {\em pathwidth-$\tau$ tree hitting bound} $\rho$
is a seven-tuple $(S, R, \ell, P, g, \zeta, W)$ where $\ell \geq 3$ is an integer, $S$, $R$, and $W$ are vertex sets in $G$ with $S$ and $R$ disjoint,
$P = p_0s_1p_1s_2p_2\ldots s_\ell p_\ell$ is a path on $2\ell+1$ vertices,
$g$ is an embedding of $G - (S \cup R)$ into $P$,
and $\zeta : \{1, \ldots, \ell\} \rightarrow 2^{V(G)}$ is a function such that the following conditions are satisfied. 

\begin{enumerate}\setlength\itemsep{-2pt}
\item\label{itm:SBound} $|S| \leq s$,
\item\label{itm:RNeigh} $N(R) \subseteq S$,
\item\label{itm:stretchOne} the stretch of $g$ is at most $1$,
\item\label{itm:connectedEnds} $G[g^{-1}(\{p_0,s_1\})]$ and $G[g^{-1}(\{s_\ell,p_\ell\})]$ are connected,
\item\label{itm:connectedMid} for every $1 \leq i < \ell$, $G[g^{-1}(\{s_i,p_i,s_{i+1}\})]$ is connected,
\item\label{itm:siSize} for every $1 \leq i \leq \ell$ we have that $|g^{-1}(s_i)| \leq s$,  
\item\label{itm:smallConnectedSeparator} for every $1 \leq i < \ell$, $\zeta(i)$ is a connected subset of $g^{-1}(p_i)$ that separates $g^{-1}(s_i)$ from $g^{-1}(s_{i+1})$ in $G' = G - (S \cup R)$, such that $|\zeta(i)| \leq d_{G'}(g^{-1}(s_i), g^{-1}(s_{i+1}))/2$.
\item\label{itm:rBound} $|W \cap R| \leq \rho$, and $G[R] - W$ has no subtree of pathwidth at least $\tau$,
\item\label{itm:hitPi} for every $0 \leq i \leq \ell$ we have that $|W \cap g^{-1}(p_i)| \leq \rho$ and $G[g^{-1}(p_i)] - W$ has no subtree of pathwidth at least $\tau$.

%$|\zeta(i)| <  d_{G'}(g^{-1}(s_i), g^{-1}(s_{i+1}))/2$.
%\item for every $1 \leq i < \ell$ there is a vertex $c_i \in g^{-1}(p_i)$ such that $$d_{G'}(g^{-1}(s_i), c_i) + d_{G'}(c_i, g^{-1}(s_{i+1})) \leq 2 \cdot d_{G'}(g^{-1}(s_i), g^{-1}(s_{i+1}))$$
\end{enumerate}

\begin{figure}[ht]
    \centering
    \includegraphics[width=0.9\textwidth]{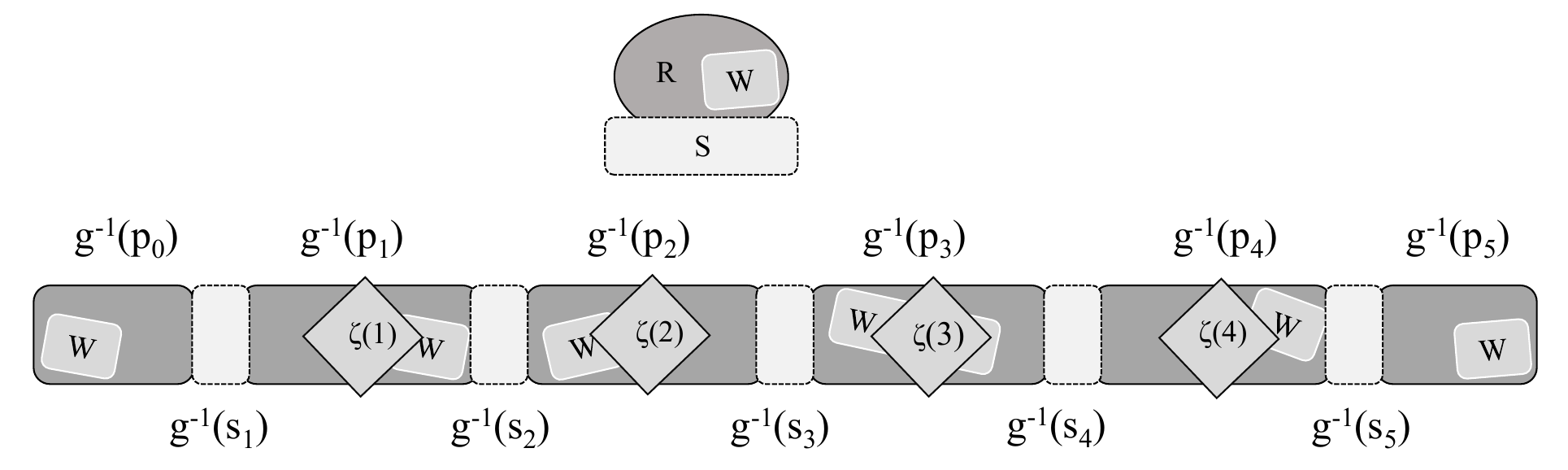}
    \caption{\em An illustration of a weak pre-layout with $\ell = 5$}
    \label{fig:weakPreLayout}
\end{figure}

{\bf Parts of a Weak Pre-Layout.} For a weak pre-layout $(S, R, \ell, P, g, \zeta, W)$ of a graph $G$, the induced subgraphs
$G[g^{-1}(\{p_0, s_1\})]$,
$G[g^{-1}(\{s_\ell, p_\ell\})]$
and
$G[g^{-1}(\{s_i,p_i,s_{i+1}\})]$ for $1 \leq i < \ell$ are called the {\em parts} of the pre-layout.
$G[g^{-1}(\{p_0, s_1\})]$ is part number $0$, $G[g^{-1}(\{s_\ell, p_\ell\})]$ is part number $\ell$, and for every $1 \leq i < \ell$, $G[g^{-1}(\{s_i,p_i,s_{i+1}\})]$ is part number $i$.
Part $0$ and $\ell$ are called {\em end parts}, and the remaining parts are called the {\em inner parts}. For each $1 \leq i < \ell$ we will refer to $\zeta(i)$ as the {\em core} of the inner part $i$. We will denote part $i$ of the weak pre-layout by $G_i$.
%\todo[inline]{make sure all weak pre-layouts and pre-layout references contain $\ell$ and $W$}
%\todo[inline]{DL: construct W and verify (iii) where (iii) mentioned first time? Or rename items 3 and 8 to come last in list.}

The main result of this section is an algorithm that extracts a weak pre-layout from a tree packing as guaranteed by Lemma~\ref{lem:treePackingPathCycleDensity}. Towards this goal we will need the following simple separation result about graphs with low radial local density and an embedding into a path. 

\begin{lemma}\label{lem:localDensitySeparator}
There exists a polynomial time algorithm that takes as input a graph $G$ with at least one edge, a path $P$ and an embedding $h$ of $G$ into $P = p_1 p_2 \ldots p_\ell$ with stretch at most $1$ such that
for every vertex $p_i \in V(P)$ the set $h^{-1}(p_i)$ is non-empty and induces a connected subgraph of $G$,
and $|V(P)| \geq 4 \hat{\Delta}_r(G) + 1$, and outputs a vertex set $S$ of size at most $4 \cdot \hat{\Delta}_r(G)$, disjoint from $h^{-1}( \{p_1, p_\ell \})$ such that $S$ separates $h^{-1}(p_1)$ from $h^{-1}(p_\ell)$.
\end{lemma}

\begin{proof}
Suppose for contradiction that such a separator does not exist. 
Then, by Menger's Theorem~\cite{diestelBook} there are $t \geq 4 \cdot \hat{\Delta}_r(G) + 1 \geq 3$ paths $R_1, \ldots, R_t$ such that each path $R_i$ has its first vertex in
$h^{-1}(p_1)$, its last vertex in $h^{-1}(p_\ell)$, and the paths $R_1, \ldots, R_t$ are internally vertex disjoint.
Here we used that $\hat{\Delta}_r(G) \geq 1/2$ because $G$ has at least one edge.
Since the embedding $h$ has stretch $1$ it holds that $R_i \cap h^{-1}(p_j) \neq \emptyset $ for every $i$, $j$.
Let $p_i$ be such that $|h^{-1}(p_i)|$ is smallest, and set $q = |h^{-1}(p_i)|$.
Without loss of generality $R_1$ is shortest among all the paths  $R_1, \ldots, R_t$.
Let $u$ be a vertex in $V(R_1) \cap  h^{-1}(p_i)$. Since $u$ is well defined, $q \geq 1$.

If $|R_1| \geq q+2$ then each $R_i$ has at least $q$ internal vertices. A ball of radius $q-1$ around $u$ contains all of $h^{-1}(p_i)$, and hence it intersects every path in $R_1, \ldots, R_t$. Therefore a ball of radius $2q-1$ around $u$ contains at least $q$ internal vertices from each path in $R_1, \ldots, R_t$. Thus the radial local density of $G$ is at least
$(tq-1)/(4q-2) > t/4 >  \hat{\Delta}_r(G)$, a contradiction.

If $|R_1| \leq q+1$ then a ball of radius $q$ around $u$ contains all of $R_1$, and therefore at least one vertex from $h^{-1}(p_j)$ for every vertex $p_j$ on $P$. Since each set  $h^{-1}(p_j)$ has at least $q$ vertices, a ball of radius $2q-1$ around $u$ contains at least $q$ vertices from each set $h^{-1}(p_j)$. It follows that the radial local density of $G$ is at least $(\ell \cdot q - 1)/(4q-2) \geq \ell/4 > \hat{\Delta}_r(G)$, a contradiction.
Thus a separator $S$ as described in the statement exists. Finding a minimum size separator can be done in polynomial time~\cite{cormen2022introduction}.
\end{proof}

\begin{lemma}\label{lem:weakPreLayoutDensity}
There exists an algorithm that takes as input
positive integers $d \geq 2$, $\tau$,
and a connected graph $G$ on $n$ vertices of radial local density at most $d$
such that $G$ has no subtree of pathwidth at least $\tau+1$,
runs in time $2^{O(9^\tau)} \cdot n^{O(1)}$ and either outputs a vertex set $S \subseteq V(G)$ of size at most $1566 \cdot d^3$ such that $G - S$ has no subtree of pathwidth at least $\tau$, or outputs a weak pre-layout $(S, R, \ell, P, g, \zeta, W)$ of $G$ with separator size at most $4d$ and pathwidth-$\tau$ tree hitting bound $792 \cdot d^3$.
\end{lemma}

\begin{proof}
The algorithm first applies the algorithm of Lemma~\ref{lem:treePackingPathCycleDensity} on $(d, \tau, G)$.
If Lemma~\ref{lem:treePackingPathCycleDensity} reports that $G$ has no subtree of pathwidth at least $\tau$, the algorithm outputs $S := \emptyset$ and halts. Otherwise let ${\cal X} = \{X_1, \dots, X_r\}$, $H$, and $\{S_Q : Q \text{ a subpath of } H\}$ denote the output, and define $h : V(G) \rightarrow \{1, \ldots, r\}$ by $h(v) = i$ if $v \in X_i$. By Lemma~\ref{lem:treePackingPathCycleDensity} $H$ is a path or a cycle; without loss of generality the vertices of $H$ appear in the order $\{1, \dots, r\}$ on the path (cycle).

If $r \leq 40d+7$, let $Q^*$ be a subpath of $H$ with $V(Q^*) = V(H)$. The algorithm outputs $S := S_{Q^*}$ and halts; the size bound
\[
|S_{Q^*}| \;\leq\; 36 d^2 \cdot r \;\leq\; 36 d^2 \cdot (40d+7) \;=\; 1440 d^3 + 252 d^2 \;\leq\; 1566 \cdot d^3
\]
(using $252 d^2 \leq 126 d^3$ for $d \geq 2$) is the size promised in the lemma statement. We therefore assume below that $r > 40d+7$.

Let $\ell = \lfloor r/(10d+2) \rfloor - 1$ and  $\gamma$ be the remainder of $r$ when divided by $10d+2$.
Since  $r \geq 40d+8$, we have that $\ell \geq 3$.
Partition $V(H)$ into $J_0,P_0,J_1,P_1,\ldots,J_\ell,P_\ell$ such that each $H[P_i]$ is a path on $6d+1$ vertices, each $H[J_i]$ is a path on 
$4d+1  + \lfloor \gamma/(\ell+1)\rfloor$ or $4d+1  + \lceil \gamma/(\ell+1)\rceil$ vertices.
In particular the first endpoint of $H[J_0]$ is adjacent to the last endpoint of $H[P_\ell]$ if and only if $H$ is a cycle. 
Since $\ell \geq 3$ each path $H[J_i]$ is a path on at least $4d+1$ and at most $6.5d + 2 \leq 7d + 1$ vertices. 

%%%%%%%%%%%%%%%%%%%%%%%%%%%%%%%%
% THIS WAS MISSING, GOT ERASED IN EDITING. 
For every $J_i$, let $a_i$ and $b_i$ be the first and last vertex of $H[J_i]$.
Apply Lemma~\ref{lem:localDensitySeparator} to
the graph $G[h^{-1}(J_i)]$, path $H[J_i]$ and the restriction of $h$ to $J_i$
to obtain a vertex set $Z_i \subseteq h^{-1}(J_i \setminus \{a_i, b_i\})$ such that $|Z_i| \leq 4d$ and $Z_i$ separates 
$h^{-1}(a_i)$ from $h^{-1}(b_i)$ in  $G[h^{-1}(J_i)]$.
While there exists a vertex $z \in Z_i$ that
does not have a neighbor in the component of $G[h^{-1}(J_i)] - Z_i$ that contains $h^{-1}(a_i)$, or
does not have a neighbor in the component of $G[h^{-1}(J_i)] - Z_i$ that contains $h^{-1}(b_i)$,
we remove $z$ from $Z_i$ and observe that $Z_i$ still separates $h^{-1}(a_i)$ from $h^{-1}(b_i)$ in $G[h^{-1}(J_i)]$ and satisfies $Z_i \subseteq h^{-1}(J_i \setminus \{a_i, b_i\})$ and $|Z_i| \leq 4d$.
When this loop terminates every vertex $z \in Z_i$ has a neighbor both in the component of $G[h^{-1}(J_i)] - Z_i$ that contains $h^{-1}(a_i)$ and in the component that contains $h^{-1}(b_i)$.
%%%%%%%%%%%%%%%%%%%%%%%%%%%%%%%%

For every $J_i$ define $A_i$ to be the connected component of $G[h^{-1}(J_i)] - Z_i$ containing $h^{-1}(a_i)$, $B_i$ to be the connected component of $G[h^{-1}(J_i)] - Z_i$ containing $h^{-1}(b_i)$, and $R_i$ to be the union of all other connected components of $G[h^{-1}(J_i)] - Z_i$. Then $A_i, B_i, R_i, Z_i$ partition $h^{-1}(J_i)$, and since $\partial_G(h^{-1}(J_i)) \subseteq h^{-1}(a_i) \cup h^{-1}(b_i)$ we have $N(R_i) \subseteq Z_i$.

If $H$ is a path, set $S := \emptyset$ and $R := \emptyset$; if $H$ is a cycle, set $S := Z_0$ and $R := R_0$. In both cases $|S| \leq 4d$, verifying (\ref{itm:SBound}); and $N(R) \subseteq S$, verifying (\ref{itm:RNeigh}). 
%Item (\ref{itm:rBound}) will be verified in the construction of $W$ below.

The algorithm now defines the path $P = p_0 s_1 p_1 s_2 p_2 \ldots s_\ell p_\ell$ and the embedding $g$ of $G - (S \cup R)$ into $P$ as follows. For every $1 \leq i \leq \ell$ and $v \in V(G) - (S \cup R)$: if $v \in Z_i$ then $g(v) = s_i$; if $v \in h^{-1}(P_i) \cup B_i \cup R_i$ then $g(v) = p_i$; if $v \in A_i$ then $g(v) = p_{i-1}$. Furthermore, $g(v) = p_0$ for every $v \in h^{-1}(P_0)$. If $H$ is a path, $g(v) = p_0$ for every $v \in h^{-1}(J_0)$; if $H$ is a cycle, $g(v) = p_0$ for every $v \in B_0$ and $g(v) = p_\ell$ for every $v \in A_0$.

Thus $g$ is an embedding of $G - (S \cup R)$ into $P$ because $\{h^{-1}(P_i), A_i, B_i, Z_i, R_i : 0 \leq i \leq \ell\}$ partitions $V(G)$.
To verify (\ref{itm:stretchOne}), namely that the stretch of $g$ is at most $1$, it suffices to observe each of the following statements, each of which follows directly from the definition of the respective sets:
\begin{itemize}\setlength\itemsep{-4pt}
\item $N(R_0) \subseteq Z_0$,
\item $N(Z_0) \subseteq A_0 \cup B_0 \cup R_0$,
\item $N(B_0) \subseteq Z_0 \cup h^{-1}(P_0)$,
\item $N(A_0) \subseteq Z_0$ if $H$ is a path, and $N(A_0) \subseteq Z_0 \cup h^{-1}(P_\ell)$ if $H$ is a cycle;
\item for every $0 \leq i < \ell$, $N(h^{-1}(P_i)) \subseteq B_{i} \cup A_{i+1}$, and $N(h^{-1}(P_\ell)) \subseteq B_\ell$ if $H$ is a path, $N(h^{-1}(P_\ell)) \subseteq B_\ell \cup A_0$ if $H$ is a cycle;
\item for every $1 \leq i \leq \ell$, $N(R_i) \subseteq Z_i$, $N(Z_i) \subseteq A_i \cup B_i \cup R_i$, $N(B_i) \subseteq Z_i \cup h^{-1}(P_i)$, and $N(A_i) \subseteq Z_i \cup h^{-1}(P_{i-1})$.
\end{itemize}

For (\ref{itm:connectedEnds}): if $H$ is a path, then $g^{-1}(\{p_0,s_1\}) = h^{-1}(J_0 \cup P_0) \cup A_1 \cup Z_1$ and $g^{-1}(\{p_\ell,s_\ell\}) = h^{-1}(P_\ell) \cup B_\ell \cup R_\ell \cup Z_\ell$, each of which is connected; if $H$ is a cycle, then $g^{-1}(\{p_0,s_1\}) = B_0 \cup h^{-1}(P_0) \cup A_1 \cup Z_1$ and $g^{-1}(\{p_\ell,s_\ell\}) = A_0 \cup h^{-1}(P_\ell) \cup B_\ell \cup R_\ell \cup Z_\ell$, again each connected.

For (\ref{itm:connectedMid}): for every $1 \leq i < \ell$,
\[
g^{-1}(\{s_i,p_i,s_{i+1}\}) \;=\; Z_i \cup R_i \cup B_i \cup h^{-1}(P_i) \cup A_{i+1} \cup Z_{i+1},
\]
which is connected because $B_i \cup h^{-1}(P_i) \cup A_{i+1}$ is connected, every vertex of $Z_i$ has a neighbor in $B_i$, every component of $R_i$ has a neighbor in $Z_i$, and every vertex of $Z_{i+1}$ has a neighbor in $A_{i+1}$.
For (\ref{itm:siSize}), $g^{-1}(s_i) = Z_i$ and $|Z_i| \leq 4d$. 
%Item (\ref{itm:hitPi}) will be verified in the construction of $W$ below.

Finally, for every $1 \leq i < \ell$ set $\zeta(i) := h^{-1}(v)$, where $v$ is a vertex in $P_i$ minimizing $|h^{-1}(v)|$. We verify (\ref{itm:smallConnectedSeparator}). For every $u \in V(P_i)$, $h^{-1}(u) \subseteq g^{-1}(p_i)$, and $h^{-1}(u)$ separates $g^{-1}(s_i)$ from $g^{-1}(s_{i+1})$ in $G' = G - (S \cup R)$; the connectivity and separation requirements for $\zeta(i)$ follow. It remains to bound $|\zeta(i)|$.

Let $z := d_{G'}(g^{-1}(s_i), g^{-1}(s_{i+1}))$, and suppose for contradiction that $|h^{-1}(u)| > z/2$ for every $u \in V(P_i)$. Let $w$ be a vertex on a shortest path from $g^{-1}(s_i)$ to $g^{-1}(s_{i+1})$ in $G'$. For every $u \in V(P_i)$, this shortest path intersects $h^{-1}(u)$ at some vertex $v_u$ at distance at most $z$ from $w$; and since $h^{-1}(u)$ is a connected set in $G'$ on more than $z/2$ vertices, more than $z/2$ vertices of $h^{-1}(u)$ lie within distance $z/2$ of $v_u$, hence within distance $3z/2$ of $w$. Thus the ball of radius $3z/2$ around $w$ in $G'$ contains more than $(z/2)\cdot|V(P_i)|$ vertices, so the radial local density of $G$ is at least
$(z/2)|V(P_i)| / (2 \cdot 3z/2) = (6d+1)/6 > d$,
contradicting the hypothesis on the radial local density of $G$.

It remains to construct the hitting set $W$ and verify items (\ref{itm:rBound}) and (\ref{itm:hitPi}). For each $i \in \{0, 1, \ldots, \ell\}$ let $Q_i$ denote the subpath of $H$ with $V(Q_i) = J_i \cup P_i \cup J_{i+1}$ for $1 \leq i < \ell$; $V(Q_0) = J_0 \cup P_0$ if $H$ is a path and $V(Q_0) = J_0 \cup P_0 \cup J_1$ if $H$ is a cycle; and $V(Q_\ell) = P_\ell$ if $H$ is a path and $V(Q_\ell) = J_\ell \cup P_\ell \cup J_0$ if $H$ is a cycle. Let $Q_R$ denote the subpath of $H$ with $V(Q_R) = J_0$. In every case $|V(Q_i)|, |V(Q_R)| \leq 2(7d+1) + (6d+1) = 20d + 3 \leq 22d$ (using $d \geq 2$). Define
\[
W \;:=\; \bigcup_{i=0}^{\ell} \bigl(S_{Q_i} \cap g^{-1}(p_i)\bigr) \;\cup\; \bigl(S_{Q_R} \cap R\bigr).
\]

For (\ref{itm:rBound}): since $R$ is disjoint from $g^{-1}(p_i)$ for every $i$, $W \cap R = S_{Q_R} \cap R$, so $|W \cap R| \leq |S_{Q_R}| \leq 36 d^2 \cdot |V(Q_R)| \leq 36 d^2 \cdot 22 d = 792 \cdot d^3$; and any subtree of $G[R] - W$ of pathwidth at least $\tau$ would also be a subtree of $G[h^{-1}(V(Q_R))] - S_{Q_R}$, contradicting the property of $S_{Q_R}$ from Lemma~\ref{lem:treePackingPathCycleDensity}.

For (\ref{itm:hitPi}): since $g^{-1}(p_0), \ldots, g^{-1}(p_\ell)$ and $R$ are pairwise disjoint, $W \cap g^{-1}(p_i) = S_{Q_i} \cap g^{-1}(p_i)$, so $|W \cap g^{-1}(p_i)| \leq |S_{Q_i}| \leq 36 d^2 \cdot |V(Q_i)| \leq 792 \cdot d^3$; and any subtree of $G[g^{-1}(p_i)] - W$ of pathwidth at least $\tau$ would also be a subtree of $G[h^{-1}(V(Q_i))] - S_{Q_i}$, contradicting the property of $S_{Q_i}$.

The total running time is dominated by the call to Lemma~\ref{lem:treePackingPathCycleDensity}, which takes $2^{O(9^\tau)} \cdot n^{O(1)}$ time. The remaining steps, namely the (at most $\ell+1$) applications of Lemma~\ref{lem:localDensitySeparator} and explicit construction of the weak pre-layout all run in polynomial time.
\end{proof}

\section{Embedding Into a Subtree}\label{sec:intoSubtree}
%\todo[inline]{
%Likely *all* numbers in this section are bogus, and need to be re-worked carefully after fixing the numbers of Lemma~\ref{lem:lowStretchToEmbedding}. fix notation (possibly adjust definition of (weak) pre-layout as sequence of graphs $G_i$ as opposed to the function notation),
%ensure every term that is defined is used and actually needs to be defined. Check indices (off by one errors are probably abundant). Make figures for (weak) pre-layout, scaffolding graph, %witness tree, possibly showing what maps where.} 

{\bf Pre-layouts.} A {\em pre-layout} of $G$ with bandwidth $\beta$ is a six-tuple $(\ell, P, g, \zeta, W, \{\lambda_i ~:~ 0 \leq i \leq \ell\})$ such that $(\emptyset, \emptyset, \ell, P, g, \zeta, W)$ is a weak pre-layout of $G$ and for each $0 \leq i \leq \ell$, $\lambda_i$ is a bandwidth-$\beta$ layout of part $G_i$ of the weak pre-layout $(\emptyset, \emptyset, \ell, P, g, \zeta, W)$.
Given a graph $G$ and a weak pre-layout of $G$, we will obtain a low bandwidth layout for each part by a recursive invocation of the algorithm (see the proof of Lemma~\ref{lem:mainAlgorithmInduction}). Thus we will obtain a pre-layout of $G - (S \cup R)$ or a subtree witnessing that $G$ has large bandwidth. Next we show how to deal with a graph $G$ that has a pre-layout. 
The goal is to show that $G$ has a subtree $T$ such that $G$ embeds into $T$ with low stretch and congestion. This shows that if $T$ has low bandwidth, then so does $G$. 
A pre-layout is sufficient to define the subtree $T$ into which we will embed $G$. We now define this subtree and prove some structural properties of $T$ that will enable us to embed $G$ into $T$. 

\smallskip
\noindent
{\bf The Witness Tree.}
Given a graph $G$ and a pre-layout $(\ell, P, g, \zeta, W, \{\lambda_i ~:~ 0 \leq i \leq \ell\})$ of $G$ we define the {\em witness tree} as follows.
For each $1 \leq i < \ell$, fix an arbitrary spanning tree $T_\zeta^i$ of $G[\zeta(i)]$.
For each $1 \leq i < \ell-1$, let $\pi_i$ be a shortest path in $G$ between the sets $\zeta(i)$ and $\zeta(i+1)$ (so the endpoints of $\pi_i$ are the only vertices of $\pi_i$ in $\zeta(i) \cup \zeta(i+1)$).

\begin{lemma}\label{lem:witnessPathsContainmentDisjoint}
For every $1 \leq i < \ell-1$, $V(\pi_i) \subseteq g^{-1}(\{p_i, s_{i+1}, p_{i+1}\})$, and the paths $\pi_1, \ldots, \pi_{\ell-2}$ are pairwise internally vertex-disjoint.
\end{lemma}

\begin{proof}
We first prove the containment claim. Since $g$ has stretch one, the walk $g(\pi_i)$ in $P$ from $p_i$ to $p_{i+1}$ must visit $s_{i+1}$, so $\pi_i$ has an internal vertex in $g^{-1}(s_{i+1})$. If $V(\pi_i) \not\subseteq g^{-1}(\{p_i, s_{i+1}, p_{i+1}\})$, then by stretch-$1$ the walk $g(\pi_i)$ visits $s_i$ or $s_{i+2}$, so $\pi_i$ contains a sub-path from a vertex in $g^{-1}(s_{i+1})$ to a vertex in $g^{-1}(s_i) \cup g^{-1}(s_{i+2})$ whose vertices are internal to $\pi_i$ and hence disjoint from $\zeta(i) \cup \zeta(i+1)$. This contradicts property~(\ref{itm:smallConnectedSeparator}) applied to $\zeta(i)$ (in the $s_i$ case) or to $\zeta(i+1)$ (in the $s_{i+2}$ case).

Next we argue that the paths $\pi_1, \ldots, \pi_{\ell-2}$ are pairwise internally vertex-disjoint. For every $i$, $j$ with  $|i-j| > 1$, the paths $\pi_i$ and $\pi_j$ have $V(\pi_i) \cap V(\pi_j) = \emptyset$ since $g^{-1}(\{p_i, s_{i+1}, p_{i+1}\}) \cap g^{-1}(\{p_j, s_{j+1}, p_{j+1}\}) = \emptyset$. If $\pi_i$ and $\pi_{i+1}$ share an internal vertex $v$, then $v \in g^{-1}(p_{i+1}) \setminus \zeta(i+1)$; concatenating the sub-path of $\pi_i$ from its $g^{-1}(s_{i+1})$-vertex $w_1$ to $v$ with the sub-path of $\pi_{i+1}$ from $v$ to its $g^{-1}(s_{i+2})$-vertex $w_2$ yields a walk from $w_1 \in g^{-1}(s_{i+1})$ to $w_2 \in g^{-1}(s_{i+2})$ all of whose vertices avoid $\zeta(i+1)$, contradicting property~(\ref{itm:smallConnectedSeparator}) applied to $\zeta(i+1)$.
\end{proof}

From Lemma~\ref{lem:witnessPathsContainmentDisjoint} it follows immediately that  $T' = \bigcup_{i=1}^{\ell-1} T_\zeta^i ~\cup~ \bigcup_{i=1}^{\ell-2} \pi_i$ is a tree. Let $T$ be an arbitrarily chosen spanning tree of $G$ that contains $T'$. We call $T$ the {\em witness tree}.
In the rest of this section, $T$ always refers to this witness tree.
We will show that the vertex set of every part $G_i$ is contained in the vertex set of a subtree $T_i$ of $T$, which itself is contained in the union of the vertex sets of the parts $G_{i-1}$, $G_i$ and $G_{i+1}$.

\begin{lemma}\label{lem:niceSubtrees}
For each $0 \leq i \leq \ell$ there exists a subtree $T_i$ of $T$ such that 
\begin{itemize}\setlength{\itemsep}{-.7pt}
    \item $g^{-1}(\{p_0, s_1\}) \subseteq V(T_0) \subseteq g^{-1}(\{p_0, s_1, p_1\})$,
    \item $g^{-1}(\{s_\ell, p_\ell\}) \subseteq V(T_\ell) \subseteq g^{-1}(\{p_{\ell-1}, s_\ell, p_\ell\})$, and 
    \item if $1 \leq i < \ell$ then $g^{-1}(\{s_i, p_i, s_{i+1}\}) \subseteq V(T_i) \subseteq g^{-1}(\{p_{i-1},s_i, p_i, s_{i+1}, p_{i+1}\})$.
\end{itemize}
Furthermore, given $T$ and the pre-layout $(\ell, P, g, \zeta, W, \{\lambda_i ~:~ 0 \leq i \leq \ell\})$ we can compute $\{ T_i ~:~ 0 \leq i \leq \ell \}$ in polynomial time.
\end{lemma}

\begin{proof}
We define $T_0$ as the subtree of $T$ induced by $\zeta(1)$, $g^{-1}(\{p_0, s_1\})$, and all vertices that lie on a path in $T$ between a vertex in $g^{-1}(\{p_0, s_1\})$ and $\zeta(1)$.
By definition $g^{-1}(\{p_0, s_1\}) \subseteq V(T_0)$. Since $T[\zeta(1)]$ is connected and every vertex in $T_0$ can reach $\zeta(1)$, $T_0$ is connected. It remains to show that
$V(T_0) \subseteq g^{-1}(\{p_0, s_1, p_1\})$. For this it is sufficient to show that every path $Q$ in $T$ from a vertex in $g^{-1}(\{p_0, s_1\})$ to a vertex in $\zeta(1)$ is fully contained in $g^{-1}(\{p_0, s_1, p_1\})$. This follows from the fact that $\zeta(1)$ separates  $g^{-1}(\{s_1\})$ from $g^{-1}(\{s_2\})$ (property~(\ref{itm:smallConnectedSeparator}) of weak pre-layouts).

We define $T_\ell$ as the subtree of $T$ induced by $\zeta(\ell-1)$, $g^{-1}(\{s_\ell, p_\ell\})$, and all vertices that lie on a path in $T$ between a vertex in $g^{-1}(\{s_\ell, p_\ell\})$ and $\zeta(\ell-1)$.
Since $T[\zeta(\ell-1)]$ is connected and every vertex in $T_\ell$ can reach $\zeta(\ell-1)$, $T_\ell$ is connected. The proof that $g^{-1}(\{s_\ell, p_\ell\}) \subseteq V(T_\ell) \subseteq g^{-1}(\{p_{\ell-1}, s_\ell, p_\ell\})$ is symmetric to the proof that $g^{-1}(\{p_0, s_1\}) \subseteq V(T_0) \subseteq g^{-1}(\{p_0, s_1, p_1\})$.

For every  $1 \leq i < \ell$ we define $T_i$ as the subtree of $T$ induced by $g^{-1}(\{s_i, p_i, s_{i+1}\})$, and all vertices that lie on a path in $T$ between a vertex in $g^{-1}(\{s_i, p_i, s_{i+1}\})$ and $\zeta(i)$.
By definition $g^{-1}(\{s_i, p_i, s_{i+1}\}) \subseteq V(T_i)$.
Since $T[\zeta(i)]$ is connected and every vertex in $T_i$ can reach $\zeta(i)$, $T_i$ is connected.
It remains to show that
$V(T_i) \subseteq g^{-1}(\{p_{i-1},s_i, p_i, s_{i+1}, p_{i+1}\})$.
Let $Q$ be a path in $T$ from a vertex in $g^{-1}(\{s_i, p_i, s_{i+1}\})$ to $\zeta(i)$.
It suffices to show that $V(Q)$ is disjoint from $g^{-1}(\{s_j\})$ for $j \in \{i-1,i+2\}$. Towards a contradiction, suppose not. If $V(Q) \cap g^{-1}(s_{i+2})$ is non-empty, let $q$ be a vertex in $V(Q) \cap g^{-1}(s_{i+2})$.
The first part of $Q$ up to $q$ contains a sub-path from $g^{-1}(s_{i+1})$ to $q$, and hence from $\zeta(i+1)$ to $q$. The second part of $Q$ (starting at $q$) contains a subpath from $q$ to $g^{-1}(s_{i+1})$, and hence from $q$ to $\zeta(i+1)$. Thus $Q$ contains a subpath that starts in  $\zeta(i+1)$, leaves  $\zeta(i+1)$ and comes back. But $T[\zeta(i+1)]$ is connected so $T[\zeta(i+1) \cup V(Q)]$ contains a cycle, a contradiction. The case that $V(Q) \cap g^{-1}(s_{i-1})$ is non-empty is symmetric.
The definitions of $T_i$ for every $i$ directly lead to polynomial time algorithms to compute them.
\end{proof}

The overall plan is to construct a tree $H$ (which is not necessarily a subtree of $G$), called the scaffolding graph, embed $G$ into $H$ and then embed $H$ back into $T$. 
When embedding $G$ into $H$ we will embed each part $G_i$ individually, just ensuring that the embeddings of $G_i$ and $G_{i+1}$ agree on the intersection of $V(G_i)$ and $V(G_{i+1})$, namely $g^{-1}(s_{i+1})$, while also controlling the congestion of the embedding by ensuring that no vertex of $H$ is used by more than a constant number of parts of $G$. 
When embedding $H$ into $T$ we will split $H$ according to some special vertices, and embed different pieces of $H$ individually. We will need to carefully control where the special vertices of $H$ will be mapped in the embedding of $H$ into $T$. The next lemma locates the vertices in $T$ to which we will map the special vertices of $H$. 
%\todo[inline]{Add statement that union of two consecutive $T_i$'s is connected?}     

\begin{lemma}\label{lem:selectAnchor}
There exists a polynomial time algorithm that takes as input a graph $G$,  
a pre-layout $(\ell, P, g, \zeta, W, \{\lambda_i ~:~ 0 \leq i \leq \ell\})$ of $G$,
and a witness tree $T$,
%runs in polynomial time, 
and outputs a set $\{a_i ~:~ 1 \leq i < \ell\}$ of vertices such that
$a_i \in \zeta(i)$ for every $1 \leq i < \ell$, and for every
$1 \leq i < \ell-1$ it holds that $d_T(a_i, a_{i+1}) \leq 3 \cdot d_G(a_i, a_{i+1})$.
\end{lemma}

 \begin{proof}
For each $2 \leq i \leq \ell-1$ let $Q_i$ be the path in $T$ with one endpoint in $\zeta(i-1)$ and the other in $\zeta(i)$.
From the construction of $T$, $Q_i$ is a shortest path from $\zeta(i-1)$ to $\zeta(i)$ in $G$.
We set $a_1$ to be the endpoint of $Q_2$ in $\zeta(1)$ and $a_{\ell-1}$ to be the endpoint of $Q_{\ell-1}$ in $\zeta(\ell-1)$.
For every $1 < i < \ell-1$ we define $\hat{Q}_i$ to be the path in $T[\zeta(i)]$ that connects the endpoint of $Q_{i}$ in $\zeta(i)$ with the endpoint of $Q_{i+1}$ in $\zeta(i)$.

Then $Q_{i}\hat{Q}_iQ_{i+1}$ is a path from $\zeta(i-1)$ to $\zeta(i+1)$. In particular it intersects both $g^{-1}(s_i)$ and $g^{-1}(s_{i+1})$. Therefore we have that 
$$|E(Q_{i})| + |E(\hat{Q}_i)| + |E(Q_{i+1})| = |E(Q_{i}\hat{Q}_iQ_{i+1})| \geq d_G(g^{-1}(s_i), g^{-1}(s_{i+1})) \geq 2|\zeta(i)| > 2|E(\hat{Q}_i)|\mbox{.}$$
Here the transition $d_G(g^{-1}(s_i), g^{-1}(s_{i+1})) \geq 2|\zeta(i)|$ follows from property (\ref{itm:smallConnectedSeparator}) of weak pre-layouts. Re-arranging terms we get $|E(\hat{Q}_i)| < |E(Q_{i})| + |E(Q_{i+1})|$.

We select $a_i$ to be an arbitrary vertex of $\hat{Q}_i$ such that $a_i$ is at most $|E(Q_{i})|$ steps on $\hat{Q}_i$ away from the endpoint of $\hat{Q}_i$ in $Q_{i}$, and at most $|E(Q_{i+1})|$ steps on $\hat{Q}_i$ away from the endpoint of $\hat{Q}_i$ in $Q_{i+1}$.
It follows that 
$$d_T(a_1, a_2) \leq 2|E(Q_2)| = 2d_G(\zeta(1), \zeta(2)) \leq 2d_G(a_1, a_2)\mbox{,}$$
$$d_T(a_{\ell-2}, a_{\ell-1}) \leq 2|E(Q_{\ell-1})| = 2d_G(\zeta(\ell-2), \zeta(\ell-1))
\le 2d_G(a_{\ell-2}, a_{\ell-1})
\mbox{,}$$
and for every $2 \leq i \leq \ell-3$ we have 
$$d_T(a_{i}, a_{i+1}) \leq 3|E(Q_{i+1})| \leq 3d_G(\zeta(i), \zeta(i+1)) \leq 3d_G(a_i, a_{i+1})\mbox{.}$$
\end{proof}

\smallskip
{\bf Anchor Vertices.} We will call the vertices $\{a_i ~:~ 1 \leq i < \ell\}$ given by Lemma~\ref{lem:selectAnchor} the {\em anchor vertices} of $G$. In the rest of this section, $\{a_1, \ldots, a_{\ell-1}\}$ always refers to these anchor vertices. Next we show that forcing a shortest path in $G$ from any vertex $u \in g^{-1}(s_i)$ to $u' \in g^{-1}(s_{i+1})$ to go via $a_{i}$ does not increase the length of the path too much.

\begin{lemma}\label{lem:goThroughAnchor}
For every $1 \leq i < \ell$, $u \in g^{-1}(s_i)$, and $u' \in g^{-1}(s_{i+1})$, it holds that
$d_G(u, a_i) + d_G(a_i, u') \leq 2 \cdot d_G(u, u')$.
\end{lemma}

\begin{proof}
%\todo[inline]{E: should be $\ge 2|\zeta(i)|$, right? DL: yes}
Since $u \in g^{-1}(s_i)$ and $u' \in g^{-1}(s_{i+1})$ we have that
$d_G(u, u') \geq d_G(g^{-1}(s_i), g^{-1}(s_{i+1})) \geq 2|\zeta(i)|$
, where the last transition follows from property (\ref{itm:smallConnectedSeparator}) of weak pre-layouts.
Let $Q$ be a shortest path from $u$ to $u'$ in $G$. Since $\zeta(i)$ separates $g^{-1}(s_i)$ from $g^{-1}(s_{i+1})$ it follows that $Q$ contains a vertex $q$ in $\zeta(i)$. We have that $d_G(q, a_i) < |\zeta(i)|$.
The triangle inequality yields that
$$d_G(u, a_i) + d_G(a_i, u') \leq d_G(u, q) + d_G(q, a_i) + d_G(a_i, q) + d_G(q, u') < d_G(u, u') + 2|\zeta(i)| \leq 2 \cdot d_G(u, u')\mbox{.}$$
This concludes the proof.
\end{proof}

We are almost ready to embed $G$ into a tree $H$, and then show that $H$ is embeddable into the witness tree $T$. All that remains is a simple upper bound stating that vertices in $G$ that are mapped close together on $P$ are not too far from each other in $G$.

\begin{lemma}\label{lem:reallySillyDistanceBound}
For every pair $u,v$ in $g^{-1}(s_1) \cup \{a_1\}$, it holds that $d_G(u, v) \leq |g^{-1}(\{s_1,p_1,s_2\})| - 1$.
For every pair $u,v$ in $g^{-1}(s_\ell) \cup \{a_{\ell-1}\}$, it holds that $d_G(u, v) \leq |g^{-1}(\{s_{\ell-1}, p_{\ell-1}, s_\ell\})| - 1$.
For every $2 \leq i \leq \ell-1$ and every pair $u,v$ in $g^{-1}(s_i) \cup \{a_{i-1}, a_i\}$, it holds that  $d_G(u, v) \leq |g^{-1}(\{s_{i-1}, p_{i-1}, s_{i}, p_i, s_{i+1}\})| - 1$. 
\end{lemma}

\begin{proof}
By property~(\ref{itm:connectedMid}) of weak pre-layouts, 
we have that $g^{-1}(\{s_1,p_1,s_2\})$ is a connected set that contains $g^{-1}(s_1) \cup \{a_1\}$,
$g^{-1}(\{s_{\ell-1}, p_{\ell-1}, s_\ell\})$ is a connected set that contains $g^{-1}(s_\ell) \cup \{a_{\ell-1}\}$,
and $g^{-1}(\{s_{i-1}, p_{i-1}, s_{i}, p_i, s_{i+1}\})$ is a connected set that contains $g^{-1}(s_i) \cup \{a_{i-1}, a_i\}$. This concludes the proof. 
\end{proof}

\smallskip
{\bf The Scaffolding Graph.} Given a graph $G$, a pre-layout $(\ell, P, g, \zeta, W, \{\lambda_i ~:~ 0 \leq i \leq \ell\})$ of $G$ with separator size $s$ and bandwidth $\beta$, and the anchor vertices  $\{a_i ~:~ 1 \leq i < \ell\}$ of $G$ we define the {\em scaffolding graph} $H$ of $G$ as follows.
Make a path $R_1$ on $12(s+1)|g^{-1}(\{p_0,s_1,p_1,s_2\})|$ vertices,
and a path $R_\ell$ on $12(s+1)|g^{-1}(\{s_{\ell-1},p_{\ell-1},s_\ell,p_\ell\})|$ vertices. For every $2 \leq i \leq \ell-1$ make a path $R_i$ on $12(s+2)|g^{-1}(\{s_{i-1},p_{i-1},s_i,p_i,s_{i+1}\})|$ vertices. 
%\todo[inline]{i think here we use that $\ell$ is bounded from below, check by how much. $3$ is enough.}

Using Theorem~\ref{thm:lineEmbedding} we obtain a function 
$r_1 : g^{-1}(s_1) \cup \{a_1\} \rightarrow \mathbb{N}$ such that for every $u, v \in g^{-1}(s_1) \cup \{a_1\}$ it holds that $d_G(u,v) \leq |r_1(u) - r_1(v)| \leq 12(s+1)d_G(u,v)$.
Similarly, obtain a function $r_\ell : g^{-1}(s_\ell) \cup \{a_{\ell-1}\} \rightarrow \mathbb{N}$ such that for every $u, v \in g^{-1}(s_\ell) \cup \{a_{\ell-1}\}$ it holds that $d_G(u,v) \leq |r_\ell(u) - r_\ell(v)| \leq 12(s+1)d_G(u,v)$.
For every $2 \leq i \leq \ell-1$ obtain a function $r_i : g^{-1}(s_i) \cup \{a_{i-1}, a_i\} \rightarrow \mathbb{N}$ such that for every $u, v \in g^{-1}(s_i) \cup \{a_{i-1}, a_i\}$ it holds that $d_G(u,v) \leq |r_i(u) - r_i(v)| \leq 12(s+2)d_G(u,v)$.
In the rest of this section, $r_1, \ldots, r_\ell$ always refer to these functions.

Let $i$ be an integer between $1$ and $\ell$, $u$ be a vertex with minimum $r_i(u)$ and $v$ be a vertex with maximum $r_i(v)$. By Lemma~\ref{lem:reallySillyDistanceBound} we have that $r_i(v) - r_i(u) \leq |E(R_i)|$. Thus we may treat each $r_i$ as a function with range $V(R_i)$ instead of $\mathbb{N}$, where for every $u$, $v$ in the domain of $r_i$, $|r_i(u) - r_i(v)|$ is replaced by $d_{R_i}(r_i(u), r_i(v))$. We obtain the graph $H$ by taking the disjoint union of all the paths $R_i$ for $1 \leq i \leq \ell$, as well as $\ell-1$ additional vertices $\hat{a}_1, \ldots, \hat{a}_{\ell-1}$ and adding for every $1 \leq i \leq \ell-1$ the edges $\hat{a}_i r_i(a_i)$ and $\hat{a}_i r_{i+1}(a_i)$.
It is easy to see that $H$ is indeed a tree. Indeed, $H - \{\hat{a}_1, \ldots, \hat{a}_{\ell-1}\}$ is just the disjoint union of the paths $R_1, \ldots, R_\ell$ and each $\hat{a}_i$ has precisely one neighbor in $R_i$ and one in $R_{i+1}$. 
Let $a$ be one of the two endpoints of $R_1$, and $b$ be one of the two endpoints of $R_\ell$. Let $B$ be the unique path in $H$ connecting $a$ and $b$. Then $B$ contains $\{\hat{a}_1, \ldots, \hat{a}_{\ell-1}\}$ and their neighbors, and in particular all vertices of $H$ of degree at least $3$. Hence $H$ is a caterpillar with backbone $B$.
In the rest of this section, $H$, $R_1, \ldots, R_\ell$, and $\hat{a}_1, \ldots, \hat{a}_{\ell-1}$ always refer to these objects.

\begin{figure}[ht]
    \centering
    \includegraphics[width=0.9\textwidth]{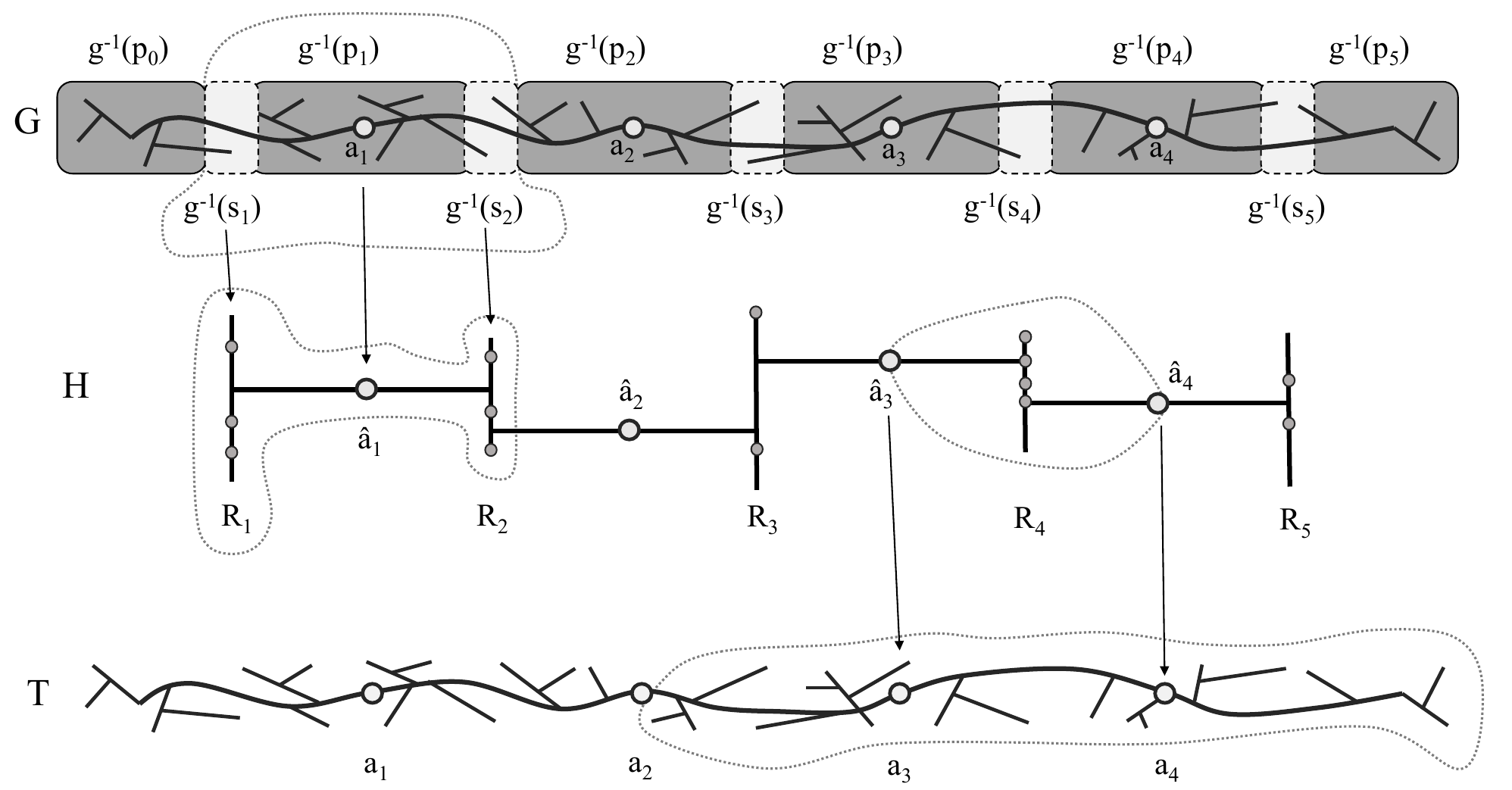}
    \caption{\em Top row: a pre-layout of a graph $G$ with $\ell=5$, anchor vertices, and witness tree $T$. Middle row: The scaffolding graph $H$. Bottom row: the witness tree $T$. The arrows from the first to second row illustrate the anchor embedding and the embedding from $G$ to $H$ constructed in Lemma~\ref{lem:embedGtoH}. The arrows from the second to third row illustrate the embedding constructed in Lemma~\ref{lem:embedHtoT}.}
    \label{fig:scaffold}
\end{figure}

\smallskip
{\bf The Anchor Embedding.} We define the {\em anchor embedding} 
$$r : g^{-1}(\{s_1, \ldots, s_\ell\}) \cup \{a_1, \ldots, a_{\ell-1}\} \rightarrow V(H)$$ 
as follows.
For every vertex $v \in g^{-1}(\{s_i\})$ for $1 \leq i \leq \ell$ we define $r(v) = r_i(v)$. For every $1 \leq i \leq \ell-1$ we set $r(a_i) = \hat{a}_i$.
In the rest of this section, $r$ always refers to this anchor embedding.
Before showing that $G$ has a low stretch and congestion embedding into $H$, we show that for every part of the pre-layout, $r$ does not stretch too much the separator and anchor vertices of the part. In the rest of the section $s$ denotes the separator size of the weak pre-layout underlying the considered pre-layout.

%\todo[inline]{s is the separator size of the pre layout, add inside lemma or note that this is how we will use it. s is a bad name (reused)}

%\todo[inline]{technically the $r_i$'s are not defined inside the scope of the lemma}

\begin{lemma}\label{lem:rHasLowStretch}
For every $1 \leq i < \ell$ and every pair $u$, $v$ of vertices in $g^{-1}(\{s_i, s_{i+1}\}) \cup \{a_i\}$ it holds that 
$d_H(r(u), r(v)) \leq 25(s+2)d_G(u, v).$
\end{lemma}

\begin{proof}
Let $1 \leq i < \ell$ and $u$, $v$ be vertices in $g^{-1}(\{s_i, s_{i+1}\}) \cup \{a_i\}$.
If both $u$ and $v$ are in $g^{-1}(\{s_i\})$ then 
$d_H(r(u), r(v)) = d_{R_i}(r(u), r(v)) \leq 12(s+2)d_G(u, v)$, and an identical bound holds if both $u$ and $v$ are in $g^{-1}(\{s_{i+1}\})$.
If $u$ is in  $g^{-1}(\{s_i\})$ and $v = a_i$ then $d_H(r(u), r(v)) = d_{R_i}(r_i(u), r_i(v)) + 1 \leq 12(s+2)d_G(u, v) \cdot 2$, and identical bound holds if $u$ is in  $g^{-1}(\{s_{i+1}\})$ and $v = a_i$.
Suppose now that $u$ is in $g^{-1}(\{s_i\})$ and $v$ is in $g^{-1}(\{s_{i+1}\})$.
It follows that 
\begin{align*}
d_H(r(u), r(v)) & \leq d_{R_i}(r_i(u), r_i(a_i)) + 2 + d_{R_{i+1}}(r_{i+1}(a_i), r_{i+1}(v)) \\
& \leq 12(s+2)d_G(u, a_i) + 2 +  12(s+2)d_G(a_i, v).
\end{align*}

From Lemma~\ref{lem:goThroughAnchor} we have that $d_G(u, a_i) + d_G(a_i, v) \leq 2d_G(u, v)$, and so the bound follows. 
\end{proof}

\todo[inline]{\color{red} E: instead of multiplying $\cdot 2$, better write $+1$, which is what the argument gives (and is tighter), and used below too, right? I changed also the displayed inequality below accordingly, +4 instead of +2.}

\todo[inline]{DL:  agree on $+1$, put back $+2$ in displayed inequality it is correct as is (and $+4$ means statement needs to be updated. )}

We now have all the ingredients ready to wrap up the proof. We start by showing that $G$ can be embedded into the scaffolding graph $H$.

\begin{lemma}\label{lem:embedGtoH}
There exists a polynomial time algorithm that takes as input a graph $G$, a pre-layout $(\ell, P, g, \zeta, W, \{\lambda_i ~:~ 0 \leq i \leq \ell\})$ of $G$ with separator size $s$ and bandwidth $\beta$, anchor vertices $\{a_1, \ldots, a_{\ell-1}\}$, scaffolding graph $H$ and anchor embedding $r$, and computes an embedding $h$ of $G$ into $H$ with
stretch at most $96000(s+2)^4(\beta+1)^2$
and congestion at most $84000(s+1)^4(\beta+1)$.
\end{lemma}

\begin{proof}
From properties (\ref{itm:connectedEnds}) and (\ref{itm:connectedMid}) of weak pre-layouts all parts $G_i$ are connected. Using Theorem~\ref{lem:lowStretchToEmbedding} we obtain an embedding $h_0$ of the part
$G_0$ into $R_1$ such that $h_0(v) = r(v)$ for every $v \in g^{-1}(s_1)$.
To apply Theorem~\ref{lem:lowStretchToEmbedding} we note that $|V(R_1)| \geq |V(G_0)|$, that by property~(\ref{itm:siSize}) we have $|g^{-1}(s_1)| \leq s$, and that by Lemma~\ref{lem:rHasLowStretch} the $g^{-1}(s_1)$-stretch of $r$ is at most $25(s+2)$. Finally we have at hand an embedding $\lambda_0$ of $G_0$ with bandwidth at most $\beta$. Thus the
stretch of $h_0$ is at most
$$480(s+1)^3(\beta+1)^2 \cdot 25(s+2) \leq 12000(s+2)^4(\beta + 1)^2$$
and the
congestion of $h_0$ is at most
$$1300(s+1)^4(\beta+1) \cdot 2 = 2600(s+1)^4(\beta+1).$$
%\todo[inline]{E: can you please detail these two computations estimating stretch and congestion? Perhaps similarly below too? DL: Will do for arxiv version}

Similarly, using Theorem~\ref{lem:lowStretchToEmbedding} we obtain an embedding $h_\ell$ of the part $G_\ell$ into $R_\ell$ such that $h_\ell(v) = r(v)$ for every $v \in g^{-1}(s_\ell)$. Calculations identical to those for $h_0$ show that the stretch of $h_\ell$ is at most $12000(s+2)^4(\beta + 1)^2$ and the
congestion of $h_\ell$ is at most $2600(s+1)^4(\beta+1)$.

For every inner part $1 \leq i < \ell$, by using Theorem~\ref{lem:lowStretchToEmbedding} we obtain an embedding $h_i$ of the part $G_i$ into $H[R_i \cup \{\hat{a}_i\} \cup R_{i+1}]$ such that $h_i(v) = r(v)$ for every $v \in g^{-1}(\{s_{i}, s_{i+1}\})$.
To apply Theorem~\ref{lem:lowStretchToEmbedding} we note that $|V(R_{i+1})| \geq |V(G_i)|$, that by property~(\ref{itm:siSize}) we have $|g^{-1}(\{s_i, s_{i+1}\})| \leq 2s$, and that by Lemma~\ref{lem:rHasLowStretch} the $g^{-1}(\{s_i, s_{i+1}\})$-stretch of $r$ is at most $25(s+2)$. Finally we have at hand an embedding $\lambda_i$ of $G_i$ with bandwidth at most $\beta$. 
Thus, the stretch of $h_i$ is at most
$$480(2(s+1))^3(\beta+1)^2 \cdot 25(s+2) \leq 96000(s+2)^4(\beta+1)^2$$
and the congestion of $h_i$ is at most
$$1300(2(s+1))^4(\beta+1) \cdot 2 \leq 42000(s+1)^4(\beta+1).$$

%\todo[inline]{need a simple lemma in manipulating embeddings section that we can put together embeddings of sungraphs that agree on common vertices}

We define the embedding $h$ of $G$ into $H$ by setting $h(v) = h_i(v)$ for every $v$ in part $i$ of the pre-layout.  Notice that if $v$ appears both in part $i$ and in part $j$ of the pre-layout and $i < j$, then $j = i+1$, $v \in g^{-1}(s_{i+1})$ and $h_i(v) = h_j(v) = r(v)$, and hence the function $h$ is well defined. 
Since every edge of $G$ has both endpoints in some part of the pre-layout, the stretch of $h$ is at most the maximum stretch of any $h_i$, in particular it is at most
$96000(s+2)^4(\beta+1)^2$.
Finally, for each vertex $v \in V(H)$ at most two functions $h_i$ have $v$ in their range, thus the congestion of $h$ is at most twice the maximum congestion of any $h_i$. In particular the congestion of $h$ is at most $84000(s+1)^4(\beta+1)$.
\end{proof}

Next we show that the scaffolding graph $H$ can be embedded back into the witness tree $T$. 

\begin{lemma}\label{lem:embedHtoT}
There exists a polynomial time algorithm that takes as input a graph $G$, a pre-layout $(\ell, P, g, \zeta, W, \{\lambda_i ~:~ 0 \leq i \leq \ell\})$ of $G$ with separator size $s$ and bandwidth $\beta$, anchor vertices $\{a_1, \ldots, a_{\ell-1}\}$, scaffolding graph $H$, anchor embedding $r$,
and witness tree $T$ of $G$, and outputs an embedding $t$ of $H$ into $T$ with stretch at most $630000$ and congestion at most $22000000(s+2)$.
\end{lemma}

\begin{proof}
Let $\{T_i ~:~ 0 \leq i \leq \ell\}$ be subtrees of $T$ as given by Lemma~\ref{lem:niceSubtrees}.
We apply Theorem~\ref{lem:lowStretchToEmbedding} to obtain an embedding $t_1$ of $H[V(R_1) \cup \hat{a}_1]$ into $T[V(T_0) \cup V(T_1)]$ such that $t_1(\hat{a}_1) = a_1$.
To apply Theorem~\ref{lem:lowStretchToEmbedding} we note that $H[V(R_1) \cup \hat{a}_1]$, being a path plus one vertex with a single edge into the path, has bandwidth at most $2$; that $V(T_0) \cup V(T_1)$ contains $g^{-1}(\{p_0,s_1,p_1,s_2\})$, hence $|V(T_0) \cup V(T_1)| \geq |V(R_1) \cup \hat{a}_1|/24(s+1)$; and that the $\{\hat{a}_1\}$-stretch of the mapping that assigns $a_1$ to $\hat{a}_1$ is $1$ (vacuous on a singleton). Thus the stretch of $t_1$ is at most
$$480(1+1)^3(2+1)^2 \cdot 1 \leq 35000$$
and the congestion of $t_1$ is at most
$$1300(1+1)^4(2+1) \cdot 24(s+1) \leq 1500000(s+1).$$
By an identical argument we obtain an embedding $t_\ell$ of $H[V(R_\ell) \cup \hat{a}_{\ell-1}]$ into $T[V(T_{\ell-1}) \cup V(T_{\ell})]$ such that $t_\ell(\hat{a}_{\ell-1}) = a_{\ell-1}$. The stretch of $t_\ell$ is at most $35000$ and the congestion is at most $1500000(s+1)$.

For every $2 \leq i \leq \ell-1$ we apply Theorem~\ref{lem:lowStretchToEmbedding} to obtain an embedding $t_i$ of $H[V(R_i) \cup \{\hat{a}_{i-1},\hat{a}_i\}]$ into $T[V(T_{i-1}) \cup V(T_i)]$ such that $t_i(\hat{a}_{i-1}) = a_{i-1}$ and $t_i(\hat{a}_{i}) = a_{i}$.
To apply Theorem~\ref{lem:lowStretchToEmbedding} we note that $H[V(R_i) \cup \{\hat{a}_{i-1},\hat{a}_i\}]$, being a path plus two vertices each with a single edge into the path, has bandwidth at most $3$; that by Lemma~\ref{lem:niceSubtrees}, $V(T_{i-1}) \cup V(T_i)$ contains $g^{-1}(\{s_{i-1},p_{i-1},s_i,p_i,s_{i+1}\})$, hence $|V(T_{i-1}) \cup V(T_i)| \geq |V(R_i) \cup \{\hat{a}_{i-1},\hat{a}_i\}|/13(s+2)$. The $\{\hat{a}_{i-1}, \hat{a}_i\}$-stretch of the function that maps $\hat{a}_i$ to $a_i$ and $\hat{a}_{i-1}$ to $a_{i-1}$ is at most $3$: by Lemma~\ref{lem:selectAnchor}, $d_T(a_{i-1}, a_i) \leq 3 d_G(a_{i-1}, a_i)$, and $d_G(a_{i-1}, a_i) \leq d_{R_i}(r_i(a_{i-1}), r_i(a_i)) \leq d_H(\hat{a}_{i-1}, \hat{a}_i)$, so $d_T(a_{i-1}, a_i) \leq 3 d_H(\hat{a}_{i-1}, \hat{a}_i)$. Thus the stretch of $t_i$ is at most
$$480(2+1)^3(3+1)^2 \cdot 3 \leq 630000$$
and the congestion of $t_i$ is at most
$$1300(2+1)^4(3+1) \cdot 13(s+2) \leq 5500000(s+2).$$

We define the embedding $t$ of $H$ to $T$ as follows. For every vertex $\hat{a}_i$ we set $t(\hat{a}_i) = a_i$. For every $R_i$ and every vertex $v \in R_i$ we set $t(v) = t_i(v)$. Note that $t$ agrees with $t_i$ on the domain of $t_i$ for every $1 \leq i \leq \ell$. Furthermore, every edge of $H$ has both endpoints in the domain of $t_i$ for some $1 \leq i \leq \ell$. Thus $t$ is an embedding of $H$ to $T$ with stretch at most $630000$ (the largest stretch of any $t_i$).

Next we argue that for every vertex $v$ in $T$ there are at most $4$ values of $i \in \{1, \ldots, \ell\}$ such that $v$ is in the range of $t_i$. Fix an arbitrary $v$ in $T$.
By Lemma~\ref{lem:niceSubtrees}, there exists a $j$ such that for every $i$, $v \in V(T_i)$ implies $i \in \{j-1,j,j+1\}$.
For every $i \in \{1, \ldots, \ell\}$, the range of $t_i$ is $V(T_{i-1}) \cup V(T_i)$. Thus, if $v$ is in the range of $t_i$ then either $i-1 \in \{j-1,j,j+1\}$ or $i \in \{j-1,j,j+1\}$, i.e., $i \in \{j-1,j,j+1,j+2\}$.
Since every vertex of $T$ is in the range of at most $4$ embeddings from $\{t_i ~:~ 1 \leq i \leq \ell\}$ it follows that the congestion of $t$ is at most $4 \cdot 5500000(s+2) = 22000000(s+2)$.
%%
%
%\todo[inline]{meow}
%There exists a polynomial time algorithm that takes as input
%a connected graph $G$, positive integers $k$, $r$, $s$,
%a layout $f$ of $G$ of bandwidth at most $k$,
%a vertex set $S$ in $G$,
%a graph $H$ such that $|V(H)| \geq |V(G)|/r$
%and a function $h : S \rightarrow V(H)$ with $S$-stretch at most $s$,
%and outputs an embedding $g$ of $G$ into $H$ with
%stretch at most $120|S|kr \cdot \max(|S|, k)$
%congestion at most  $1400|S|^3kr$
%such that for every $s \in S$, $g(s) = h(s)$.
%
%    
\end{proof}

We are now ready to show that when a graph $G$ has a pre-layout, we can either find a low bandwidth layout of $G$ or a subtree of $G$ with large bandwidth. First we recall the approximation algorithm for the bandwidth of caterpillars given by Dregi and Lokshtanov~\cite{DregiL14}.

%\todo[inline]{for the approximation algorithm that does not need to find a subtree of high bandwidth we can instead use the approximation algorithm on $H$ rather than $T$ since $H$ has pathwidth $2$ and the dregi-lokshtanov algorithm has $poly(k)$  approximation ratio then}

%\begin{theorem}[\cite{DregiL14,dregi2017beyond}]\label{lem:treeApproximation}
%There exists polynomial time algorithm that takes as input a tree $T$ and an integer $k$ and either returns that $\bw(T) > k$ or produces a layout of $T$ of bandwidth at most $(5k)^{6k}$.
%\end{theorem}

\begin{theorem}[\cite{DregiL14}]\label{lem:caterpillarApproximation}
There exists an algorithm that, given a caterpillar $H$ and a positive integer $k$ either
returns a layout of $H$ of bandwidth at most $48k^3$ or correctly concludes that $\bw(H) > k$ in time $O(kn^3)$.
\end{theorem}

Armed with Theorem~\ref{lem:caterpillarApproximation} we prove the main result of this section.

\begin{lemma}\label{lem:prelayoutGraphs}
There exists a polynomial time algorithm that takes as input a graph $G$, a pre-layout $(\ell, P, g, \zeta, W, \{\lambda_i ~:~ 0 \leq i \leq \ell\})$ of $G$ with separator size $s$ and bandwidth $\beta$, together with an integer $k$, and either outputs a layout of $G$ of bandwidth at most $10^{54} \cdot (s+2)^{11}(\beta+1)^3 k^3$, or a subtree $T$ of $G$ with bandwidth at least $k$.
\end{lemma}

\begin{proof}
%\todo[inline]{E: a witness tree, right? T is not uniquely defined. DL: the, same as above.}
Given $G$ and $(\ell, P, g, \zeta, W, \{\lambda_i ~:~ 0 \leq i \leq \ell\})$ the algorithm first computes the witness tree $T$.
It then computes a set $\{a_i ~:~ 1 \leq i < \ell\}$ of anchor vertices using Lemma~\ref{lem:selectAnchor}.
Next it defines the scaffolding graph $H$ and anchor embedding $r$, and applies Lemma~\ref{lem:embedGtoH} to obtain an embedding $h$ of $G$ into $H$ with
stretch at most $96000(s+2)^4(\beta+1)^2$
and congestion at most $84000(s+1)^4(\beta+1)$.
Then it applies Lemma~\ref{lem:embedHtoT} to obtain an embedding $t$ of $H$ into $T$ with
stretch at most $630000$ and
congestion at most $22000000(s+2)$.

The algorithm approximates the bandwidth of $H$ by setting $z = 10^{14}sk$ and invoking Theorem~\ref{lem:caterpillarApproximation}. The algorithm of Theorem~\ref{lem:caterpillarApproximation} either produces a layout $\lambda_H$ of $H$ with bandwidth at most $48z^3$ or concludes that $\bw(H) > z$. 

If we obtain a layout $\lambda_H$ of bandwidth at most $48z^3$, we treat $\lambda_H$ as an embedding of $H$ into $P_\mathbb{Z}$ with congestion $1$ and stretch at most $48z^3$. By Lemma~\ref{lem:composeEmbeddings}, composing $h$ and $\lambda_H$ gives an embedding of $G$ into  $P_\mathbb{Z}$ with
congestion at most $84000(s+1)^4(\beta+1)$
and
stretch at most $5000000(s+2)^4(\beta+1)^2z^3$.
By Lemma~\ref{lem:embeddingToLayout}, this yields an embedding of $G$ of bandwidth at most
$10^{12} \cdot (s+2)^8(\beta+1)^3z^3 \leq 10^{54} \cdot (s+2)^{11}(\beta+1)^3 k^3$.

If the algorithm of Theorem~\ref{lem:caterpillarApproximation} concludes that $\bw(H) > z$, we output the tree $T$ and claim that $\bw(T) \geq k$. 
Suppose for contradiction that $T$ has a layout $\lambda_T$ of bandwidth at most $k-1$. This is an embedding of $T$ into $P_\mathbb{Z}$ with congestion $1$ and stretch at most $k-1$.
By Lemma~\ref{lem:composeEmbeddings}, composing $t$ with $\lambda_T$ yields an embedding of $H$ into $P_\mathbb{Z}$ with
congestion at most $22000000(s+2)$
and stretch at most $630000(k-1)$.
Lemma~\ref{lem:embeddingToLayout} then yields a layout of $H$ of bandwidth at most $2 \cdot 630000(k-1) \cdot 22000000(s+2) \leq 10^{14}s(k-1) < z$ (using $(s+2) \leq 3s$ for $s \geq 1$), contradicting that $\bw(H) > z$.
\end{proof}
\todo[inline]{E: should this be $630000 (k-1)$? DL: yes, fixed.}

\section{Proof of the Main Result}\label{sec:proofsOfMain}

\begin{lemma}\label{lem:mainAlgorithmInduction}
There exists an algorithm that takes as input a graph $G$, a non-negative integer $\tau$ and a positive integer $k$, such that $\tau \leq k$ and $G$ does not have a subtree of pathwidth $\tau + 1$, runs in time $2^{O(9^k)} \cdot n^{O(1)}$, and either outputs a subtree $T$ of $G$ of bandwidth at least $k$ or a layout of $G$ of bandwidth at most $\beta(k, \tau)$, where $\beta(k, \tau) = (10^{85} \cdot k^{28})^{4^\tau}$.
\end{lemma}

\begin{proof}
If $G$ is disconnected we handle each connected component individually: If at least one component has a subtree $T$ of bandwidth at least $k$, the algorithm outputs $T$. If each of the connected components have a layout of bandwidth at most $\beta(k, \tau)$, then these layouts can be concatenated to yield a layout of $G$ of bandwidth at most $\beta(k, \tau)$. Suppose now that $G$ is connected. 
The algorithm is a recursive algorithm which will make recursive calls on instances with strictly smaller value of $\tau$. We will assume by induction on $\tau$ that those recursive calls behave according to the statement of the lemma. In the base case $\tau = 0$, the hypothesis that $G$ has no subtree of pathwidth at least $1$ implies that $G$ has no edges. Since $G$ is connected, $G$ is a single vertex, and the algorithm outputs the trivial layout of bandwidth $0 \leq \beta(k, 0)$. Suppose now that $\tau \geq 1$.

If the radial local density $\hat{\Delta}_r(G)$ of $G$ is at least $k+1$, the algorithm outputs a subtree $T$ of $G$ with $\hat{\Delta}_r(T) = \hat{\Delta}_r(G) \geq k+1$ using Lemma~\ref{lem:preserveRadialLocalDensity}; since $\bw(T) \geq \hat{\Delta}_r(T) \geq k+1 \geq k$, this $T$ satisfies the lemma's requirement. 
We now assume that $\hat{\Delta}_r(G) \leq k$, and that therefore the maximum degree of $G$ is at most $2k$. Set $d = \max(\hat{\Delta}_r(G), 2)$, so $d \leq \max(k, 2)$. We apply the algorithm of Lemma~\ref{lem:weakPreLayoutDensity} with parameters $d$ and $\tau$. This algorithm runs in time $2^{O(9^\tau)} \cdot n^{O(1)}$ and either outputs a vertex set $S \subseteq V(G)$ of size at most $1566 d^3$ such that $G - S$ has no subtree of pathwidth at least $\tau$, or outputs a weak pre-layout $(S, R, \ell, P, g, \zeta, W)$ of $G$ with separator size at most $4d$ and pathwidth-$\tau$ tree hitting bound $\rho \leq 792 d^3$.

If the algorithm of Lemma~\ref{lem:weakPreLayoutDensity} returns a vertex set $S$ of size at most $1566 d^3$, the algorithm applies itself recursively to $G - S$ with parameter $\tau - 1$. The recursive call returns either a subtree $T$ of $G - S$ of bandwidth at least $k$, or a layout of $G - S$ of bandwidth at most $\beta(k, \tau-1)$. In the first case the algorithm outputs $T$ and halts. In the latter case it applies Lemma~\ref{lem:deleteSet} (with $|S| \leq 1566 d^3$ and max degree $\leq 2k$) to obtain a layout of $G$ of bandwidth at most
$$12 \cdot 1566 d^3 \cdot 2k \cdot \beta(k, \tau-1) = 37584 \cdot d^3 k \cdot \beta(k, \tau-1) \leq 10^{85} \cdot k^{28} \cdot \beta(k, \tau-1) \leq \beta(k, \tau),$$
where the last inequality follows from $\beta(k, \tau) = \beta(k, \tau-1)^4$ together with $\beta(k, \tau-1) \geq 10^{85} \cdot k^{28}$.
 
Suppose now the algorithm of Lemma~\ref{lem:weakPreLayoutDensity} outputs a weak pre-layout $(S, R, \ell, P, g, \zeta, W)$ of $G$ with separator size at most $4d$ and pathwidth-$\tau$ tree hitting bound $\rho \leq 792 d^3$.
Let $P = p_0 s_1 p_1 \ldots s_\ell p_\ell$, and set $Q = g^{-1}(\{s_1, s_2, \ldots, s_\ell\})$. By property~(\ref{itm:siSize}), $|g^{-1}(s_i)| \leq 4d$ for every $1 \leq i \leq \ell$, so for every part $G_i$ of the weak pre-layout we have $|V(G_i) \cap Q| \leq 2 \cdot 4d = 8d$. Further, for every two distinct parts $G_i$, $G_j$ it holds that $(V(G_i) - Q) \cap (V(G_j) - Q) = \emptyset$.
For every $0 \leq i \leq \ell$ let $W_i = W \cap g^{-1}(p_i)$. Then $G[g^{-1}(p_i)] - W_i = G[g^{-1}(p_i)] - W$, which by property~(\ref{itm:hitPi}) has no subtree of pathwidth at least $\tau$ and satisfies $|W_i| \leq \rho$. Similarly, let $W_R = W \cap R$; then $G[R] - W_R = G[R] - W$, which by property~(\ref{itm:rBound}) has no subtree of pathwidth at least $\tau$ and satisfies $|W_R| \leq \rho$.

For every $0 \leq i \leq \ell$ we call the algorithm recursively on $G_i - (Q \cup W_i)$ with parameter $\tau - 1$. The recursive call is valid since $G_i - (Q \cup W_i) = G[g^{-1}(p_i)] - W_i$, and $G[g^{-1}(p_i)] - W_i$ has no subtree of pathwidth at least $\tau$ as established above. It returns either a subtree $T_i$ of $G_i - (Q \cup W_i)$ of bandwidth at least $k$, or a layout of $G_i - (Q \cup W_i)$ of bandwidth at most $\beta(k, \tau-1)$.
Additionally we call the algorithm recursively on $G[R] - W_R$ with parameter $\tau - 1$, obtaining either a subtree $T_R$ of $G[R] - W_R$ of bandwidth at least $k$, or a layout of $G[R] - W_R$ of bandwidth at most $\beta(k, \tau-1)$.
If any of these recursive calls returns a subtree of bandwidth at least $k$, the algorithm returns that subtree and halts. Suppose now that all recursive calls return layouts of bandwidth at most $\beta(k, \tau-1)$.

For each $0 \leq i \leq \ell$, apply Lemma~\ref{lem:deleteSet} to the layout of $G_i - (Q \cup W_i)$ together with the vertex set $(Q \cap V(G_i)) \cup W_i$ (of size at most $8d + \rho \leq 800 d^3$, using $d \geq 2$) to obtain a layout $\lambda_i$ of the part $G_i$ of bandwidth at most $\beta'$, where $\beta' = 12 \cdot \beta(k, \tau-1) \cdot 800 d^3 \cdot 2k = 19200 \cdot d^3 k \cdot \beta(k, \tau-1)$. Similarly, apply Lemma~\ref{lem:deleteSet} to the layout of $G[R] - W_R$ together with $W_R$ (of size at most $\rho$) to obtain a layout of $G[R]$ of bandwidth at most $12 \cdot \beta(k, \tau-1) \cdot \rho \cdot 2k \leq 19008 \cdot d^3 k \cdot \beta(k, \tau-1) \leq \beta'$.

Now $(\ell, P, g, \zeta, W, \{\lambda_i : 0 \leq i \leq \ell\})$ is a pre-layout of $G - (S \cup R)$ with separator size at most $4d$ and bandwidth at most $\beta'$: by construction $(\emptyset, \emptyset, \ell, P, g, \zeta, W)$ is a weak pre-layout of $G - (S \cup R)$, and each $\lambda_i$ is a bandwidth-$\beta'$ layout of part $G_i$.

We apply Lemma~\ref{lem:prelayoutGraphs} to this pre-layout (with separator size $s \leq 4d$ and bandwidth $\beta'$) and obtain in polynomial time either a subtree $T$ of $G$ with bandwidth at least $k$ or a layout of $G - (S \cup R)$ of bandwidth at most
$$10^{54}(4d+2)^{11}(\beta'+1)^3 k^3 \leq 10^{76} \cdot d^{20} k^6 \cdot \beta(k, \tau-1)^3,$$
using $d \geq 2$ to bound $4d + 2 \leq 5d$. Since we also have a layout of $G[R]$ of bandwidth at most $\beta' \leq 10^{76} d^{20} k^6 \beta(k, \tau-1)^3$, and $N(R) \subseteq S$ by property~(\ref{itm:RNeigh}) of weak pre-layouts, concatenating these two layouts yields a layout of $G - S$ of the same bandwidth. Applying Lemma~\ref{lem:deleteSet} to this layout together with $S$ (with $|S| \leq 4d$ and max degree $\leq 2k$) yields a layout of $G$ of bandwidth at most
$$12 \cdot 4d \cdot 2k \cdot 10^{76} d^{20} k^6 \beta(k, \tau-1)^3 = 96 \cdot 10^{76} \cdot d^{21} k^7 \beta(k, \tau-1)^3 \leq 10^{85} \cdot k^{28} \beta(k, \tau-1)^3,$$
where the last inequality uses $d \leq 2k$. The inductive hypothesis $\beta(k, \tau-1) \leq (10^{85} \cdot k^{28})^{4^{\tau-1}}$ gives
$$(10^{85} \cdot k^{28}) \cdot \beta(k, \tau-1)^3 \leq (10^{85} \cdot k^{28})^{1 + 3 \cdot 4^{\tau-1}} \leq (10^{85} \cdot k^{28})^{4^\tau} = \beta(k, \tau).$$

To complete the proof we analyze the running time of the algorithm. The dominant cost at each node of the recursion tree is the call to Lemma~\ref{lem:weakPreLayoutDensity}, which runs in time $2^{O(9^\tau)} n^{O(1)}$; since $\tau \leq k$, the per-node time is at most $2^{O(9^k)} n^{O(1)}$. Each recursive invocation of the algorithm is on an instance with strictly smaller value of $\tau$, so the depth of the recursion tree is at most $\tau$. All recursive invocations of the algorithm are made on disjoint vertex sets, thus the number of leaves in the recursion tree is at most $n$. We conclude that the total running time of the algorithm is at most $2^{O(9^k)} n^{O(1)} \cdot n \cdot \tau \leq 2^{O(9^k)} \cdot n^{O(1)}$, where the $\tau$ factor is absorbed into $2^{O(9^k)}$ using $\tau \leq k$.
\end{proof}

We are now ready to prove Theorem~\ref{thm:main}. 

\begin{proof}[Proof of Theorem~\ref{thm:main}]
The algorithm proceeds as follows. First it determines, using the algorithm of Lemma~\ref{lem:detectPathwidthSubtree} with parameter $\tau = k+1$, in time $2^{O(9^k)} \cdot n^{O(1)}$, whether $G$ contains a subtree $T$ of pathwidth at least $k+1$. If it does, the algorithm returns $T$, which satisfies $\bw(T) \geq \pw(T) > k$. Otherwise it calls the algorithm of Lemma~\ref{lem:mainAlgorithmInduction} on $G$ with $\tau = k$, and obtains in time $2^{O(9^k)} \cdot n^{O(1)}$ either a subtree $T$ of $G$ of bandwidth at least $k$ or a layout of $G$ of bandwidth at most $(10^{85} \cdot k^{28})^{4^k}$.
\end{proof}

We conclude with the proof of Theorem~\ref{thm:mainInTermsOfObstructions}

\begin{proof}[Proof of Theorem~\ref{thm:mainInTermsOfObstructions}]
Let $G$ be a graph such that $G$ does not contain a subtree $T$ such that $\pw(T) \geq k+1$, or $\hat{\Delta}(T) \geq k+1$, or $T$ is a skewed $(k+1)$-Cantor comb of depth $k+1$.
By Theorem~\ref{thm:treeBandwidth} the bandwidth of every subtree of $G$ is at most $(5k)^{6k}$. Setting $K_0 := (5k)^{6k} + 1$, no subtree of $G$ has bandwidth at least $K_0$, so applying Theorem~\ref{thm:main} to $G$ with parameter $K_0$ must produce a layout, giving
$$\bw(G) \leq (10^{85} \cdot K_0^{28})^{4^{K_0}} \leq 4^{4^{(5k)^{7k}}}.$$
%The last inequality follows by taking $\log_4$ twice: for $k \geq 1$ we have $K_0 \leq 2(5k)^{6k}$, so
%$$\log_4 \log_4 \bw(G) \leq K_0 + \log_4 \log_4(10^{85} \cdot K_0^{28}) \leq 2 K_0 \leq 4(5k)^{6k} \leq (5k)^{7k},$$using $\log_4 \log_4(10^{85} K_0^{28}) \leq K_0$ (an extremely loose bound for $k \geq 1$) and $(5k)^{7k} = (5k)^k \cdot (5k)^{6k} \geq 5 (5k)^{6k}$.
\end{proof}

\section{Conclusion}\label{sec:conclusion}
\todo{the three things that are reasonable to improve are structure, running time, approx ratio,now we talk about 1,2,3}
In this work we obtain an algorithm that takes as input a graph $G$ and integer $k$, runs in time $2^{O(9^k)} \cdot n^{O(1)}$, and either outputs a subtree $T$ of $G$ such that $\bw(T) \geq k$ or produces a layout of $G$ with bandwidth at most $(10^{85} \cdot k^{28})^{4^k}$.
The approximation function $(10^{85} \cdot k^{28})^{4^k}$ of our algorithm grows quite quickly with $k$, and it would be interesting to see how much this bound can be improved. 
An intriguing related question left open by our work is to determine the slowest growing function $h$ such that every graph with no subtree of bandwidth at least $k$ has bandwidth at most $h(k)$. Our algorithm establishes that $h(k)$ is at most double exponential in $k$. At the same time, we are not even aware of {\em super-linear} lower bounds on $h(k)$. %\todo{Eran, Maria, Please check!!! E: I'm not aware either:)}

It is worth noting that the running time ($2^{c \cdot 9^k}n^{O(1)}$ for some constant $c$) of our algorithm can be improved to {\em polynomial} time, even when $k$ is part of the input, by using a simple trick, at the cost of a small additional factor in the approximation bound and at the cost of not outputting the witness subtree $T$.
If $n > 2^{c \cdot 9^k}$ the running time of the algorithm of Theorem~\ref{thm:main} is already upper bounded by $2^{c \cdot 9^k} \cdot n^{O(1)} = n^{O(1)}$, so we may assume $n \leq 2^{c \cdot 9^k}$.
In this regime, applying the polynomial-time $O(\log^3 n \log\log n)$-approximation algorithm for {\sc Bandwidth} in general graphs of Dunagan and Vempala~\cite{DunaganV01} to $G$ outputs a layout of bandwidth at most $O(\log^3 n \log\log n) \cdot \bw(G) \leq O(c^4 \cdot 9^{4k}) \cdot \bw(G)$. Combined with the structural guarantee of Theorem~\ref{thm:main} that $\bw(G) \leq (10^{85} \cdot k^{28})^{4^k}$ (whenever no witness subtree exists), this yields a layout of bandwidth at most $O(c^4 \cdot 9^{4k}) \cdot (10^{85} \cdot k^{28})^{4^k}$, only a low-order factor worse than the original bound. 

%For the running time, the only reason our algorithm does not run in polynomial time is the algorithm of Lemma~\ref{lem:erdosPosa}.
%to compute a set $S$ of size at most ... such that $G - S$ does not have any subtree of pathwidth at least $\tau$. 
%
%All other super-polynomial time subroutines can be replaced by appropriate polynomial time approximation algorithms (at the cost of a slightly worse bound on the bandwidth of the produced layout of $G$). 
%
%The algorithm of Lemma~\ref{lem:erdosPosa} invokes (a variant of) Courcelle's Theorem~\cite{borie1992automatic,courcelle1990monadic}, which is the reason why the running time dependence on $k$ of our algorithm is some unspecified function. 
%
%\todo[inline]{removing the paragraph below from STOC version because arxiv version should fix this anyway}
%\todo[inline]
%{It is quite plausible that there exists an appropriate polynomial time (or perhaps a $f(k)n^{O(1)}$ time for a ``reasonable'' function $f$) approximation version of Lemma~\ref{lem:erdosPosa} that can be plugged into our arguments. 
%
%More concretely we conjecture that there exists a {\em polynomial} time version of Theorem~\ref{thm:main} where the bandwidth of the produced layout is still $2^{2^{O(k)}}$.}

Finally, one should also consider parameterized approximation algorithms for the {\sc Bandwidth} problem untethered from the bandwidth of a subtree of $G$. 
Theorem~\ref{thm:main} shows that the  {\sc Bandwidth} problem is FPT-approximable. While no constant factor approximation is possible in {\em polynomial} time~\cite{dubey2011hardness} unless \textsf{P} $=$ \textsf{NP}, there might exist an algorithm that takes as input a graph $G$ and integer $k$, runs in time  $f(k)n^{O(1)}$, and either concludes that $\bw(G) \geq k$ or produces a layout of $G$ of bandwidth at most $(\log k)^{O(1)}$, or even $O(k)$.

\paragraph{AI use statement.}
The authors used Claude (Anthropic) during the revision stage of this paper. Claude was used for line-by-line audits of lemma statements and proofs (identifying typos, off-by-one errors, calculational slips, and minor logical inconsistencies), for suggesting tighter constants and cleaner phrasings, and for drafting some proof passages from author-provided sketches and structural outlines. All original results, proof strategies, and mathematical content predate the use of AI assistance. All AI-suggested edits and drafted passages were reviewed and approved by the authors, who bear sole responsibility for the content.

%%%%%%%%%%%%%%%%%%%%%%%%%%%%%%%%%%%%%%%%%%%%%%%%%%%%%%

\bibliographystyle{alpha}
\bibliography{fptapproxbandwidth}

\end{document}